\documentclass[english,english,aps,reprint,superscriptaddress,amssymb,amsfonts,nofootinbib]{revtex4-1}
\usepackage[T1]{fontenc}
\usepackage[latin9]{inputenc}
\usepackage{xcolor}
\usepackage{verbatim}
\usepackage{mathrsfs}
\usepackage{mathtools}
\usepackage{amsmath}
\usepackage{amsthm}
\usepackage{amssymb}
\usepackage{graphicx}
\usepackage{comment}

\makeatletter

\providecommand{\tabularnewline}{\\}

\theoremstyle{plain}
\ifx\thechapter\undefined
\newtheorem{prop}{\protect\propositionname}
\else
\newtheorem{prop}{\protect\propositionname}[chapter]
\fi
\theoremstyle{remark}
\ifx\thechapter\undefined
\newtheorem{rem}{\protect\remarkname}
\else
\newtheorem{rem}{\protect\remarkname}[chapter]
\fi

\usepackage{color}
\usepackage{tikz}
\usepackage{tikz-cd}
\usepackage{comment}
\usepackage{amsfonts}
\usepackage{amssymb}
\usepackage{amsmath}
\usepackage{amsthm}
\usepackage{mathrsfs}
\usepackage{bm}
\usepackage{graphicx}
\usepackage{enumerate}
\usepackage{multirow}
\usepackage{hyperref}
\usepackage{float}
\usepackage{framed}
\usepackage{adjustbox}
\usepackage{appendix}
\usepackage{lipsum}

\usepackage[normalem]{ulem}

\hypersetup{colorlinks=true}

\newtheoremstyle{definitionstyle}% name
{0.2cm}% Space above
{0.2cm}% Space below
{\it}% Body font
{}% Indent amount (empty = no indent, \parindent = para indent)
{\it\bfseries}% Thm head font
{.}% Punctuation after thm head
{ }% Space after thm head: " " = normal interword space; \newline = linebreak
{\thmname{#1}\thmnumber{~#2}\thmnote{~(#3)}}% Thm head spec

\newtheoremstyle{nameddefinitionstyle}% name
{\baselineskip\@plus.2\baselineskip\@minus.2\baselineskip}% Space above
{\baselineskip\@plus.2\baselineskip\@minus.2\baselineskip}% Space belowqi2008A
{}% Body font
{}% Indent amount (empty = no indent, \parindent = para indent)
{\bfseries}% Thm head font
{.}% Punctuation after thm head
{ }% Space after thm head: " " = normal interword space; \newline = linebreak
{\thmnote{#3}}% Thm head spec

\newtheoremstyle{framednameddefinitionstyle}% name
{0.2cm}% Space above
{0.2cm}% Space below
{\it}% Body font
{}% Indent amount (empty = no indent, \parindent = para indent)
{\it\bfseries}% Thm head font
{.}% Punctuation after thm head
{ }% Space after thm head: " " = normal interword space; \newline = linebreak
{\thmnote{#3}}% Thm head spec

\newtheoremstyle{theoremstyle}% name
{0.2cm}% Space above
{0.2cm}% Space below
{\it}% Body font
{}% Indent amount (empty = no indent, \parindent = para indent)
{\it\bfseries}% Thm head font
{.}% Punctuation after thm head
{ }% Space after thm head: " " = normal interword space; \newline = linebreak
{\thmname{#1}\thmnumber{~#2}\thmnote{~(#3)}}% Thm head spec

\newtheoremstyle{framedtheoremstyle}% name
{\baselineskip\@plus.2\baselineskip\@minus.2\baselineskip}% Space above
{\baselineskip\@plus.2\baselineskip\@minus.2\baselineskip}% Space below
{\sl}% Body font
{}% Indent amount (empty = no indent, \parindent = para indent)
{\bfseries}% Thm head font
{.}% Punctuation after thm head
{ }% Space after thm head: " " = normal interword space; \newline = linebreak
{\thmname{#1}\thmnumber{~#2}\thmnote{~(#3)}}% Thm head spec

\newtheoremstyle{proofstyle}% name
{\baselineskip\@plus.2\baselineskip\@minus.2\baselineskip}% Space above
{\baselineskip\@plus.2\baselineskip\@minus.2\baselineskip}% Space below
{}% Body font
{}% Indent amount (empty = no indent, \parindent = para indent)
{}% Thm head font
{.}% Punctuation after thm head
{ }% Space after thm head: " " = normal interword space; \newline = linebreak
{\textsc{\thmname{#1}\thmnote{~#3}}}% Thm head spec

\theoremstyle{theoremstyle}

\newtheorem{lem}{Lemma}

\newtheorem{thm}{Theorem}

\theoremstyle{framedtheoremstyle}

\theoremstyle{definitionstyle}
\theoremstyle{definitionstyle}

\theoremstyle{definitionstyle}
\theoremstyle{definitionstyle}

\theoremstyle{nameddefinitionstyle}

\theoremstyle{framednameddefinitionstyle}
\newtheorem*{framednameddef}{}

\theoremstyle{proofstyle}

\theoremstyle{definitionstyle}

\newcommand{\fromto}{\rightarrow}

\newcommand{\ZZZ}{\mathbb{Z}}
\newcommand{\NNN}{\mathbb{N}}

\newcommand{\RRR}{\mathbb{R}}
\newcommand{\CCC}{\mathbb{C}}

\newcommand{\F}{\mathcal{F}}

\DeclareMathOperator{\Aut}{Aut}

\DeclareMathOperator{\kernel}{ker}

\newcommand{\coloneq}{\mathrel{\mathop:}=}

\newcommand{\homotopic}{\simeq}

\newcommand{\isomorphic}{\cong}

\newcommand{\paren}[1]{\left( #1 \right)}

\newcommand{\brackets}[1]{\left[ #1 \right]}
\newcommand{\braces}[1]{\left\{ #1 \right\}}

\def\ind\hspace{0.2in}

\newcommand{\charles}[1]{{\color{red}[Charles: #1]}}

\newcommand{\hsong}[1]{{\color{green}[HaoSong: #1]}}

\makeatother

\usepackage{babel}
\providecommand{\propositionname}{Proposition}
\providecommand{\remarkname}{Remark}

\begin{document}

\title{Bosonic Crystalline Symmetry Protected Topological Phases Beyond the Group Cohomology Proposal}
\author{Hao Song}\thanks{These authors contributed equally to this work.}
\affiliation{Departamento de F\'isica Te\'orica, Universidad Complutense, 28040 Madrid, Spain}
\author{Charles Zhaoxi Xiong}\thanks{These authors contributed equally to this work.}
\affiliation{Department of Physics, Harvard University, Cambridge, Massachusetts 02138, USA}
\author{Sheng-Jie Huang}
\affiliation{Department of Physics and Center for Theory of Quantum Matter, University of Colorado, Boulder, Colorado 80309, USA}
\date{\today}

\begin{abstract}
It is demonstrated by explicit construction that three-dimensional bosonic crystalline symmetry protected topological (cSPT) phases are classified by $H_{\phi}^{5}(G;\mathbb{Z})\oplus H_{\phi}^{1}(G;\mathbb{Z})$ for all 230 space groups $G$, where $H^n_\phi(G;\ZZZ)$ denotes 
the $n$th twisted group cohomology of $G$ with $\ZZZ$ coefficients, and $\phi$ indicates that
$g\in G$ acts non-trivially on coefficients by sending them to their inverses if $g$ reverses spacetime orientation and acts trivially otherwise.
The previously known summand $H_{\phi}^{5}(G;\mathbb{Z})$ corresponds only to crystalline phases built without the $E_8$ state or its multiples on 2-cells of space. It is the crystalline analogue of the ``group cohomology proposal'' for classifying bosonic symmetry protected topological (SPT) phases, which takes the form $H_{\phi}^{d+2}(G;\mathbb{Z})\cong H_{\phi}^{d+1}(G;U(1))$ for finite internal symmetry groups in $d$ spatial dimensions. 
The new summand $H_{\phi}^{1}(G;\mathbb{Z})$ classifies possible configurations of $E_8$ states on 2-cells that can be used to build crystalline phases beyond the group cohomology proposal.
The completeness of our classification and the physical meaning of $H_{\phi}^{1}(G;\mathbb{Z})$ are established through a combination of dimensional reduction, surface topological order, and explicit cellular construction. The value of $H_{\phi}^{1}(G;\mathbb{Z})$ can be easily read off from the international symbol for $G$. Our classification agrees with the prediction of the ``generalized cohomology hypothesis,'' which concerns the general structure of the classification of SPT phases, and therefore provides strong evidence for the validity of the said hypothesis in the realm of crystalline symmetries.
\end{abstract}

\maketitle

\section{Introduction}

As the interacting generalization of topological insulators and superconductors \cite{kane2005A, kane2005B, Qi_Hughes_Zhang, Schnyder, Kitaev_TI, Hasan_Kane, Qi_Zhang, Chiu_Teo_Schnyder_Ryu}, symmetry protected topological (SPT) phases \cite{SPT_origin} have garnered considerable interest in the past decade \citep{pollmann2010, Wen_1d, Cirac, Wen_sgSPT_1d, Wen_2d, Wen_Fermion, levin2012, 2dChiralBosonicSPT, chen2013cohomology, Kapustin_Boson, Kapustin_Fermion, Kapustin_equivariant, Freed_SRE_iTQFT, Freed_ReflectionPositivity, Kitaev_Stony_Brook_2011_SRE_1, Kitaev_Stony_Brook_2011_SRE_2, Kitaev_Stony_Brook_2013_SRE, Kitaev_IPAM, Husain, 2dChiralBosonicSPT, 2dChiralBosonicSPT_erratum, 3dBTScVishwanathSenthil, 3dBTScWangSenthil, 3dBTScBurnell, 3dFTScWangSenthil_1, 3dFTScWangSenthil_2, 3dFTScWangSenthil_2_erratum, 2dFermionGExtension, Lan_Kong_Wen_1, Lan_Kong_Wen_2, SOinfty, Else_edge, Jiang_sgSPT, Thorngren_sgSPT, Wang_Levin_invariants, Wang_intrinsic_fermionic, Huang_dimensional_reduction, Lu_sgSPT}. The early studies of SPT phases focused on phases with \emph{internal} symmetries (\emph{i.e.}, symmetries that do not change the position of local degrees of freedom, such as Ising symmetry, $U(1)$ symmetry, and time reversal symmetry). Now it is slowly being recognized \citep{Kitaev_Stony_Brook_2011_SRE_1,Kitaev_Stony_Brook_2013_SRE, Kitaev_IPAM, Xiong, Gaiotto_Johnson-Freyd} that the classification of internal SPT phases naturally satisfies certain axioms which happen to define a well-known structure in mathematics called generalized cohomology \citep{Hatcher,DavisKirk,Adams1,Adams2}. In particular, different existing proposals for the classification of internal SPT phases are simply different examples of generalized cohomology theories.

Ref.\,\citep{Xiong} distilled the above observations regarding the general structure of the classification of SPT phases into a ``generalized cohomology hypothesis.'' It maintained that (a) there exists a generalized cohomology theory $h$ that correctly classifies internal SPT phases in all dimensions for all symmetry groups, and that (b) even though we may not know exactly what $h$ is, meaningful physical results can still be derived from the fact that $h$ is a generalized cohomology theory alone. Indeed, it can be shown, on the basis of the generalized cohomology hypothesis, that three-dimensional bosonic SPT phases with internal symmetry $G$ are classified by $H_{\phi}^{5}(G;\mathbb{Z})\oplus H_{\phi}^{1}(G;\mathbb{Z})$ \cite{Xiong, Gaiotto_Johnson-Freyd},\footnote{The direct sum of two abelian groups is the same as their direct product, but the direct sum notation $\oplus$ is more common for abelian groups in the mathematical literature.} where $H^n_\phi(G;\ZZZ)$ denotes the $n$th twisted group cohomology of $G$ with $\ZZZ$ coefficients, and $\phi$ emphasizes that
$g\in G$ acts non-trivially on coefficients by sending them to their inverses if $g$ reverses spacetime orientation and acts trivial otherwise. The first summand, $H_{\phi}^{5}(G;\mathbb{Z})$,\footnote{For $n=0,1,2,\cdots$ and a compact group $G$, $H_{\phi}^{n+1}(G;\mathbb{Z})$ is isomorphic to the ``Borel group cohmology'' $H_{\text{Borel}, \phi}^{n}(G;U(1))$ considered in Ref.~\cite{Wen_Boson}, which is simply $H_{\phi}^{n}(G;U(1))$ if $G$ is finite.} corresponds to the ``group cohomology proposal'' for the classification of SPT phases \citep{chen2013cohomology}. The second summand, $H_{\phi}^{1}(G;\mathbb{Z})$, corresponds to phases beyond the group cohomology proposal, and are precisely the phases constructed in Ref.\,\cite{decorated_domain_walls} using decorated domain walls. Specifically, in Ref.\,\cite{decorated_domain_walls}, the domain walls were decorated with multiples of the $E_8$ state, which is a 2D bosonic state with quantized thermal Hall coefficient \citep{Kitaev_honeycomb,2dChiralBosonicSPT,2dChiralBosonicSPT_erratum,Kitaev_KITP}.

However, physical systems tend to crystallize. What is the classification of SPT phases if $G$ is a \emph{space-group} symmetry rather than internal symmetry? In the fermionic case, one can incorporate crystalline symmetries by imposing point-group actions on the Brillouin zone before activating interactions \cite{TCI_Fu}. This is obviously not applicable to bosonic systems due to the lack of a Brillouin zone. As a get-around, Refs. \citep{Hermele_torsor, Huang_dimensional_reduction} proposed to build bosonic crystalline SPT (cSPT) phases by focusing on high-symmetry points in the real space rather than momentum space. Concretely, on every high-symmetry line, plane, etc.\,of a $d$-dimensional space with space-group action by $G$, one can put an SPT phase of the appropriate dimensions with an internal symmetry equal to the stabilizer subgroup  of (any point on) the line, plane, etc. In particular, it was shown, for every element of $H_{\phi}^{5}(G;\mathbb{Z})$, that there is a 3D bosonic crystalline SPT phase with space group symmetry $G$ that one can construct. Curiously, the same mathematical object, $H_{\phi}^{5}(G;\mathbb{Z})$, is also the group cohomology proposal for the classification of 3D bosonic SPT phases with \emph{internal} symmetry $G$. A heuristic insight into this apparent correspondence between crystalline and internal SPT phases was provided in Ref.\,\citep{gauging2018}, which drew an analogy between internal gauge fields and a certain notation of crystalline gauge fields. The correspondence was referred to therein as the ``crystalline equivalence principle.''

Just like for internal symmetries, the group cohomology proposal $H_{\phi}^{5}(G;\mathbb{Z})$ does not give the complete classification for crystalline symmetries either. In fact, in the block-state construction of Ref.\,\citep{Huang_dimensional_reduction}, $E_{8}$ state was excluded from being used as a building block for simplicity. Appealing to the crystalline equivalence principle, one might guess that the complete classification of 3D bosonic crystalline SPT phases with space group symmetry $G$ would be $H_{\phi}^{5}(G;\mathbb{Z})\oplus H_{\phi}^{1}(G;\mathbb{Z})$, since that is the classification when $G$ is internal. The recent work \citep{Shiozaki2018} gives us added confidence in this conjecture. In that work, the authors justified the extension of generalized cohomology theory to crystalline symmetries by systematically interpreting terms of a spectral sequence of a generalized cohomology theory as building blocks of crystalline phases. Related discussions along this direction can also be found in Ref.~\citep{Else2018, Song2018}. (To be precise, Ref.~\citep{Shiozaki2018} focused on generalized \emph{homology} theories, but it is highly plausible that that is equivalent to a generalized cohomology formulation via a Poincar\'e duality.)

In this paper, we will conduct a thorough investigation into 3D bosonic cSPT phases protected by any space group symmetry $G$, dubbed $G$-SPT phases for short, and establish that their classification is indeed given by 
\begin{equation}
H_{\phi}^{5}(G;\mathbb{Z})\oplus H_{\phi}^{1}(G;\mathbb{Z}). \label{eq:classification}
\end{equation}
We will see that distinct embedding copies of the $E_8$ state in the Euclidean space $\mathbb E^3$ produce $G$-SPT phases with different $H_{\phi}^{1}(G;\mathbb{Z})$ labels. To obtain the classification of $G$-SPT phases and to understand its physical meaning, we will invoke three techniques: dimensional reduction, surface topological order, and explicit cellular construction. Technically, we will establish
%The combination of these techniques will establish that
\begin{enumerate}[(a)]
\item every 3D bosonic $G$-SPT phase can be mapped, via a homomorphism, to an element of $H_{\phi}^{1}(G;\mathbb{Z})$,
\item every element of $H_{\phi}^{1}(G;\mathbb{Z})$ can be mapped, via a homomorphism, to a 3D bosonic $G$-SPT phase,
\item the first map is a left inverse of the second, and
\item a 3D bosonic $G$-SPT phase comes from $H_{\phi}^{5}(G;\mathbb{Z})$ if and only if it maps to the trivial element of $H_{\phi}^{1}(G;\mathbb{Z})$.
\end{enumerate}
Therefore, with minor caveats such as the correctness of $H_{\phi}^{5}(G;\mathbb{Z})$ in classifying non-$E_8$-based phases and certain assumptions about the correlation length of short-range entangled (SRE) states, we will be providing an essentially rigorous proof that 3D bosonic cSPT phases are classified by $H_{\phi}^{5}(G;\mathbb{Z})\oplus H_{\phi}^{1}(G;\mathbb{Z})$. 
%\sj{Adjust the following discussion to include the results for magnetic space groups} 
In addition, we will show that the value of $H_{\phi}^{1}(G;\mathbb{Z})$ can be easily read off from the international symbol for $G$, following this formula:
\begin{eqnarray}
H_{\phi}^{1}(G;\ZZZ)=\begin{cases}
\ZZZ^{k}, & \mbox{\ensuremath{G} preserves orientation},\\
\ZZZ^{k}\times\ZZZ_{2}, & \mbox{otherwise},
\end{cases}\label{formula!}
\end{eqnarray}
where $k=0$ if there is more than one symmetry direction listed in the international symbol, $k=3$ if the international symbol has one symmetry direction listed and it is $1$ or $\overline{1}$, and $k=1$ if the international symbol has one symmetry direction listed and it is not $1$ or $\overline{1}$.

Naturally, we expect that $H_{\phi}^{5}(G;\mathbb{Z})\oplus H_{\phi}^{1}(G;\mathbb{Z})$ also works for more general crystalline symmetries with $H_{\phi}^{1}(G;\mathbb{Z})$ can be easily computed in a similar way. In particular, for magnetic space groups $G$ (which include space groups as a subclass called type I), we still have $H_{\phi}^{1}(G;\mathbb{Z})=\ZZZ^{k}\times \ZZZ_{2}^{\ell}$ with $k\in\mathbb{Z}$ and $\ell\in\{0,1\}$ only depending on the associated magnetic point group. For magnetic group $G$ of type II or type IV, $H_{\phi}^{1}(G;\mathbb{Z})$ is simply $\mathbb{Z}_{2}$, while $k$ and $\ell$ can be read off from the associated magnetic point group  type according to Table~\ref{tab:mg} for $G$ of type III.

This paper is organized as follows. In Sec.\,\ref{sec:prediction}, we will explain how $H_{\phi}^{5}(G;\mathbb{Z})\oplus H_{\phi}^{1}(G;\mathbb{Z})$ arises from the generalized cohomology hypothesis. In Sec.~\ref{sec:examples}, we will present examples of cSPT phases described by $H_{\phi}^{1}(G;\mathbb{Z})$ for select space groups. In Sec.~\ref{sec:reduction_and_construction}, we will establish that 3D bosonic cSPT phases are classified by $H_{\phi}^{5}(G;\mathbb{Z})\oplus H_{\phi}^{1}(G;\mathbb{Z})$ in full generality. In Sec.~\ref{sec:discussions}, we will discuss possible generalizations and conclude the paper. There are three appendices to the paper. In Appendix~\ref{sec:comp_H1}, we will prove formula (\ref{formula!}) and tabulate $H_{\phi}^{1}(G;\ZZZ)$ for all 230 space groups. Moreover, we generalize the calculation to magnetic space groups as well. In Appendix~\ref{sec:stacking}, we review the stacking operation on SPT phases. In Appendix~\ref{sec:GCH}, we review the generalized cohomology hypothesis.

Throughout this paper, we will use the term SPT phases with symmetry $G$ to mean all invertible topological phases with that symmetry \cite{Kitaev_Stony_Brook_2011_SRE_2, Kitaev_Stony_Brook_2013_SRE, Kapustin_Boson, Freed_SRE_iTQFT, Freed_ReflectionPositivity, McGreevy_sSourcery, Xiong, Xiong_Alexandradinata}. Therefore, an SPT phase may or may not be trivializable by breaking the symmetry, in contrast to an SPT phase in the traditional sense \cite{SPT_origin}.

\section{Prediction by generalized cohomology hypothesis\label{sec:prediction}}

To pave the way for us to generally construct and completely classify
3D $E_{8}$-based cSPT phases, let us make an prediction
for what the classification of these phases might be using the generalized cohomology hypothesis \cite{Xiong, Xiong_Alexandradinata}. The generalized cohomology hypothesis was based on Kitaev's proposal \citep{Kitaev_Stony_Brook_2011_SRE_1,Kitaev_Stony_Brook_2013_SRE,Kitaev_IPAM}
that the classification of SPT phases must carry the structure of
generalized cohomology theories \citep{Hatcher,DavisKirk,Adams1,Adams2}.
This proposal was further developed in Refs.\citep{Xiong, Xiong_Alexandradinata, Gaiotto_Johnson-Freyd}.

The key idea here is that the classification of SPT phases
can be encoded by a sequence $F_{\bullet}=\{F_{d}\}$ of topological spaces, 
\begin{equation}
F_{0},F_{1},F_{2},F_{3},F_{4},\ldots
\end{equation}
where $F_{d}$ is the space made up of all $d$-dimensional short-range
entangled (SRE) states. It can be argued \cite{Kitaev_Stony_Brook_2011_SRE_1, Kitaev_Stony_Brook_2013_SRE, Kitaev_IPAM, Xiong, Gaiotto_Johnson-Freyd} that the spaces $F_{d}$
are related to each other: the $d$-th space is homotopy equivalent
to the loop space \citep{Hatcher} of the $(d+1)$st space, 
\begin{equation}
F_{d}\homotopic\Omega F_{d+1}.\label{loop_relation}
\end{equation}
Physically, this says that there is a correspondence between $d$-dimensional
SRE states and one-parameter families of $(d+1)$-dimensional SRE
states.

To state how the sequence $F_{\bullet}$ determines the classification of SPT
phases, let us introduce a homomorphism, 
\begin{equation}
\phi:G\fromto\braces{\pm1},\label{phi_hom}
\end{equation}
that tracks which elements of the symmetry group $G$ preserve the
orientation of spacetime (mapped to $+1$) and which elements do not
(mapped to $-1$). As discussed in Appendix~\ref{sec:stacking}, the collections of $G$-SPT phases (\emph{i.e.}, topological phase protected by symmetry $G$) in $d$ dimensions are classified by the Abelian group of $G$-SPT orders
\begin{equation}
\text{SPT}^d(G),
\end{equation}
whose addition operation is defined by stacking. It conjectured that the group structure of $\text{SPT}^d(G)$ can be obtained by computing the mathematical object
\begin{equation}
h_{\phi}^d(G; F_{\bullet})\coloneq\brackets{EG,\Omega F_{d+1}}_{G}.\label{hypothesis_expression}
\end{equation}
Here, $EG$ is the total space of the universal principal $G$-bundle
\citep{AdemMilgram}, and $\brackets{EG,\Omega F_{d+1}}_{G}$ denotes
the set of deformation classes of $G$-equivariant maps from the space
$EG$ to the space $\Omega F_{d+1}$ (see Appendix\,\ref{sec:GCH} for
detail). Explicitly, the generalized cohomology hypothesis states that we have an isomorphism
\begin{equation}
\text{SPT}^d(G) \isomorphic h_{\phi}^d(G; F_{\bullet}). \label{iso_2}
\end{equation}

To compute (\ref{hypothesis_expression}), we note by definition that
the 0th homotopy group of $F_{d}$, 
\begin{equation}
\pi_{0}(F_{d}),
\end{equation}
(\emph{i.e.}, the set of connected components of $F_{d}$) classifies $d$-dimensional
invertible topological orders (\emph{i.e.}, SPT phases without symmetry).
In 0, 1, 2, and 3 dimensions, the classification of invertible topological
orders is believed to be \citep{Kitaev_Stony_Brook_2011_SRE_1,Kitaev_Stony_Brook_2013_SRE,Kitaev_IPAM, Xiong, Gaiotto_Johnson-Freyd}
\begin{equation}
\pi_{0}(F_{0})=0,~\pi_{0}(F_{1})=0,~\pi_{0}(F_{2})=\ZZZ,~\pi_{0}(F_{3})=0,\label{classification_iTO}
\end{equation}
respectively, where the $\ZZZ$ in 2 dimensions is generated by the
$E_{8}$ phase \citep{Kitaev_honeycomb,2dChiralBosonicSPT,2dChiralBosonicSPT_erratum,Kitaev_KITP}.
Next, we note that 0-dimensional SRE states are nothing but rays in
Hilbert spaces. These rays form the infinite-dimensional complex projective
space \citep{Hatcher}, so 
\begin{equation}
F_{0}=\CCC P^{\infty}.\label{F0}
\end{equation}
Finally, we note, as a consequence of Eq.\,(\ref{loop_relation}),
that the $(k+1)$st homotopy group of $F_{d+1}$ is the same as the
$k$-th homotopy group of $F_{d}$ for all $k$ and $d$: 
\begin{equation}
\pi_{k}(F_{d})\isomorphic\pi_{k+1}(F_{d+1}).
\end{equation}
This allows us to determine all homotopy groups of $F_{1},F_{2},F_{3}$
from the classification of invertible topological orders (\ref{classification_iTO})
and the known space of 0D SRE states (\ref{F0}). The results are
shown in Table \ref{table:homotopy_groups}.

\begin{table}[t]
\caption{Homotopy groups of the space $F_{d}$ of $d$-dimensional SRE states,
for $0\protect\leq d\protect\leq3$.}
\label{table:homotopy_groups} %
\begin{tabular}{c|cccc}
$\pi_{>5}$  & 0  & 0  & 0  & 0 \tabularnewline
$\pi_{5}$  & 0  & 0  & 0  & $\ZZZ$ \tabularnewline
$\pi_{4}$  & 0  & 0  & $\ZZZ$  & 0 \tabularnewline
$\pi_{3}$  & 0  & $\ZZZ$  & 0  & 0 \tabularnewline
$\pi_{2}$  & $\ZZZ$  & 0  & 0  & 0 \tabularnewline
$\pi_{1}$  & 0  & 0  & 0  & $\ZZZ$ \tabularnewline
$\pi_{0}$  & 0  & 0  & $\ZZZ$  & 0 \tabularnewline
\hline 
~  & $F_{0}$  & $F_{1}$  & $F_{2}$  & $F_{3}$ \tabularnewline
\end{tabular}
\end{table}
It turns out the homotopy groups in Table \ref{table:homotopy_groups}
completely determine the space $F_{1}$, $F_{2}$, $F_{3}$ themselves
\citep{Xiong}: 
\begin{eqnarray}
F_{1} & = & K(\ZZZ,3),\label{F1}\\
F_{2} & = & K(\ZZZ,4)\times\ZZZ,\label{F2}\\
F_{3} & = & K(\ZZZ,5)\times\mathbf{S}^{1},\label{F3}
\end{eqnarray}
where $K(\ZZZ,n)$ is the $n$-th Eilenberg-MacLane space of $\ZZZ$
{[}defined by the property $\pi_{k}\paren{K(\ZZZ,n)}=\ZZZ$ for $k=n$
and $0$ otherwise{]} \citep{Hatcher}. Plugging Eqs.\,(\ref{F0})(\ref{F1})(\ref{F2})(\ref{F3})
into Eqs.\,(\ref{hypothesis_expression})(\ref{iso_2}), we arrive at the prediction 
\begin{eqnarray}
\text{SPT}^0(G) & \isomorphic & H_{\phi}^{2}(G;\ZZZ),\\
\text{SPT}^1(G) & \isomorphic & H_{\phi}^{3}(G;\ZZZ),\\
\text{SPT}^2(G) & \isomorphic & H_{\phi}^{4}(G;\ZZZ)\oplus H_{\phi}^{0}(G;\ZZZ), \label{2D_prediction}\\ 
\text{SPT}^3(G) & \isomorphic & H_{\phi}^{5}(G;\ZZZ)\oplus H_{\phi}^{1}(G;\ZZZ),\label{3D_prediction}
\end{eqnarray}
where $H_{\phi}^{n}(G;\ZZZ)$ denotes the $n$-th twisted group cohomology of $G$ with coefficient $\ZZZ$ and twist $\phi$ \citep{AdemMilgram}.
For finite or compact groups, we have $H_{\phi}^{5}(G;\ZZZ)\isomorphic H_{{\rm Borel}, \phi}^{4}(G;U(1))$;
we identify this as the contribution from the group cohomology proposal
\citep{Wen_Boson} to the 3D classification (\ref{3D_prediction}).
The existence of $H_{\phi}^{1}(G;\ZZZ)$ in Eq.\,(\ref{3D_prediction}),
on the other hand, can be traced back to the fact that $\pi_{0}(F_{2})=\ZZZ$
in Eq.\,(\ref{classification_iTO}); we identify it as the contribution
of $E_{8}$-based phases \citep{Kitaev_honeycomb, 2dChiralBosonicSPT,2dChiralBosonicSPT_erratum, Kitaev_KITP}
to the 3D classification. Therefore, we predict that 3D bosonic cSPT
phases built from $E_{8}$ states are classified by 
\begin{equation}
H_{\phi}^{1}(G;\ZZZ)
\end{equation}
(up to some non-$E_8$-based phases),
where $G$ is the space group and $\phi:G\fromto\braces{\pm1}$ keeps
track of which elements of $G$ preserve/reverse the orientation [see
Eq.\,(\ref{phi_hom})]. 

$H_{\phi}^{1}(G;\ZZZ)$ can be intuitively thought of as the set of
``representations of $G$ in the integers.'' Explicitly,
an element of $H_{\phi}^{1}(G;\ZZZ)$ is represented by a map (called
a group 1-cocycle) 
\begin{equation}
\nu^{1}:G\fromto\ZZZ\label{cochain}
\end{equation}
satisfying the cocycle condition
\begin{equation}
\nu^{1}(g_{1}g_{2})=\nu^{1}(g_{1})+\phi(g_{1})\nu^{1}(g_{2})\label{cocycle_condition}
\end{equation}
for all $g_{1},g_{2}\in G$. Suppose $\phi$ was trivial for the moment
(mapping all elements to $+1$). Then we would have $\nu^{1}(g_{1}g_{2})=\nu^{1}(g_{1})+\nu^{1}(g_{2})$,
which is precisely the axiom $\rho(g_{1}g_{2})=\rho(g_{1})\rho(g_{2})$
for a representation $\rho$ of $G$, written additively as opposed
to multiplicatively. Now if we allowed $\rho$ to be antilinear, then
we would have the modified condition $\rho(g_{1}g_{2})=\rho(g_{1})\overline{\rho(g_{2})}$
(overline denoting complex conjugation) for all $g_{1}$'s that are
represented antilinearly. The analogue of this for $\nu^{1}$ is precisely
(\ref{cocycle_condition}). The cohomology group $H_{\phi}^{1}(G;\ZZZ)$ itself is defined to be the quotient of group 1-cocycles
{[}as in Eqs.\,(\ref{cochain})(\ref{cocycle_condition}){]} by what is called group
1-coboundaries. In Appendix\,\ref{sec:comp_H1}, we make the definition of group 1-coboundary explicit and show how one can easily read off $H_{\phi}^{1}(G;\ZZZ)$ from the international symbol of $G$.

\section{Examples of 3D crystalline phases\label{sec:examples}}
\begin{comment}
Hao

Examples of order-2 and integer phases (Space Groups No. 1 (P1), No.
)

Show nontriviality by STO

\hsong{Consider moving Sections II and III after Sec. IV.} 
\end{comment}

\begin{comment}
To demonstrate the cSPT phases we are dealing with in this paper,
let us first construct explicit examples. Each space groups will be
referred by its Hermann-Mauguin symbol with or without its sequential
number in the international table of for crystallography. For example,
the space group with only translation symmetries is referred to as
$P1\left(1\right)$ or simply $P\overline{1}$. Below, Secs.~2.1-2.3
deal with space groups generated by translations and possibly mirror
reflections; Sec. 2.5 deals with inversion; Sec. 2.6 deals with glide
reflection.
\end{comment}

Now let us focus on cSPT phases in three spatial dimensions (\emph{i.e.},
$d=3$). For each fixed space group $G$, they (or more precisely $G$-SPT orders) form a Abelian group $\text{SPT}^{3}(G)$ equipped with the stacking operation as explained in Appendix~\ref{sec:stacking}.
To examine whether $G$-SPT phases are classified by $h_{\phi}^{d}\left(G;F_{\bullet}\right)\cong H_{\phi}^{5}\left(G,\mathbb{Z}\right)\oplus H_{\phi}^{1}\left(G,\mathbb{Z}\right)$,
we notice that the first summand $H_{\phi}^{5}\left(G,\mathbb{Z}\right)$ has been already computed in Ref.~\citep{Thorngren_sgSPT}
and matches with the phases built by lower dimensional group cohomology
states explicitly constructed and classified in Ref.~\citep{Huang_dimensional_reduction}.
It was also noticed that distinct cSPT phases can be built with $E_{8}$
states \citep{Hermele_torsor}. In the following sections, we are going
to construct such phases systematically for all 230 space groups and
show that they are classified by the second summand $H_{\phi}^{1}\left(G,\mathbb{Z}\right)$.

To clarify the notion of cSPT phases and to motivate the general construction,
let us present explicit examples of cSPT phases for the space groups
$P1\left(1\right)$, $P\overline{1}\left(2\right)$, $Pm\left(6\right)$,
$Pc\left(7\right)$, $Pmm2\left(25\right)$, and $Pmmm\left(47\right)$.
Here, for the reader's convenience, the sequential number
of a space group (as in the \emph{International Tables for Crystallography}~\citep{ITA2006}) are provided
in the parentheses after its Hermann-Mauguin symbol.

\subsection{Space group $P1\left(1\right)$}

The space group $P1$ contains only translation symmetries. In
proper coordinates, $P1$ is generated by 
\begin{align}
t_{x} & :\left(x,y,z\right)\mapsto\left(x+1,y,z\right),\label{eq:tx}\\
t_{y} & :\left(x,y,z\right)\mapsto\left(x,y+1,z\right),\label{eq:ty}\\
t_{z} & :\left(x,y,z\right)\mapsto\left(x,y,z+1\right).\label{eq:tz}
\end{align}
As an abstract group, it is isomorphic to $\mathbb{Z}\times\mathbb{Z}\times\mathbb{Z}\equiv\mathbb{Z}^{3}$.
Thus, we have 
\begin{align}
H_{\phi}^{5}\left(P1,\mathbb{Z}\right) & =0,\\
H_{\phi}^{1}\left(P1,\mathbb{Z}\right) & =\mathbb{Z}^{3}.
\end{align}
The latter can be identified with layered $E_{8}$ states.

For instance, we can put a copy of $E_{8}$ state on each of the planes
$y=\cdots,-2,-1,0,1,2,\cdots$ as in Fig.~\ref{fig:layerE8_y}. Since
each $E_{8}$ state can be realized with translation symmetries $t_{x}$
and $t_{z}$ respected, the layered $E_{8}$ states can be made to
respect all three translations and thus realize a 3D cSPT phase with
$P1$ symmetry.

\begin{figure}
\includegraphics[width=0.75\columnwidth]{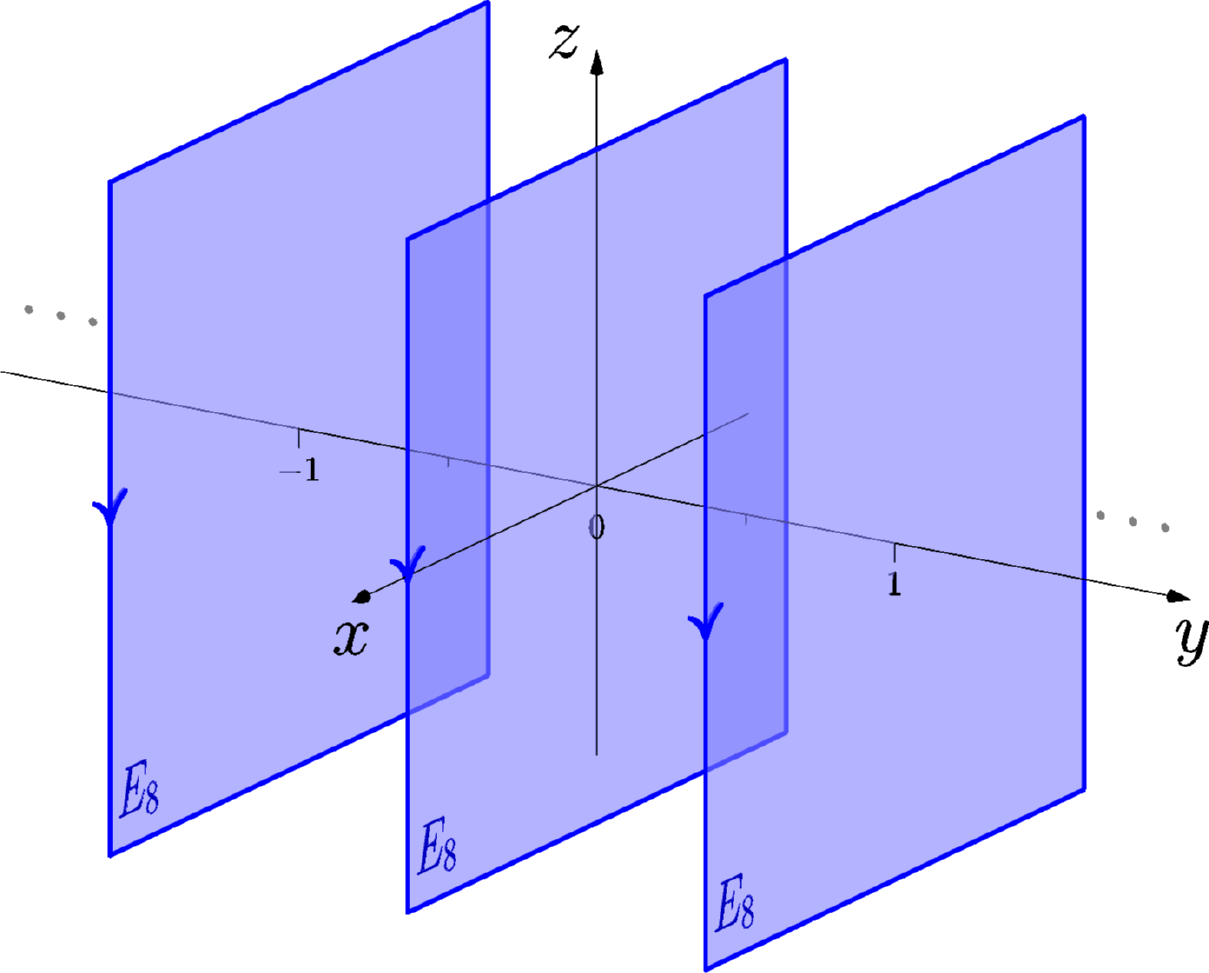}

\caption{Layered $E_{8}$ states in the $t_{y}$ direction.}

\label{fig:layerE8_y} 
\end{figure}
To characterize this cSPT phase, we take a periodic boundary condition
in $t_{y}$ direction requiring $t_{y}^{L_{y}}=1$. Such a procedure
is called a compactification in the $t_{y}$ direction, which is well-defined
for any integer $L_{y}\gg1$ in general for a gapped system with a
translation symmetry $t_{y}$. The resulting model has a finite thickness in the $t_{y}$ direction and thus can be viewed as a 2D system
extending in the $t_{x}$, $t_{z}$ directions. Further, neglecting
the translation symmetries $t_{x}$ and $t_{z}$, we take an open
boundary condition of the 2D system. Then its edge supports $8L_{y}$
co-propagating chiral boson modes (chiral central charge $c_{-}^{y}=8L_{y}$).
The resulting quantized thermal Hall effect is proportional to $L_{y}$
and shows the nontriviality of the layered $E_{8}$ states as a cSPT
phase with the $P1$ symmetry.

Actually, even without translation symmetries, we cannot trivialize
such a system into a tensor product state with local unitary gates
with a universal finite depth homogeneously in space. However, such a state is trivial in a weaker
sense: if the system has correlation length less than $\xi>0$, then
any ball region of size much larger than $\xi$ can be trivialized
with correlation length kept smaller than $\xi$. We call such a state
\emph{weakly trivial}.  In this paper, the notion of cSPT phase includes
all gapped quantum phases without emergent nontrivial
quasiparticles and particularly these weakly trivial states.

A bosonic SPT system has $c_{-}^{y}=\gamma_{y}L_{y}$ with $\gamma_{y}$
a multiple of $8$ (\emph{i.e.}, $\gamma_{y}\in8\mathbb{Z}$) in general. To
see $\gamma_{y}\in8\mathbb{Z}$, we notice that the net number of
chiral boson modes along the interface between compactified systems
of thicknesses $L_{y}$ and $L_{y}+1$ is $\gamma_{y}$. Then the
absence of anyons in both sides implies $\gamma_{y}\in8\mathbb{Z}$
\citep{KITAEV20062}. It is obvious that $\gamma_{y}$'s are added
during a stacking operation.

Clearly, this cSPT phase of layered $E_{8}$ states is invertible;
its inverse is made of layered $\overline{E_{8}}$ states, where $\overline{E_{8}}$
denotes the chiral twin of $E_{8}$. The edge modes of $\overline{E_{8}}$
propagate in the direction opposite to those of $E_{8}$ and hence we have
$\gamma_{y}=-8$ for the cSPT phase of the layered $\overline{E_{8}}$
states. Via the stacking operation explained in Appendix~\ref{sec:stacking},
all possible $\gamma_{y}\in8\mathbb{Z}$ is generated by the cSPT
phase of layered $E_{8}$ states and its inverse.

Analogously, we can define $\gamma_{x}$ (resp. $\gamma_{z}$) by
compactifying a system in the $t_{x}$ (resp. $t_{z}$) direction.
Thus, the cSPT phases with $P1$ symmetry are classified by $\frac{1}{8}\left(\gamma_{x},\gamma_{y},\gamma_{z}\right)\in\mathbb{Z}^{3}=H_{\phi}^{1}\left(P1,\mathbb{Z}\right)$.
Via the stacking operation, they can be generated by the three cSPT
phases made of layered $E_{8}$ states in the $t_{x}$, $t_{y}$ and
$t_{y}$ directions respectively and their inverses.

\subsection{Space groups $Pm\left(6\right)$, $Pmm2\left(25\right)$, and $Pmmm\left(47\right)$}

To explore possible cSPT phases in the present of reflection symmetries
(\emph{i.e.}, mirror planes), let us look at the space groups $Pm$,
$Pmm2$ and $Pmmm$ as examples.

\subsubsection{Space group $Pm\left(6\right)$\label{subsec:Pm}}

The space group $Pm$ is generated by $t_{x}$, $t_{y}$, $t_{y}$
and a reflection 
\begin{equation}
m_{y}:\left(x,y,z\right)\mapsto\left(x,-y,z\right).
\end{equation}
Thus, the mirror planes are $y=\cdots,-1,-\frac{1}{2},0,\frac{1}{2},1,\cdots$
(\emph{i.e.}, integer $y$ planes and half-integer $y$ planes). For
$Pm$, the group structure of $G$ and its subgroup $G_{0}$ of orientation
preserving symmetries are 
\begin{align}
G & =\mathbb{Z}\times\mathbb{Z}\times\left(\mathbb{Z}\rtimes\mathbb{Z}_{2}^{\phi}\right),\\
G_{0} & =\mathbb{Z}\times\mathbb{Z}\times\mathbb{Z},
\end{align}
where $\mathbb{Z}_{2}^{\phi}$ is the group generated by $m_{y}$
with the superscript $\phi$ emphasizing $\phi\left(m_{y}\right)=-1$.

Theorem~\ref{thm:H1_formula} in Appendix~\ref{sec:comp_H1} computes
the second part of $h_{\phi}^{3}(G,\phi)$ shown in Eq.~(\ref{3D_prediction});
the result is 
\begin{gather}
H_{\phi}^{1}\left(G,\mathbb{Z}\right)=\mathbb{Z}\times\mathbb{Z}_{2}.\label{eq:Pm_h1}
\end{gather}
The $\mathbb{Z}$ factor specifies $\frac{1}{8}\gamma_{y}$ of the
cSPT phases compatible with $Pm$ symmetry; the reflection symmetry
$m_{y}$ requires that $\gamma_{x}=\gamma_{z}=0$. Explicitly, putting
an $E_{8}$ state on each integer $y$ plane produces a phase with
$\gamma_{y}=1$.

\begin{figure}
\includegraphics[width=0.75\columnwidth]{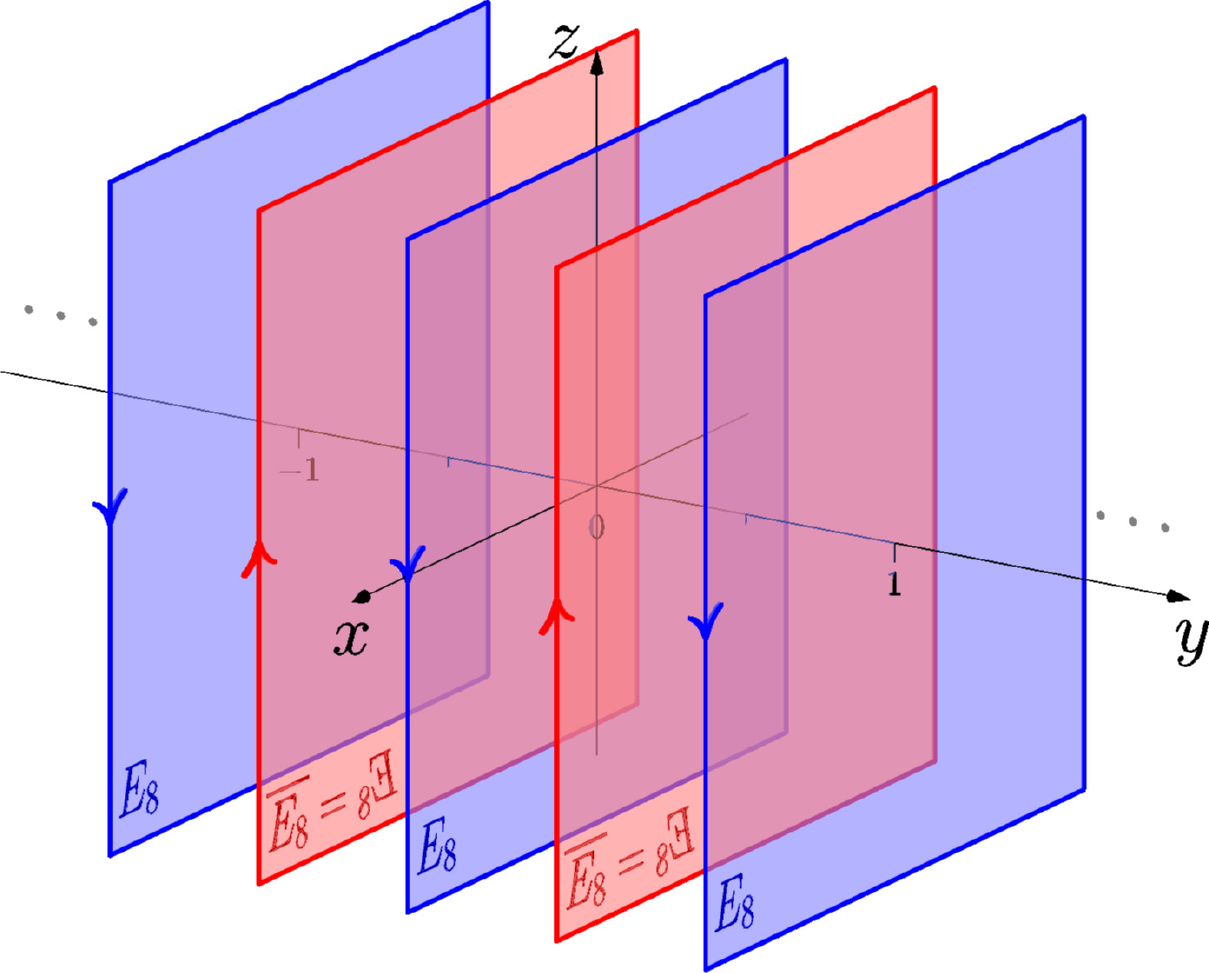}

\caption{Alternately layered $E_{8}$ states in the $t_{y}$ direction.}

\label{fig:layerE8_y2} 
\end{figure}
The factor $\mathbb{Z}_{2}$ in Eq.~(\ref{eq:Pm_h1}) is generated
by the cSPT phase built by putting an $E_{8}$ state at each integer
$y$ plane and its inverse $\overline{E_{8}}$ at each half-integer
$y$ plane, as shown in Fig.~\ref{fig:layerE8_y2}. This construction
was proposed and the resulting phases were studied in Ref.~\citep{Hermele_torsor}.
In particular, the order of this cSPT phase is $2$, meaning that
stacking two copies of such a phase produces a trivial phase. In fact,
we will soon see that this generating phase can be realized by a model
with higher symmetry like $Pmmm$ and is protected nontrivial by any
single orientation-reversing symmetry.

In addition, the other part of $h_{\phi}^{3}\left(G,\phi\right)$ is 
\begin{multline}
H_{\phi}^{5}\left(G;\mathbb{Z}\right)=H_{\phi}^{5}\left(\mathbb{Z}\rtimes\mathbb{Z}_{2}^{\phi};\mathbb{Z}\right)\oplus\left[H_{\phi}^{4}\left(\mathbb{Z}\rtimes\mathbb{Z}_{2}^{\phi};\mathbb{Z}\right)\right]^{2}\\
\oplus H_{\phi}^{3}\left(\mathbb{Z}\rtimes\mathbb{Z}_{2}^{\phi};\mathbb{Z}\right)=\mathbb{Z}_{2}^{2}\oplus0\oplus\mathbb{Z}_{2}^{2}.\label{eq:Pm_h5}
\end{multline}
The three summands correspond to the phases built from group cohomology
SPT phases in 2, 1, and 0 spatial dimensions respectively in Ref.~\citep{cSPT2017}.

For the reader's convenience, let us review the construction briefly
here. The four cSPT phases corresponding to first summand $\mathbb{Z}_{2}^{2}$
in Eq.~(\ref{eq:Pm_h5}) can be built by putting 2D SPT states with
Ising symmetry on mirror planes. We notice that each reflection acts
as an Ising symmetry (\emph{i.e.}, a unitary internal symmetry of order
2) on its mirror plane. Moreover, the space group $Pm$ contains two
families of inequivalent mirrors (\emph{i.e.}, integer $y$ planes
and half-integer $y$ planes). For each family, there are two choices
of 2D SPT phases with Ising symmetry, classified by $H^{4}\left(\mathbb{Z}_{2},\mathbb{Z}\right)\cong H^{3}\left(\mathbb{Z}_{2},U\left(1\right)\right)\cong\mathbb{Z}_{2}$.
Thus, we realize four cSPT phases with group structure $\mathbb{Z}_{2}^{2}$.
Further, noticing the translation symmetries within each mirror, we
can put, to every unit cell of the mirror, a 0D state carrying eigenvalue
$\pm1$ of the corresponding reflection; 0D states with Ising symmetry are classified by their symmetry charges and are formally labeled 
by $H^{2}\left(\mathbb{Z}_{2},\mathbb{Z}\right)\cong H^{1}\left(\mathbb{Z}_{2},U\left(1\right)\right)=\mathbb{Z}_{2}$.
They produce another $\mathbb{Z}_{2}$ factor of cSPT phases associated
with each inequivalent family of mirrors. Thus, we get another four
cSPT phases labeled by the last summand $\mathbb{Z}_{2}^{2}$ in Eq.~(\ref{eq:Pm_h5}).

\begin{figure}
\includegraphics[width=0.6\columnwidth]{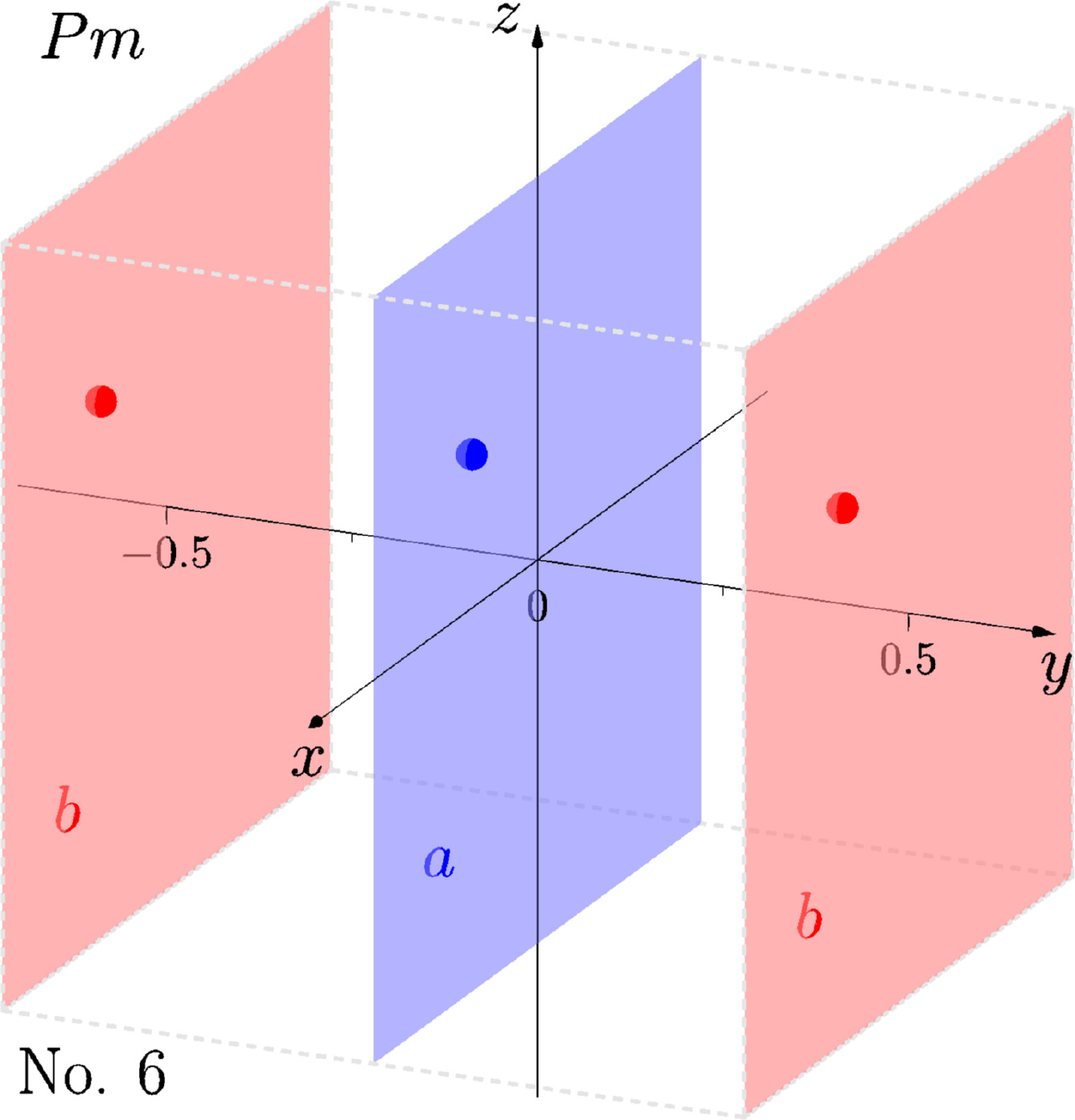}

\caption{For space group No.~6 (\emph{i.e.}, $G=Pm$), there are two inequivalent
family of mirrors (blue and red online), whose Wyckoff position labels
are $a$ and $b$ respectively in Ref.~\citep{ITA2006}. The cSPT
phases corresponding to $H_{\phi}^{5}\left(G,\mathbb{Z}\right)$ are
built by assigning to the mirrors 2D SPT phases protected by Ising
symmetry and $\mathbb{Z}_{2}$ point charges (blue and red dots online). }

\label{fig:Pm_SPT} 
\end{figure}
To summarize, this construction of $H_{\phi}^{5}\left(G,\mathbb{Z}\right)$
phases for the space group $Pm$ can be presented by Fig.~\ref{fig:Pm_SPT}.
It can be generalized to all the other space groups: to each Wyckoff
position, we assign group cohomology SPT phases protected by its site symmetry.
Further technical details can be found in Ref.~\citep{cSPT2017}.

\subsubsection{Space groups $Pmm2\left(25\right)$}

In the following, we will focus on developing a universal construction
of $H_{\phi}^{1}\left(G,\mathbb{Z}\right)$ phases. To get motivated,
let us look at more examples.

The space group $G=Pmm2$ is generated by $t_{x}$, $t_{y}$, $t_{y}$,
and two reflections 
\begin{align}
m_{x}: & \left(x,y,z\right)\mapsto\left(-x,y,z\right),\\
m_{y}: & \left(x,y,z\right)\mapsto\left(x,-y,z\right).
\end{align}
This time, Theorem~\ref{thm:H1_formula} in Appendix~\ref{sec:comp_H1}
tells us that 
\begin{gather}
H_{\phi}^{1}\left(G,\mathbb{Z}\right)=\mathbb{Z}_{2}.\label{eq:Pmm2_h1}
\end{gather}
There is no $\mathbb{Z}$ factor any more as expected, because $m_{x}$
and $m_{y}$ together require $\gamma_{x}=\gamma_{y}=\gamma_{z}=0$.

The cSPT phase with $Pmm2$ symmetry generating $H_{\phi}^{1}\left(G,\mathbb{Z}\right)$
is actually compatible with a higher symmetry $Pmmm$. Hence let us
combine the study on this phase with the discussion of the space group
$Pmmm$ below.

\subsubsection{Space group $Pmmm\left(47\right)$\label{subsec:cSPT_Pmmm}}

\begin{figure}
\includegraphics[width=0.6\columnwidth]{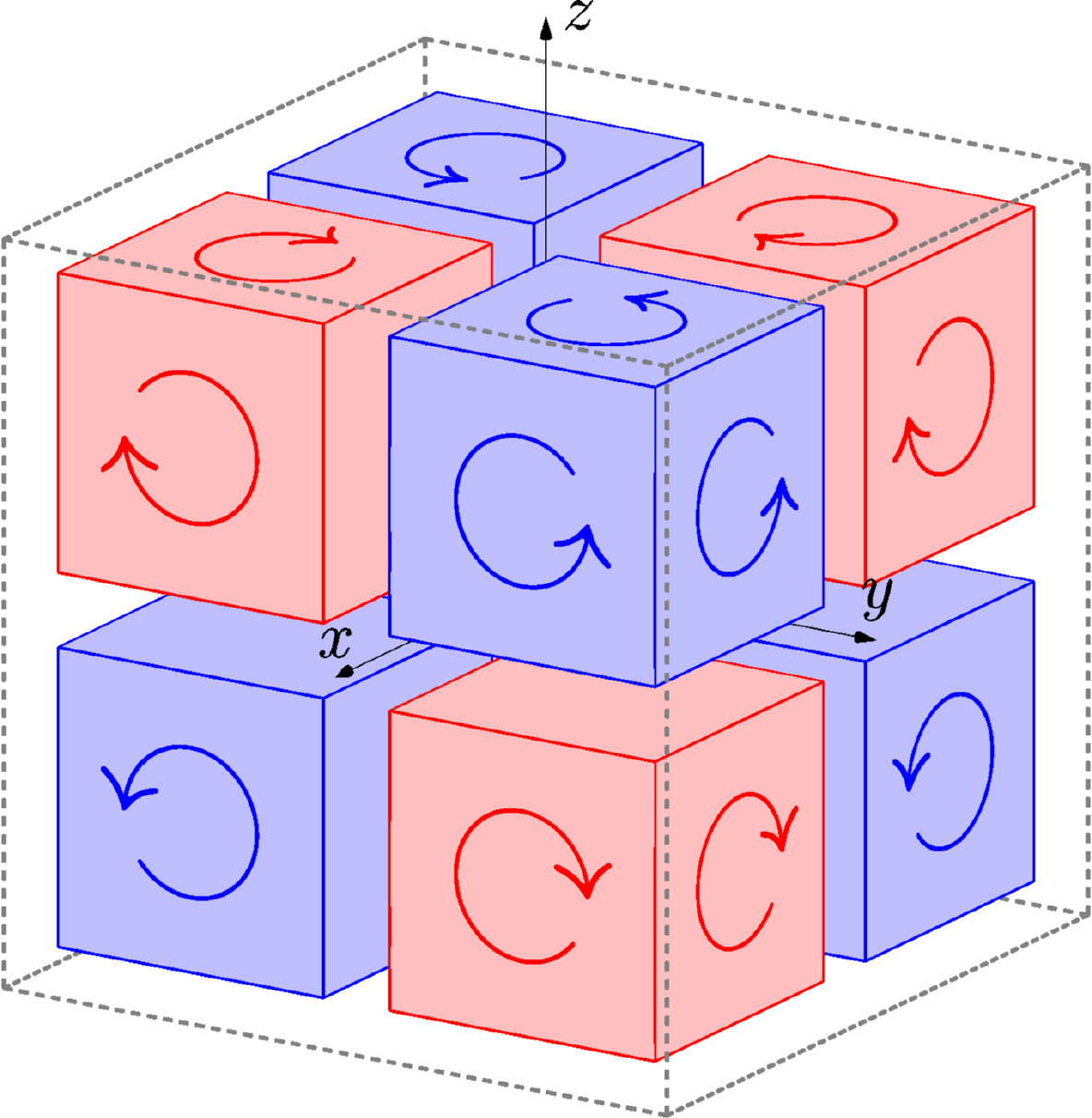}

\caption{A unit cell for the space group $Pmmm\left(47\right)$ is partitioned
into eight cuboids (red and blue online), each of which works as a
fundamental domain. A chiral $\mathfrak{e_{f}m_{f}}$ state is put
on the surface of each fundamental domain; its chirality is indicated
by the arrowed arcs such that all symmetries are respected. On each
interface between two fundamental domains, there are two copies of
$\mathfrak{e_{f}m_{f}}$ states. Let $\left(\mathfrak{e_{f}},\mathfrak{e_{f}}\right)$
(resp. $\left(\mathfrak{m_{f}},\mathfrak{m_{f}}\right)$) denote the
anyon formed by pairing $\mathfrak{e_{f}}$ (resp. $\mathfrak{m_{f}}$)
from each copy. Then the anyons $\left(\mathfrak{e_{f}},\mathfrak{e_{f}}\right)$
and $\left(\mathfrak{m_{f}},\mathfrak{m_{f}}\right)$ can be condensed,
leading to a model that generates (via the stacking operation) cSPT
phases corresponding to $H_{\phi}^{1}\left(G,\mathbb{Z}\right)$ for
$G=Pmmm$.}

\label{fig:sg47_E8}
\end{figure}
The space group $Pmmm$ is generated by $t_{x}$, $t_{y}$, $t_{y}$, 
and three reflections 
\begin{align}
m_{x}: & \left(x,y,z\right)\mapsto\left(-x,y,z\right),\label{eq:mx}\\
m_{y}: & \left(x,y,z\right)\mapsto\left(x,-y,z\right),\label{eq:my}\\
m_{z}: & \left(x,y,z\right)\mapsto\left(x,y,-z\right).\label{eq:mz}
\end{align}
For $G=Pmmm$, Theorem~\ref{thm:H1_formula} in Appendix~\ref{sec:comp_H1}
gives
\begin{gather}
H_{\phi}^{1}\left(G,\mathbb{Z}\right)=\mathbb{Z}_{2}.\label{eq:Pmmm_h1}
\end{gather}
There is no $\mathbb{Z}$ factor as in the case of $Pmm2\left(25\right)$
above.

The cSPT phase generating $H_{\phi}^{1}\left(G,\mathbb{Z}\right)=\mathbb{Z}_{2}$
can be constructed as in Fig.~\ref{fig:sg47_E8}. Step 1: we partition
the 3D space into cuboids of size $\frac{1}{2}\times\frac{1}{2}\times\frac{1}{2}$.
Each such cuboid works as a \emph{fundamental domain}, also know as
an \emph{asymmetric unit} in crystallography \citep{international_tables};
it is a smallest simply connected closed part of space from which,
by application of all symmetry operations of the space group, the
whole of space is filled. Every orientation-reversing symmetry relates
half of these cuboids (blue online) to the other half (red online). Step 2: we attach an $\mathfrak{e_{f}}\mathfrak{\mathfrak{m}_{f}}$
topological state to the surface of each cuboid from inside with all
symmetries in $G$ respected. An $\mathfrak{e_{f}m_{f}}$ topological
state hosts three anyon species, denoted $\mathfrak{e_{f}}$, $\mathfrak{m_{f}}$,
and $\boldsymbol{\varepsilon}$, all with fermionic self-statistics.
Such a topological order can be realized by starting with a $\nu=4$
integer quantum Hall state and then coupling the fermion parity to
a $\mathbb{Z}_{2}$ gauge field in its deconfined phase \citep{KITAEV20062}.
This topological phase exhibits net chiral edge modes under an open
boundary condition. Step 3: there are two copies of $\mathfrak{e_{f}m_{f}}$
states at the interface of neighbor cuboids and we condense $\left(\mathfrak{e_{f}},\mathfrak{e_{f}}\right)$
and $\left(\mathfrak{m_{f}},\mathfrak{m_{f}}\right)$ simultaneously
without breaking any symmetries in $G$, where $\left(\mathfrak{e_{f}},\mathfrak{e_{f}}\right)$
(resp. $\left(\mathfrak{m_{f}},\mathfrak{m_{f}}\right)$) denotes
the anyon formed by pairing $\mathfrak{e_{f}}$ (resp. $\mathfrak{m_{f}}$)
from each copy. After condensation, all the other anyons are confined,
resulting in the desired cSPT phase, denoted $\mathcal{E}$. 

To see that $\mathcal{E}$ generates $H_{\phi}^{1}\left(G,\mathbb{Z}\right)=\mathbb{Z}_{2}$,
we need to check that $\mathcal{E}$ is nontrivial and that $2\mathcal{E}$
(\emph{i.e.}, two copies of $\mathcal{E}$ stacking together) is trivial.
First, we notice that any orientation-reversing symmetry (\emph{i.e.},
a reflection, a glide reflection, an inversion, or a rotoinversion) $g\in G$ is enough
to protect $\mathcal{E}$ nontrivial. Let us consider an open boundary condition of the model, keeping only the fundamental domains enclosed by the surface shown in ~\ref{fig:sg47_E8_STO}. Then the construction in Fig.~\ref{fig:sg47_E8}
leaves a surface of the $\mathfrak{e_{f}m_{f}}$ topological order
respecting $g$. If the bulk cSPT phase is trivial, then the surface
topological order can be disentangled in a symmetric way from the
bulk by local unitary gates with a finite depth. However, this strictly
2D system of the $\mathfrak{e_{f}m_{f}}$ topological order is chiral,
wherein the orientation-reversing symmetry $g$ has to be violated.
This contradiction shows that the bulk is nontrivial by the protection
of $g$.

\begin{figure}
\noindent\begin{minipage}[t]{1\columnwidth}%
\includegraphics[width=0.6\columnwidth]{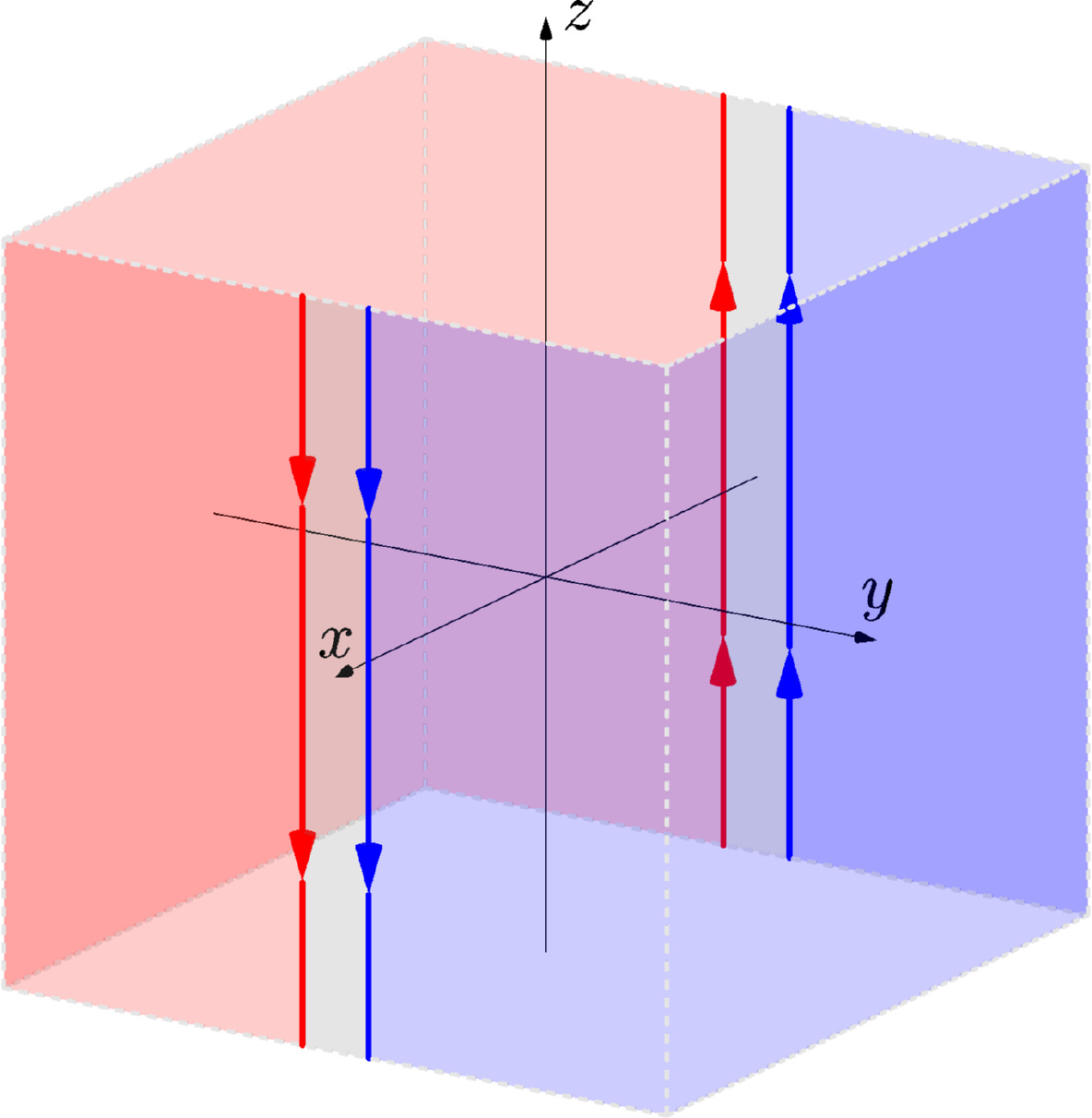}

(a)%
\end{minipage}

\noindent\begin{minipage}[t]{1\columnwidth}%
\includegraphics[width=0.6\columnwidth]{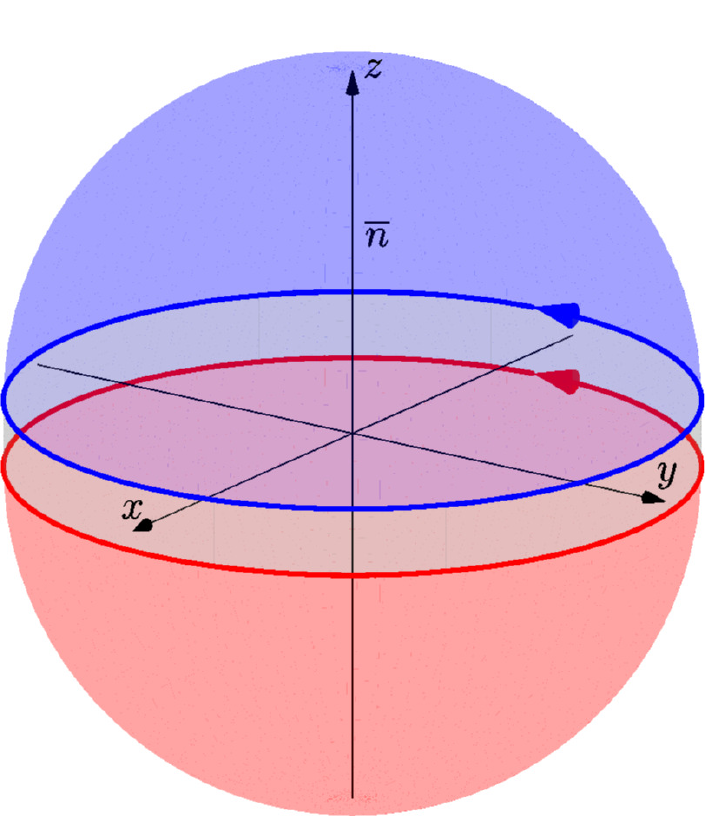}

(b)%
\end{minipage}

\caption{Two 2D systems (red and blue online) of the $\mathfrak{e_{f}m_{f}}$
topological order cannot be glued together into a \emph{gapped} state
respecting any orientation-reversing symmetry; their chiral edge modes
must propagate in the same direction because of the symmetry. (a)
illustrates cases with an rotoinversion axis $\overline{n}$ (including
the spacial cases with an inversion center or a mirror plane). (b)
illustrates cases with a glide reflection plane such as $c:\left(x,y,z\right)\protect\mapsto\left(x,-y,z+1/2\right)$. }

\label{fig:sg47_E8_STO}
\end{figure}

To better understand the incompatibility of a strictly 2D system of
the $\mathfrak{e_{f}m_{f}}$ topological order with any orientation-reversing
symmetry $g$, we view the 2D system as a gluing result of two regions
related by $g$ as in Fig.~\ref{fig:sg47_E8_STO}. The $\mathfrak{e_{f}m_{f}}$
topological order implies net chiral edge modes for each region. Further,
$g$ requires that edge modes from the two region propagate in the
same direction at their 1D interface. Thus, a gapped gluing is impossible,
which shows the non-existence of the $\mathfrak{e_{f}m_{f}}$ topological
order compatible with $g$ in a strictly 2D system. Therefore, the surface $\mathfrak{e_{f}m_{f}}$
topological order respecting $g$ proves the non-triviality of the
bulk cSPT phase $\mathcal{E}$.

On the other hand,  $2\mathcal{E}$ is equivalent to the model obtained
	by attaching an $E_{8}$ state to the surface of each fundamental
	domain from inside in a symmetric way. To check this, let us trace
back the construction of $2\mathcal{E}$: we start with attaching
two copies of $\mathfrak{e_{f}m_{f}}$ states on the surface of each
fundamental domain from inside. Then let us focus a single rectangle interface
between two cuboids. There are four copies of $\mathfrak{e_{f}m_{f}}$
states along it, labeled by $i=1,2$ for one side and $i=3,4$ for
the other side. The symmetries relate $i=1\leftrightarrow i=4$ and
$i=2\leftrightarrow i=3$ separately. If we further condense $(\mathfrak{e}_{\mathfrak{f}}^{1},\mathfrak{e}_{\mathfrak{f}}^{4})$,
$(\mathfrak{m}_{\mathfrak{f}}^{1},\mathfrak{m}_{\mathfrak{f}}^{4})$,
$(\mathfrak{e}_{\mathfrak{f}}^{2},\mathfrak{e}_{\mathfrak{f}}^{3})$,
and $(\mathfrak{m}_{\mathfrak{f}}^{2},\mathfrak{m}_{\mathfrak{f}}^{3})$,
then we get the local state of $2\mathcal{E}$ near the rectangle.
Here $\mathfrak{e}_{\mathfrak{f}}^{i}$ (resp. $\mathfrak{m}_{\mathfrak{f}}^{i}$)
denotes the $\mathfrak{e_{f}}$ (resp. $\mathfrak{m_{f}}$) particle
from the $i^{th}$ copy of $\mathfrak{e_{f}m_{f}}$ state and $(\mathfrak{e}_{\mathfrak{f}}^{i},\mathfrak{e}_{\mathfrak{f}}^{j})$
(resp. $(\mathfrak{m}_{\mathfrak{f}}^{i},\mathfrak{m}_{\mathfrak{f}}^{j})$)
is the anyon obtained by pairing $\mathfrak{e}_{\mathfrak{f}}^{i}$,
$\mathfrak{e}_{\mathfrak{f}}^{j}$ (resp. $\mathfrak{m}_{\mathfrak{f}}^{i}$,
$\mathfrak{m}_{\mathfrak{f}}^{j}$). However, we may alternately condense
$(\mathfrak{e}_{\mathfrak{f}}^{1},\mathfrak{e}_{\mathfrak{f}}^{2})$,
$(\mathfrak{m}_{\mathfrak{f}}^{1},\mathfrak{m}_{\mathfrak{f}}^{2})$,
$(\mathfrak{e}_{\mathfrak{f}}^{3},\mathfrak{e}_{\mathfrak{f}}^{4})$,
and $(\mathfrak{m}_{\mathfrak{f}}^{3},\mathfrak{m}_{\mathfrak{f}}^{4})$.
The resulting local state is two copies of $E_{8}$ states connecting
the same environment in a symmetric gapped way. Thus, the two local
states produced by different condensation procedures have the same
edge modes with the same symmetry behavior and hence are equivalent.
Therefore, $2\mathcal{E}$ is equivalent to the model constructed
by attaching an $E_{8}$ state to the surface of each fundamental
domain from inside. Since the later can be obtained by blowing an
$E_{8}$ state bubble inside each fundamental domain, it (and hence
$2\mathcal{E}$) is clearly trivial.

Thus, we have shown that $\mathcal{E}$ is nontrivial and that $2\mathcal{E}$
is trivial. Therefore, the cSPT phases generated by $\mathcal{E}$
have the $\mathbb{Z}_{2}$ group structure, which holds for any non-orientation-preserving
subgroup of $Pmmm$ such as $Pm\left(6\right)$ and $Pmm2\left(25\right)$.

In particular, for $Pm$, the model constructed in Fig.~\ref{fig:layerE8_y2}
actually presents the same cSPT phase as $\mathcal{E}$ constructed
in Fig.~\ref{fig:sg47_E8}. To see this, we could blow an $E_{8}$
state bubble inside cuboids centered at $\frac{1}{4}\left(1,\pm1,-1\right)+\mathbb{Z}^{3}$
and $\frac{1}{4}\left(-1,\pm1,1\right)+\mathbb{Z}^{3}$, with chirality
opposite to those indicated by the arrowed arcs shown on the corresponding
cuboids in Fig.~\ref{fig:sg47_E8}. This relates $\mathcal{E}$ to
alternately layered $E_{8}$ states; however, each reflection does
not acts trivially on the resulting $E_{8}$ layer on its mirror as
the model in Fig.~\ref{fig:layerE8_y2}. To further show their equivalence,
we look at a single $E_{8}$ layer at $y=0$ for instance. Since it
may be obtained by condensing $(\mathfrak{e_{f}},\mathfrak{e_{f}})$
and $(\mathfrak{m_{f}},\mathfrak{m_{f}}$) in a pair of $\mathfrak{e_{f}m_{f}}$
topological states attached to the mirror, it hence can connect the
$\mathfrak{e_{f}m_{f}}$ surface shown in Fig.~\ref{fig:sg47_E8_STO}
in a gapped way with the reflection $m_{y}$ respected. On the other hand, we know
that an $E_{8}$ state put at $y=0$ with trivial
$m_{y}$ action can also connect to this surface in gapped $m_{y}$-symmetric
way \citep{Hermele_torsor} and is thus equivalent to the corresponding $E_{8}$ layer just mentioned with nontrivial $m_{y}$ action. As a result, the models constructed in Figs.~\ref{fig:layerE8_y2}
and \ref{fig:sg47_E8} realize the same cSPT phase with $Pm$ symmetry.

\subsection{Space group $P\overline{1}\left(2\right)$}

Let us explain the role of inversion symmetry by the example of space
group $P\overline{1}$, which is generated by $t_{x}$, $t_{y}$,
$t_{y}$, and an inversion 
\begin{equation}
\overline{1}:\left(x,y,z\right)\mapsto\left(-x,-y,-z\right).
\end{equation}
For $G=P\overline{1}$, Theorem~\ref{thm:H1_formula} in Appendix~\ref{sec:comp_H1}
tells us that 
\begin{equation}
H_{\phi}^{1}\left(G,\mathbb{Z}\right)=\mathbb{Z}^{3}\times \mathbb{Z}_{2}.\label{eq:sg2_E8}
\end{equation}
The factor $\mathbb{Z}^{3}$ specifies $\frac{1}{8}\left(\gamma_{x},\gamma_{y},\gamma_{z}\right)$
as in the case of $P1$. For instance, the phase labeled by $\frac{1}{8}\left(\gamma_{x},\gamma_{y},\gamma_{z}\right)=\left(0,1,0\right)$
can be constructed by putting a copy of $E8$ state on each of the
planes $y=\cdots,-2,-1,0,1,2,\cdots$ with the symmetries $t_{x},t_{y},t_{z}$
and $\overline{1}$ respected. 

On the other hand, we can generate the factor $\mathbb{Z}_{2}$ in
Eq.~(\ref{eq:sg2_E8}) by the phase constructed in Fig.~\ref{fig:sg47_E8}.
In particular, we have shown that this phase is nontrivial under the
protection any orientation-reversing symmetry, like the inversion
symmetry here, in Sec.~\ref{subsec:cSPT_Pmmm}.

\subsection{Space group $Pc\left(7\right)$}

Finally, we explain the role of glide reflection symmetry by the example
of space group $Pc$, which is generated by $t_{x}$, $t_{y}$, $t_{y}$, 
and a glide reflection
\begin{equation}
c:\left(x,y,z\right)\mapsto\left(x,-y,z+1/2\right).
\end{equation}
For $G=Pc$, Theorem~\ref{thm:H1_formula} in Appendix~\ref{sec:comp_H1}
tells us that 
\begin{equation}
H_{\phi}^{1}\left(G,\mathbb{Z}\right)=\mathbb{Z}\times\mathbb{Z}_{2}.
\end{equation}
The factor $\mathbb{Z}$ specifies $\frac{1}{8}\left(0,\gamma_{y},0\right)$;
the glide reflection symmetry $c$ requires $\gamma_{x}=\gamma_{z}=0$. 

To construct the cSPT phase that generates the summand $\mathbb{Z}_{2}$,
we consider the space group $G'$ generated by translations $t_{x}$, $t_{y}$,
$t_{z}^{\prime}:\left(x,y,z\right)\mapsto\left(x,y,z+\frac{1}{2}\right)$
together with reflections $m_{x}$, $m_{y}$, $m_{z}$ in Eqs.~(\ref{eq:mx}-\ref{eq:mz}). Obviously,
$G\subset G'$ and $G'$ is a space group of type $Pmmm$. Then the
model constructed as in Fig.~\ref{fig:sg47_E8} with respect to $G'$
generates this  $\mathbb{Z}_{2}$ factor of cSPT phases protected
by space group $G$. Particularly, the glide reflection $c\in G$
is enough to protect the corresponding cSPT phase nontrivial by the argument in Sec.~\ref{subsec:cSPT_Pmmm}.

\section{Dimensional reduction and general construction\label{sec:reduction_and_construction}}

\begin{comment}
Hao

\charles{See my comments in Sec. \ref{sec:prediction} and Appendix \ref{sec:indep_rep}.}

Introduce fundamental domain

Review dimensional reduction: reduce to 2-dimensional subcomplex.

Set up sewing construction: all that matters is cancellation of edge
modes.

Topological nature: equivalently, (a) can blow bubbles from nothing
to turn one configuration to another, (b) can move domains wall around
freely, {[}or more technically, (c) trivial $F$ symbol{]}
\end{comment}

%\hsong{in revision}

In the above examples, we have presented cSPT phases by lower dimensional
short-range entangled (SRE) states. Such a representation
for a generic cSPT phase can be obtained by a dimensional reduction procedure;
it is quite useful for constructing, analyzing, and classifying cSPT
phases \citep{Hermele_torsor, Huang_dimensional_reduction}. Below, let us review the idea
of dimension reduction and illustrate how to build cSPT phases with
lower dimensional SRE states in general. More emphasis will be put
on the construction of cSPT phases involving $E_{8}$ states, which
have not been systematically studied for all space groups in the literature. 

Given a space group $G$, we first partition the 3D euclidean space
$\mathbb{E}^{3}$ into fundamental domains accordingly. A \emph{fundamental
domain}, also know as an \emph{asymmetric unit} in crystallography
\citep{international_tables}, is a smallest simply connected closed
part of space from which, by application of all symmetry operations
of the space group, the whole of space is filled. Formally, the partition
is written as 
\begin{equation}
\mathbb{E}^{3}=\bigcup_{g\in G}g\mathcal{F},
\end{equation}
where $\mathcal{F}$ is a fundamental domain and $g\mathcal{F}$ its
image under the action of $g\in G$. If $g$ is not the identity of space group $G$,
then by definition $\mathcal{F}$ and $g\mathcal{F}$ only intersect
in their surfaces at most. In general, $\mathcal{F}$ can be chosen
to be a convex polyhedron: the Dirichlet-Voronoi cell of a point $\mathsf{P}\in\mathbb{E}^{3}$
to its $G$-orbit, with $\mathsf{P}$ chosen to have a trivial stabilizer
subgroup (this is always possible by discreteness of space groups)
\citep{fundamental_domain}. 

\begin{figure}
\includegraphics[width=0.7\columnwidth]{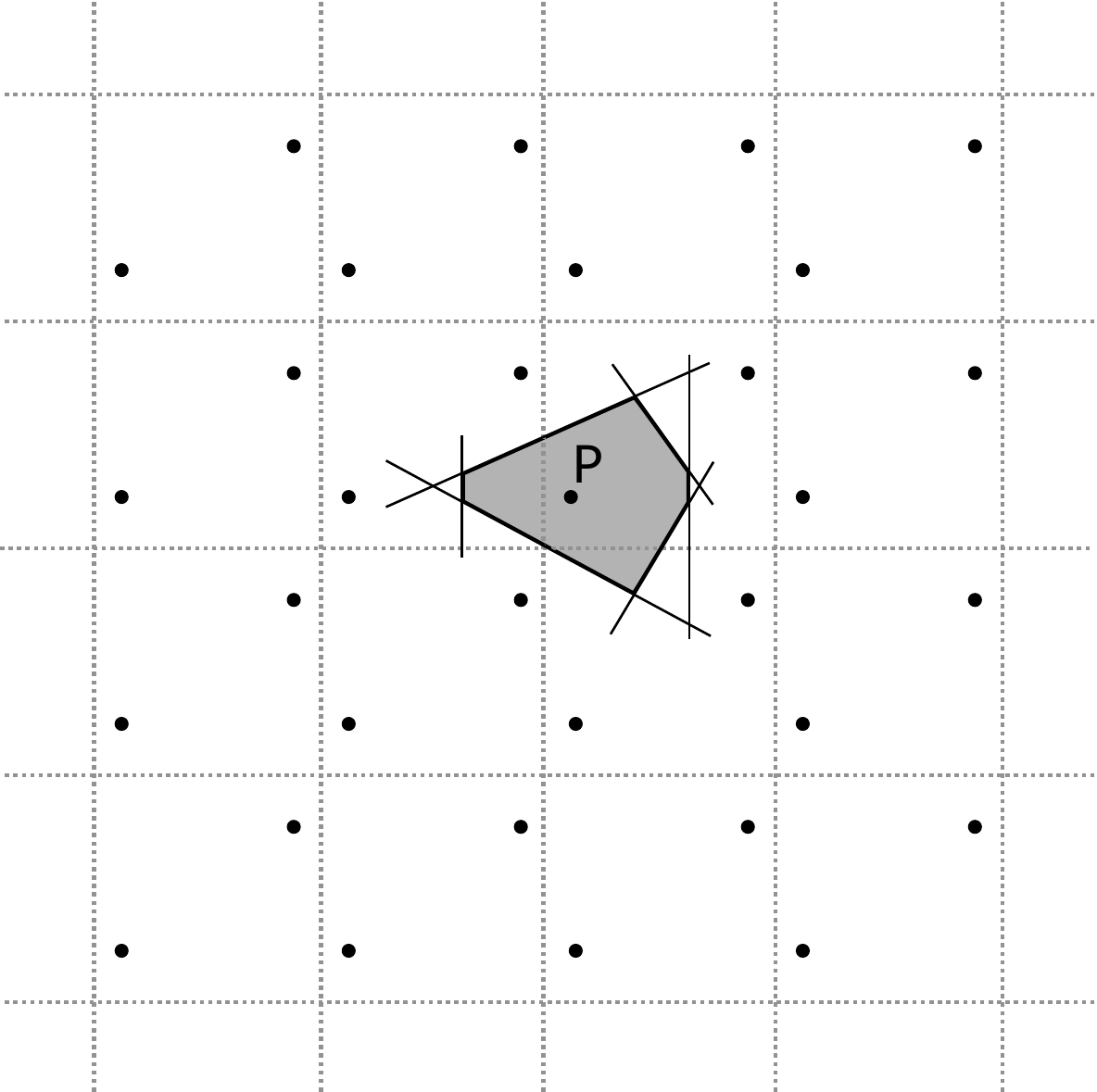}

\caption{The Dirichlet-Voronoi cell (dark region) of a point $\mathsf{P}\in\mathbb{E}^{2}$
to its $G$-orbit (black dots), where $G$ is generated by translations
and an in-plane two-fold rotation.}

\label{fig:p2_DV.pdf}
\end{figure}
The above definition and general construction of fundamental domain
works for space groups in any dimensions. Let us take a lower dimensional
case for a simple illustration: Fig.~\ref{fig:p2_DV.pdf} shows a
fundamental domain given by the Dirichlet-Voronoi cell construction
for wallpaper group No.~2, which is generated by translations and
a two-fold rotation. Clearly, the choice of fundamental domain is
often not unique as in this case; a regular choice of fundamental
domain for each wallpaper group and 3D space group is available in
the International Tables for Crystallography \citep{ITA2006}.

To use terminology from simplicial homology \citep{Elements_AT},
we further partition the fundamental domain $\mathcal{F}$ (resp.
$g\mathcal{F}$) into tetrahedrons $\left\{ \varsigma_{\alpha}\right\} _{\alpha=1,2,\cdots,N}$
(resp. $\left\{ g\varsigma_{\alpha}\right\} _{\alpha=1,2,\cdots,N}$)
such that $X=\mathbb{E}^{3}$ becomes a \emph{$G$-simplicial} complex. In a $G$-simplicial complex, each of its simplices is either completely fixed or mapped onto
another simplex by $g$, $\forall g\in G$. Clearly, all internal
points of every simplex share the same site symmetry\footnote{The \emph{site symmetry} of a point $p$ is the group of symmetry
	operations under which $p$ is not moved, \emph{i.e.}, $G_{p}\coloneqq\left\{ g\in G|gp=p\right\} $.}. To encode the simplex structure, let $\Delta^{n}\left(X\right)$
be the set of $n$-simplices (\emph{i.e.}, vertices for $n=0$, edges
for $n=1$, triangles for $n=2$, and tetrahedrons for $n=3$) and
$X_{n}$ the $n$-skeleton of $X$ (\emph{i.e.}, the subspace made
of all $k$-simplices of $X$ for $k\leq n$).

\subsection{Dimensional reduction of cSPT phases}

Topological phases of matter should admit a topological quantum field
theory (TQFT) description, whose correlation functions do not depend
on the metric of spacetime. Thus, it is natural to conjecture the
following two basic properties of topological phases (including cSPT
phases). (1) Each phase can be presented by a state $\left|\psi\right\rangle $
with arbitrary short correlation length. (2) In each phase, any two
states $\left|\psi_{0}\right\rangle $ and $\left|\psi_{1}\right\rangle $
with correlation length shorter than $r>0$ can be connected a path
of states $\left|\psi_{\tau}\right\rangle $ (parameterized by $\tau\in\left[0,1\right]$)
whose correlation length is shorter than $r$ for all $\tau$. The
conjecture is satisfied by all topological states investigated in
this paper, allowing the dimensional reduction procedure described
below. Its rigorous proof in a reasonable setting, however, remains
an interesting question and goes beyond the scope of this paper. If
the conjecture holds in general for cSPT phases, our classification
in this paper will be complete. Otherwise, we would miss the cSPT
phases where the two properties fail.

Given any cSPT phase for space group $G$, let us now describe the
dimensional reduction procedure explaining why it can be built by
lower dimensional states in general. To start, as conjectured, we
can present this cSPT phase by a state $\left|\Psi\right\rangle $
with correlation length $\xi$ much smaller than the linear size of
the fundamental domain $\mathcal{F}$. The short-range correlation
nature implies that $\left|\Psi\right\rangle $ is the ground state
of some gapped Hamiltonian $H$ whose interaction range is $\xi$
as well. The local part of $H$ inside $\mathcal{F}$ thus describes
a 3D SRE state. It is believed that all 3D SRE states are trivial
or weakly trivial (e.g. layered $E_{8}$ states). Thus, we can
continuously change $H$ into a trivial Hamiltonian inside $\mathcal{F}$
(except within a thin region near the boundary of $\mathcal{F}$)
keeping the correlation length of its ground state smaller than $\xi$
all the time. Removing the trivial degrees of freedom, we are left
with a system on the $2$-skeleton $X_{2}$ of $\mathbb{E}^{3}$. 

Still, the reduced system host no nontrivial excitations. Thus, there
is an SRE state on each 2-simplex (\emph{i.e.}, triangle) $\tau_{\alpha}$, 
indexed by $\alpha$, of $X_{2}$. In particular, it could be $q_{2}\left(\alpha\right)$
copies of $E_{8}$ states (without specifying symmetry), where $q_{2}\left(\alpha\right)\in\mathbb{Z}$
with sign specifying the chirality. These data may be written collectively
as a formal sum $q_{2}=\sum_{\alpha}q_{2}\left(\alpha\right)\tau_{\alpha}$,
which may contain infinitely many terms as $X=\mathbb{E}^{3}$ is
noncompact. On each edge $\ell$, the chiral modes from all triangles
connecting to $\ell$ have to cancel in order for the system to be
gapped. In terms of the simplicial boundary map $\partial$, we thus
have that $\partial q_{2}\coloneqq\sum q_{2}\left(\alpha\right)\partial\tau_{\alpha}$
equals 0. In general, let $C_{k}\left(X\right)$ the set of such formal
sums of $k$-simplices of $X$. Naturally, $C_{k}\left(X\right)$
has an Abelian group structure; $C_{-1}\left(X\right)$ is taken
to be the trivial group. For any integer $k\geq0$, the boundary map
$\partial_{k}$ (or simply $\partial$)$:C_{k}\left(X\right)\rightarrow C_{k-1}\left(X\right)$
is a group homomorphism and let $B_{k}\left(X\right)$ (resp. $Z_{k-1}\left(X\right)$)
denote its kernel (resp. image). Thus, the $E_{8}$ state configuration
on $X_{2}$ is encoded by $q_{2}\in B_{2}\left(X\right)$. 

It clear that $\partial_{k+1}\circ\partial_{k}=0$ and hence $Z_{k}\left(X\right)\subseteq B_{k}\left(X\right)$.
As $C_{k}\left(X\right)$ contains \emph{infinite} sums of simplices,
it is not a standard group of $k$-chains. Instead, it can be viewed
as $\left(3-k\right)$-cochains on the \emph{dual polyhedral decomposition}
(also called \emph{dual block decomposition} \citep{Elements_AT})
of $X$. Thus, $B_{k}\left(X\right)/Z_{k}\left(X\right)$ should be
understood as $\left(3-k\right)^{th}$ cohomology $H^{3-k}\left(X\right)$
rather than $k^{th}$ homology of $X$. Since $X=\mathbb{E}^{3}$
is contractible, its cohomology groups are the same as a point: $H^{n}\left(X\right)$
is $\mathbb{Z}$ for $n=0$ and trivial for $n>0$. Thus, $B_{2}\left(X\right)=Z_{2}\left(X\right)$
and hence any gapped $E_{8}$ state configuration $q_{2}\in B_{2}\left(X\right)$
can be expressed as $q_{2}=\partial q_{3}$ for some $q_{3}\in C_{3}\left(X\right)$.
Also, it is clear that $B_{3}\left(X\right)$ is generated by the
sum of all $3$-simplices with the right-handed orientation, which
is simply denoted by $X$ as well.

As $q_{2}=\partial q_{3}$ is symmetric under $G$, we have $\partial\left(gq_{3}\right)=q\partial q_{3}=\partial q_{3}$
and hence 
\begin{equation}
gq_{3}=q_{3}+\nu^{1}\left(g\right)X,\label{eq:mu1}
\end{equation}
for some $\nu^{1}\left(g\right)\in\mathbb{Z}$. Clearly, $\nu^{1}\left(e\right)=0$
for the identity element $e\in G$. The consistent condition $(gh)q_{3}=g(hq_{3})$
requires that 
\begin{equation}
\nu^{1}(gh)=\phi(g)\nu^{1}(h)+\nu^{1}(g),
\end{equation}
\emph{i.e.}, the cocycle condition for $Z_{\phi}^{1}(G;\mathbb{Z})$. Definitions of group cocycles $Z_{\phi}^{1}(G;\mathbb{Z})$ as well as coboundaries $B_{\phi}^{1}(G;\mathbb{Z})$ and cohomologies $H_{\phi}^{1}(G;\mathbb{Z})$ are given in Appendix~\ref{subsec:H1_sg}.
Thus, $\nu^{1}$ is a normalized 1-cocycle. Moreover, we notice that
$q_{3}$ is not uniquely determined by the $E_{8}$ state configuration
$q_{2}$; solutions to $\partial q_{3}=q_{2}$ may differ by a multiple
of $X$. According to Eq.~(\ref{eq:mu1}), the choice change $q_{3}\rightarrow q_{3}+\nu^{0}X$
leads to $\nu^{1}\left(g\right)\rightarrow\nu^{1}\left(g\right)+\left(d\nu^{0}\right)\left(g\right)$,
where $\nu^{0}\in\mathbb{Z}$ and $\left(d\nu^{0}\right)\left(g\right)\coloneqq\phi\left(g\right)\nu^{0}-\nu^{0}$.
Thus, $\nu^{1}$ is only specified up to a $1$-coboundary.  Therefore,
any $G$-symmetric model (with correlation length much shorter than
simplex size) on $X_{2}$ defines a cohomology group element $[\nu^{1}]\in H_{\phi}^{1}\left(G;\mathbb{Z}\right)$. 

By Lemma~\ref{lemma1}, each $[\nu^{1}]\in H_{\phi}^{1}\left(G;\mathbb{Z}\right)$
can be parameterized by $\nu^{1}\left(t_{v_{1}}\right),\nu^{1}\left(t_{v_{2}}\right),\nu^{1}\left(t_{v_{3}}\right)\in\mathbb{Z}$
(together with $\nu^{1}\left(r\right)\pmod2\in\mathbb{Z}_{2}$ if
$G$ is non-orientation-preserving), where $t_{v_{1}}$, $t_{v_{2}}$,
$t_{v_{3}}$ are three elementary translations that generate the translation
subgroup and $r$ is an orientation-reversing symmetry. Let us explain
their physical meaning by examples. For the model in Fig.~\ref{fig:layerE8_y},
we could pick 
\begin{equation}
q_{3}=-\sum_{i,j,k\in\ZZZ}j\left(t_{x}^{i}t_{y}^{j}t_{z}^{k}\mathcal{F}\right)\label{eq:q3_translations}
\end{equation}
with fundamental domain $\mathcal{F}=\left[0,1\right]\times\left[0,1\right]\times\left[0,1\right]$.
To use the terminology of simplicial homology, $\mathcal{F}$ can
be partitioned into tetrahedrons and be viewed as a formal sum of
them with the right-handed orientation. Moreover, $t_{x}^{i}t_{y}^{j}t_{z}^{k}\mathcal{F}$
denote the translation result of $\mathcal{F}$. It is straightforward
to check that $\partial q_{3}$ equals the sum of 2-simplices (oriented
toward the positive $y$ direction according to the right-hand rule)
on integer $y$ planes; thus, $\partial q_{3}$ corresponds to the
model in Fig.~\ref{fig:layerE8_y}. It is also clear that $t_{v}q_{3}-q_{3}=v^{y} X$
for a translation by $v=\left(v^{x},v^{y},v^{z}\right)$; hence $\nu^{1}\left(t_{x}\right)=0$,
$\nu^{1}\left(t_{y}\right)=1$, and $\nu^{1}\left(t_{z}\right)=0$.
Thus, we get a physical interpretation of $\nu^{1}\left(t_{v}\right)$:
if the model is compactified in the $v$ direction such that $t_{v}^{L}=1$,
then it is equivalent to $L\nu^{1}\left(t_{v}\right)$ copies of $E_{8}$
states as a 2D system. Clearly, models with different $\nu^{1}\left(t_{v}\right)$
on any translation symmetry $t_{v}$ must present distinct cSPT phases.

As another example, for $G=Pmmm$, $H_{\phi}^{1}\left(G;\mathbb{Z}\right)=\mathbb{Z}_{2}$
with element $[\nu^{1}]$ parameterized by $\nu^{1}\left(r\right)\pmod2$
on any orientation-reversing element $r$ of $G$; here $\nu^{1}\left(t_{v}\right)$
on any translation $t_{v}\in G$ is required to be zero by symmetry.
The $E_{8}$ state configuration of the model in Fig.~\ref{fig:sg47_E8}
can be encoded by $q_{3}=\sum_{g\in G_{0}}g\mathcal{F}$ (\emph{i.e.},
the sum of cuboids colored blue online), where $\mathcal{F}=[0,\frac{1}{2}]\times[0,\frac{1}{2}]\times[0,\frac{1}{2}${]}
is a fundamental domain \footnote{To use the terminology of simplicial homology, $\mathcal{F}$ can
	be partitioned into tetrahedrons and be viewed as a formal sum of
	them with the right-handed orientation.} and $G_{0}$ is the orientation-preserving subgroup of $G$. For
any orientation-reversing symmetry $r\in G$ (e.g. $m_{x}$, $m_{y}$,
and $m_{z}$ in Eqs.~(\ref{eq:mx}-\ref{eq:mz})), it is clear that
$rq_{3}=q_{3}-\gamma_{X}$ and hence $\nu^{1}\left(r\right)=-1$. 

For any space group $G$, let $\text{SPT}^{3}(G)$ be the set of $G$-SPT phases, which forms a group under the stacking operation as shown in Appendix~\ref{sec:stacking}.
we will see in Sec.~\ref{subsec:H1_invariance} that the dimensional reduction procedure actually gives a well-define group homorphism $\mathfrak{D}: \text{SPT}^{3}(G)\rightarrow H_{\phi}^{1}(G;\mathbb{Z})$. In particular,
$[\nu^1]\in H_{\phi}^{1}(G;\mathbb{Z})$ is a well-defined invariant, independent of the dimensional reduction details, for each generic $G$-SPT phase. Conversely, we will show in Sec.~\ref{subsec:Construction} that a $G$-symmetric SRE state, denoted $\left|[\nu^{1}]\right\rangle $,  can always be constructed to present a $G$-SPT phase corresponding to each $[\nu^1]\in H_{\phi}^{1}(G;\mathbb{Z})$, resulting in a group homomorphism $\mathfrak{C}:H_{\phi}^{1}(G;\mathbb{Z}) \rightarrow \text{SPT}^{3}(G) $.

%Pick $\mathfrak{C}([\nu^1])$ as reference states. Comparing to them, a generic 
\begin{comment}

By constructions similar to those in Figs.~\ref{fig:layerE8_y} and
\ref{fig:sg47_E8}, we can get a model on $X_{2}$, labeled by its
SRE ground state $\left|X_{2};[\nu^{1}]\right\rangle $, for each
$[\nu^{1}]\in H_{\phi}^{1}\left(G;\mathbb{Z}\right)$. Moreover, we
will show that each such constructed models present a distinct cSPT
phase and that the resulting phases equipped with the stacking operation
have the same group structure as $H_{\phi}^{1}\left(G;\mathbb{Z}\right)$.
In particular, both $\left|X_{2};0\right\rangle $ and $\left|X_{2};[\nu^{1}]\right\rangle \otimes\left|X_{2};-[\nu^{1}]\right\rangle $
present the trivial cSPT phase. The details of these discussions will
be postponed to Secs.~ and . 

Then a generic $G$-SPT phase can be represented by a $G$-symmetry SRE state of the form $\left|[\nu^{1}]\right\rangle \otimes\left|\Psi_{2}\right\rangle $ with  $[\nu^{1}]$ specified by $\left|\Psi\right\rangle $ via the dimensional reduction and $\left|\Psi_{2}\right\rangle \coloneqq \left|-[\nu^{1}]\right\rangle \otimes\left|\Psi\right\rangle $. 

\end{comment}

Then a generic $G$-symmetric SRE state $\left|\Psi\right\rangle $ is equivalent (as a $G$-SPT phase) to $\left|[\nu^{1}]\right\rangle \otimes\left|\Psi_{2}\right\rangle $ with $[\nu^{1}]\in H_{\phi}^{1}(G;\mathbb{Z})$  specified by $\left|\Psi\right\rangle $ via the dimensional reduction and $\left|\Psi_{2}\right\rangle \coloneqq \left|-[\nu^{1}]\right\rangle \otimes\left|\Psi\right\rangle $ obtained by stacking $\left|-[\nu^{1}]\right\rangle$ (an inverse of $\left|[\nu^{1}]\right\rangle$) with $\left|\Psi\right\rangle $. By dimensional reduction, we may represent $\left|\Psi_{2}\right\rangle$ by a $G$-symmetry SRE state (with arbitrarily short correlation length) on $X_{2}$, whose $E_8$ state configuration specifies $0\in H_{\phi}^{1}\left(G;\mathbb{Z}\right)$. Actually, $\left|\Psi_{2}\right\rangle$ can be represented
without using $E_{8}$ states: since $E_{8}$ state configuration
of $\left|\Psi_{2}\right\rangle $ is given by $q_{3}$ satisfying
$gq_{3}=q_{3}$, we can reduce $q_{2}=\partial q_{3}$ to $q_{2}=0$
by blowing $-q_{3}\left(\alpha\right)$ copies of $E_{8}$ state bubbles
in a process respecting $G$ symmetry. Explicitly, as in the example
of $Pm\left(6\right)$ in Sec.~\ref{subsec:Pm}, such phases equipped
with the stacking operation form the summand $H_{\phi}^{5}\left(G;\mathbb{Z}\right)$
in Eq.~(\ref{3D_prediction}) and can be built with lower dimensional
group cohomology phases \citep{Thorngren_sgSPT,cSPT2017}. Let us
briefly describe how to decompose $\left|\Psi_{2}\right\rangle $
(with $q_{2}=0$) into these lower dimensional components. 

\begin{comment}

(with correlation length much shorter than unit cell size) on $X_{2}$
corresponding to $[\nu^{1}]\in H_{\phi}^{1}\left(G;\mathbb{Z}\right)$
is equivalent to the stacking of $\left|X_{2};[\nu^{1}]\right\rangle $
and $\left|\Psi_{2}\right\rangle $, where $\left|\Psi_{2}\right\rangle =\left|X_{2};-[\nu^{1}]\right\rangle \otimes\left|\Psi\right\rangle $
is obtained by stacking $\left|X_{2};-[\nu^{1}]\right\rangle $ and
$\left|\Psi\right\rangle $. Formally, 
\begin{equation}
\left|\Psi\right\rangle \sim_{G}\left|X_{2};[\nu^{1}]\right\rangle \otimes\left|\Psi_{2}\right\rangle ,
\end{equation}
where $\sim_{G}$ denotes the equivalent relation as a cSPT phase
with respect to space group $G$. Clearly, $\left|\Psi_{2}\right\rangle $
is a state with $0\in H_{\phi}^{1}\left(G;\mathbb{Z}\right)$, but
it is not necessarily a trivial cSPT phase. Such a phase 
\end{comment}

Without $E_{8}$ states, $\left|\Psi_{2}\right\rangle $ (represented on $X_{2}$) can only
have nontrivial 2D phases on 2-simplices in mirror planes, where each
point has a $\mathbb{Z}_{2}$ site symmetry effectively working as
an Ising symmetry protecting 2D phases classified by $H^{3}\left(\mathbb{Z}_{2},U\left(1\right)\right)=\mathbb{Z}_{2}$.
As the system is symmetric and gapped, the 2D states associated with
all 2-simplices in the same mirror plane have to be either trivial
or nontrivial simultaneously. Thus, each inequivalent mirror plane contributes
a $\mathbb{Z}_{2}$ factor to cSPT phase classification. Conversely,
a reference model of 2D nontrivial phase protected by the $\mathbb{Z}_{2}$
site symmetry on a mirror plane can be constructed by sewing 2-simplices
together with all symmetries respected. Adding such a model to each
mirror plane with nontrivial state in $\left|\Psi_{2}\right\rangle $
resulting in a state $\left|\Psi_{1}\right\rangle $ which may be
nontrivial only along the 1-skeleton $X_{1}$. Only 1D states lie
in the intersection of inequivalent mirror planes can be nontrivial
projected by the site symmetry $C_{nv}$ for $n=2,4,6$. A reference
nontrivial model, which generates the 1D phases protected by the $C_{nv}$
site symmetry and classified by $H^{2}\left(C_{nv};U\left(1\right)\right)=\mathbb{Z}_{2}$
for $n$ even, can be constructed by connecting states on 1-simplices
in a gapped and symmetric way. Adding such reference states to the
$C_{nv}$ axes where $\left|\Psi_{1}\right\rangle $ is nontrivial,
we get a state $\left|\Psi_{0}\right\rangle $ which may be nontrivial
only on the 0-skeleton $X_{0}$. Explicitly, $\left|\Psi_{0}\right\rangle $
is a tensor product of trivial degrees of freedom and some isolated
0D states centered at the 0-simplices (i.e., vertices) of $X$ carrying
nontrivial one-dimensional representation of its site symmetry (\emph{i.e.},
site symmetry charges). However, some different site symmetry charge configurations may be changed into each other by charge splitting and fusion; they are
not in 1-1 correspondence with cSPT phases, whose classification and
characterization are studied in Ref.~\citep{cSPT2017}. 

Putting all the above ingredient together, we get that a generic cSPT
state $\left|\Psi\right\rangle $ can be reduced to the stacking of
$\left|[\nu^{1}]\right\rangle $, 2D nontrivial states protected
by $\mathbb{Z}_{2}$ site symmetry on some mirror planes, 1D nontrivial
states protected by $C_{nv}$ site symmetry on some $C_{nv}$ axes
with $n$ even, and 0D site symmetry charges. As we have mentioned,
the cSPT phases built without using $E_{8}$ states have a group structure
$H_{\phi}^{5}\left(G;\mathbb{Z}\right)$ (\emph{i.e.}, the first summand
in Eq.~(\ref{3D_prediction})) under stacking operation \citep{cSPT2017,Thorngren_sgSPT}.
Next, we will focus on the models with ground states $\left|[\nu^{1}]\right\rangle $
for $[\nu^{1}]\in H_{\phi}^{1}\left(G;\mathbb{Z}\right)$.

\subsection{$H_{\phi}^{1}\left(G;\mathbb{Z}\right)$ as a cSPT phase invariant
	\label{subsec:H1_invariance}}

Let $\text{SPT}^{3}(G)$ be the set of cSPT phases with space group $G$
symmetry in $d=3$ spatial dimension. As shown in Appendix~\ref{sec:stacking}, the stacking operation equips $\text{SPT}^{3}(G)$ with an Abelian
group structure. Below, let us prove that the dimensional reduction
procedure defines a group homomorphism
\begin{equation}
\mathfrak{D}:\text{SPT}^{3}(G)\rightarrow H_{\phi}^{1}\left(G;\mathbb{Z}\right).\label{eq:D}
\end{equation}
To check the well-definedness of $\mathfrak{D}$, we notice the following
two facts. 

\begin{prop}
	Let $\mathcal{M}$ be a model on $X_{2}$ with $\nu^{1}\left(t_{v_{1}}\right)$
	specified by its $E_{8}$ configuration, where $t_{v_{1}}$ is the
	translation symmetry along a vector $v_{1}\in\mathbb{R}^{3}$. Compactifying
	$\mathcal{M}$ such that $t_{v_{1}}^{L}=1$ results in a 2D system with the invertible topologicial order of $L\nu^{1}\left(t_{v_{1}}\right)$
	copies of $E_{8}$ states.
\end{prop}
\begin{proof}
	Let us only keep the subgroup $H\subseteq G$ of symmetries generated
	by three translations $t_{v_{1}}$, $t_{v_{2}}$, and $t_{v_{3}}$
	along linearly independent vectors $v_{1}$, $v_{2}$, and $v_{3}$
	respectively. For convenience, we now use the coordinate system such
	that $t_{v_{1}}$, $t_{v_{2}}$, $t_{v_{3}}$ work as $t_{x}$, $t_{y}$,
	$t_{z}$ in Eqs.~(\ref{eq:tx}-\ref{eq:tz}). Let $\mathcal{F}'=\left[0,1\right]\times\left[0,1\right]\times\left[0,1\right]$.
	Then $\mathcal{F}'$ is both a fundamental domain and a unit cell
	for $H$. The original triangulation $X$ partitions $\mathcal{F}'$
	into convex polyhedrons, which can be further triangulated resulting
	in a finer $H$-simplex complex $X'$. Clearly, $q_{3}$
	can be viewed as an element of $C_{3}\left(X'\right)$ as well. Moreover,
	putting $\nu^{1}\left(t_{v_{1}}\right)$ (resp. $\nu^{1}\left(t_{v_{2}}\right)$,
	$\nu^{1}\left(t_{v_{3}}\right)$) copies of $\overline{E_{8}}$ states
	on each of integer $x$ (resp. $y$, $z$) planes produces a $H$-symmetric
	model $\mathcal{M}'$ with $E_{8}$ configuration encoded
	by $q_{3}^{\prime}$ satisfying $t_{v_{j}}q_{3}^{\prime}=q_{3}^{\prime}+\nu^{1}\left(t_{v_{j}}\right)X,\forall j=1,2,3$.
	Then $q_{3}-q_{3}^{\prime}$ is invariant under the action of $H$.
	Thus, the original model $\mathcal{M}$ can be continuously changed
	into $\mathcal{M}'$ through a path of $H$-symmetric SRE states. In particular,
	it reduces to a 2D system with the invertible topologicial order of
	$L\nu^{1}\left(t_{v_{1}}\right)$ copies of $E_{8}$ states, when
	the system is compactified such that $t_{v_{1}}^{L}=1$.
\end{proof}

\begin{rem}
	This compactification procedure provides an alternate interpretation
	of $\nu^{1}\left(t_{v_{1}}\right)$, which is now clearly independent
	of the dimensional reduction details and invariant under a continuous
	change of cSPT states. Thus, $\nu^{1}\left(t_{v_{1}}\right)$ is well-defined
	for a cSPT phase. 
\end{rem}

\begin{prop}
	Suppose that $G$ contains an orientation-reversing symmetry $r$.
	Given two $G$-symmetric models $\mathcal{M}_{a}$ and $\mathcal{M}_{b}$
	on $G$-simplicial complex structures $X_{a}$ and $X_{b}$
	of $\mathbb{E}^{3}$ respectively, let $\nu_{a}^{1}\left(r\right)\pmod2$
	and $\nu_{b}^{1}\left(r\right)\pmod2$ be specified by their $E_{8}$
	configurations. If $\mathcal{M}_{a}$ and $\mathcal{M}_{b}$ are
	in the same $G$-SPT phase, then $\nu_{a}^{1}\left(r\right)=\nu_{b}^{1}\left(r\right)\pmod2$.
\end{prop}
\begin{proof}
	The triangulation of $X_{a}$ partitions each tetrahedron of $X_{b}$
	into convex polyhedrons, which can be further triangulated
	resulting a simplicial complex $X$ finer than both $X_{a}$ and $X_{b}$.
	Naturally, both $\mathcal{M}_{a}$ and $\mathcal{M}_{b}$ can be
	viewed as models on $X$ for the convenience of making a comparison;
	the values of  $\left[\nu_{a}^{1}\right]$ and $\left[\nu_{b}^{1}\right]$
	are clearly invariant when computed on a finer triangulation. 
	
	To make a proof by contradiction, suppose $\nu_{a}^{1}\left(r\right)\neq\nu_{b}^{1}\left(r\right)\pmod2$.
	Without loss of generality, we may assume $\nu_{a}^{1}\left(r\right)\pmod2=1$
	and $\nu_{b}^{1}\left(r\right)\pmod2=0$. To compare the difference
	between $\mathcal{M}_{a}$ and $\mathcal{M}_{b}$, let us construct
	an inverse of $\mathcal{M}_{b}$ with symmetry $r$ respected. Since
	$\nu_{b}^{1}\left(r\right)\pmod2=0$, the $E_{8}$ configuration of
	$\mathcal{M}_{b}$ can be described by $\partial q_{3}^{b}$ with
	$q_{3}^{b}\in C_{3}\left(X\right)$ satisfying $rq_{3}^{b}=q_{3}^{b}$.
	Let $\overline{\mathcal{M}_{b}^{\prime}}$ be a model obtained by
	attaching a bubble of $-q_{3}^{b}\left(\varsigma\right)$ copies of
	$E_{8}$ states to $\partial\varsigma,\forall\varsigma\in\Delta_{3}\left(X\right)$
	from its inside. Then adding $\overline{\mathcal{M}_{b}^{\prime}}$
	to $\mathcal{M}_{b}$ cancels its $E_{8}$ configuration; the resulting
	state may only have group cohomology SPT phases left on simplices
	and hence has an inerse, denoted $\overline{\mathcal{M}_{b}^{\prime\prime}}$.
	Let $\overline{\mathcal{M}_{b}}=\overline{\mathcal{M}_{b}^{\prime}}\oplus\overline{\mathcal{M}_{b}^{\prime\prime}}$
	(\emph{i.e.}, the stacking of $\overline{\mathcal{M}_{b}^{\prime}}$
	and $\overline{\mathcal{M}_{b}^{\prime\prime}}$). Then $\overline{\mathcal{M}_{b}}$
	clearly is an inverse of $\mathcal{M}_{b}$ with symmetry $r$ respected.
	Since $\mathcal{M}_{a}$ and $\mathcal{M}_{b}$ are in the same
	$G$-SPT phase, $\overline{\mathcal{M}_{b}}$ is also
	an inverse of $\mathcal{M}_{a}$ with symmetry $r$ respected. Thus,
	$\mathcal{M}_{a}\oplus\overline{\mathcal{M}_{b}}$ admits a $r$-symmetric
	surface.  Moreover, since $\nu_{a}^{1}\left(r\right)\pmod2=1$, the
	$E_{8}$ configuration of $\mathcal{M}_{a}\oplus\overline{\mathcal{M}_{b}}$
	can be described by $\partial q_{3}$ for some $q_{3}\in C_{3}\left(X\right)$
	satisfying $rq_{3}=q_{3}-X$. 
	
	To proceed, we pick an $r$-symmetric region $Y$ with surface shown
	in Fig.~\ref{fig:sg47_E8_STO}. For convenience, $Y$ is chosen to
	be a subcomplex of $X$ (\emph{i.e.}, union of simplices in $X$).
	In addition, $r$ can be an inversion, a rotoinversion, a reflection,
	and a glide reflection in three dimensions. Let $\Pi$ be a plane
	passing the inversion/rotoinversion center, the mirror plane, and
	the glide reflection plane respectively. Clearly, $r\Pi=\Pi$. Adding
	$\Pi$ to the triangulation of $Y$, some tetrahedrons are divided
	in convex polyhedrons, which can be further triangulated resulting
	in a finer triangulation of $Y$. Below, we treat $Y$ as a simplicial
	complex specified by the new triangulation. By restriction, $q_{3}$
	can be viewed as an element of $C_{3}\left(Y\right)$ satisfying $rq_{3}=q_{3}-Y$.
	Let $\mathcal{M}$ be a model on $Y$ with an $r$-symmetric SRE
	surface and a bulk identical to $\mathcal{M}_{a}\oplus\overline{\mathcal{M}_{b}}$.
	Thus, $\partial q_{3}$ describes the bulk $E_{8}$ configuration
	of $\mathcal{M}$.
	
	On the other hand, putting an $r$-symmetric $E_{8}$ state on $\Pi$
	gives a configuration described by $\partial q_{3}^{\prime}$ with
	$q_{3}^{\prime}\left(\varsigma\right)$ equal to $1$ for all 3-simplex
	$\varsigma$ on side of $\Pi$ and $0$ on the other side; accordingly,
	$rq_{3}^{\prime}=q_{3}^{\prime}-Y$ and hence $r\left(q_{3}^{\prime}-q_{3}\right)=q_{3}^{\prime}-q_{3}$.
	Thus, adding a bubble of $q_{3}^{\prime}\left(\varsigma\right)-q_{3}\left(\varsigma\right)$
	copies of $E_{8}$ states to $\partial\varsigma,\forall\varsigma\in\Delta_{3}\left(Y\right)$
	from inside gives an $r$-symmetric SRE model, denoted $\mathcal{M}'$.
	By construction, the bulk $E_{8}$ configuration of $\mathcal{M}'$
	get concentrated on $\Pi$; more precisely, $\mathcal{M}'$ can be
	viewed as a gluing result of the two half surfaces (red and blue)
	in Fig.~\ref{fig:sg47_E8_STO} and an $E_{8}$ state on $\Pi$. As
	$\partial Y$ hosts no anyons, each half surface has $8n$ co-propagating
	chiral boson modes along its boundary, where $n\in\mathbb{Z}$. Due
	to the symmetry $r$, they add up to $16n$ co-propagating chiral
	boson modes, which cannot be canceled by boundary modes of the $E_{8}$
	state on $\Pi$. This implies that $\mathcal{M}'$ cannot be gapped,
	which contradicts that $E_{8}$ is SRE and hence disproves our initial
	assumption. Therefore, $\nu_{a}^{1}\left(r\right)=\nu_{b}^{1}\left(r\right)\pmod2$
	for any two models $\mathcal{M}_{a}$ and $\mathcal{M}_{b}$ in
	the same $G$-symmetric cSPT phase.
\end{proof}

By Lemma~\ref{lemma1}, $\left[\nu^{1}\right]\in H_{\phi}^{1}\left(G;\mathbb{Z}\right)$
is specified by $\nu^{1}\left(t_{v_{1}}\right)$, $\nu^{1}\left(t_{v_{2}}\right)$,
$\nu^{1}\left(t_{v_{3}}\right)$ together with $\nu^{1}\left(r\right)\pmod2\in\mathbb{Z}_{2}$
if $G$ is non-orientation-preserving, where $t_{v_{1}}$, $t_{v_{2}}$,
$t_{v_{3}}$ are three linearly independent translation symmetries
and $r$ is an orientation-reversing symmetry. Combining the above
two facts, we get that any two models (probably on different simplicial
complex structures of $\mathbb{E}^{3}$) in the same $G$-symmetric
cSPT phase must determine a unique $\left[\nu^{1}\right]\in H_{\phi}^{1}\left(G;\mathbb{Z}\right)$.
Thus, $\mathfrak{D}$ in Eq.~(\ref{eq:D}) is well-defined, independent
of the details of the dimensional reduction. Moreover, it clearly
respects the group structure. 

\subsection{Construction of $H_{\phi}^{1}\left(G;\mathbb{Z}\right)$ cSPT phases
	\label{subsec:Construction}}

To show that the group structure of $\text{SPT}^{3}(G)$ (\emph{i.e.},
$G$-SPT phases in $d$ spatial dimensions) is $H_{\phi}^{5}\left(G;\mathbb{Z}\right)\oplus H_{\phi}^{1}\left(G;\mathbb{Z}\right)$,
we analyze three group homomorphisms $\mathfrak{I}$, $\mathfrak{C}$,
and $\mathfrak{D}$, which can be organized as \begin{equation} 
\begin{tikzcd}%[column sep=small]
H_{\phi}^{5}(G;\mathbb{Z})   \arrow[r, "\mathfrak{I}", hook] & 
\text{SPT}^{3}(G)	 \arrow[r, "\mathfrak{D}", twoheadrightarrow] &
H_{\phi}^{1}(G;\mathbb{Z}). \arrow[l, "\mathfrak{C}", bend right, swap]
\end{tikzcd}
\label{eq:extension}
\end{equation}A hooked (resp. two-head) arrow is used to indicate that $\mathfrak{I}$
is injective (resp. $\mathfrak{D}$ is surjective). The map $\mathfrak{I}$
is an inclusion identifying $H_{\phi}^{5}\left(G;\mathbb{Z}\right)$
as a subgroup of $\text{SPT}^{3}(G)$; Refs.~\citep{cSPT2017,Thorngren_sgSPT}
have shown that $H_{\phi}^{5}\left(G;\mathbb{Z}\right)$ classifies
the $G$-symmetric cSPT phases built with lower dimensional group
cohomology phases protected site symmetry. We have defined $\mathfrak{D}$
via dimensional reduction. Clearly, $\mathfrak{D}$ maps all $G$-SPT
phases labeled by $H_{\phi}^{5}\left(G;\mathbb{Z}\right)$ to $0\in H_{\phi}^{1}\left(G;\mathbb{Z}\right)$.
Conversely, the $E_{8}$ configuration of any model on $X_{2}$ with
$\left[\nu^{1}\right]=0$ can be described by $\partial q_{3}$ with
$q_{3}\in C_{3}\left(X\right)$ satisfying $gq_{3}=q_{3},\forall g\in G$.
Thus, it is possible to grow a bubble of $-q_{3}\left(\varsigma\right)$
of $E_{8}$ states inside $\varsigma$ to its boundary for all $\varsigma\in\Delta_{3}\left(X\right)$
in a $G$-symmetric way, canceling all $E_{8}$ states on 2-simplicies.
Thus, any phase with $\left[\nu^{1}\right]=0$ can be represented by
a model built with group cohomology phases only. Formally, the image
of $\mathfrak{I}$ equals the kernel of $\mathfrak{D}$. Below, we
will define the group homomophism $\mathfrak{C}$ by constructing
a $G$-symmetric SRE state representing a $G$-SPT phase with each $\left[\nu^{1}\right]\in H_{\phi}^{1}\left(G;\mathbb{Z}\right)$
and show that $\mathfrak{D}\circ\mathfrak{C}$ equals the identity
map on $H_{\phi}^{1}\left(G;\mathbb{Z}\right)$, which implies the
surjectivity of $\mathfrak{D}$ and further
\begin{equation}
\text{SPT}^{3}(G)\cong H_{\phi}^{5}\left(G;\mathbb{Z}\right)\oplus H_{\phi}^{1}\left(G;\mathbb{Z}\right)
\end{equation}
by the splitting lemma in homological algebra. 

Suppose that $t_{v_{1}}$, $t_{v_{2}}$, and $t_{v_{3}}$ generate
the translation subgroup of $G$. Let $G_{0}$ be the orientation-preserving
subgroup of $G$, \emph{i.e.}, $G_{0}\coloneqq\left\{ g\in G\;|\;\phi\left(g\right)=1\right\} $.
By Lemma~\ref{lemma1}, each $[\nu^{1}]\in H_{\phi}^{1}\left(G;\mathbb{Z}\right)$
can be parameterized by $\nu^{1}\left(t_{v_{1}}\right),\nu^{1}\left(t_{v_{2}}\right),\nu^{1}\left(t_{v_{3}}\right)\in\mathbb{Z}$
together with $\nu^{1}\left(r\right)\pmod2\in\mathbb{Z}_{2}$ if $G$
is non-orientation-preserving (\emph{i.e.}, $G_{0}\neq G$), where
$r$ is an orientation-reversing symmetry. 
%According to Theorem~\ref{thm:H1_formula}, three (resp. one, none) of $\nu^{1}\left(t_{v_{1}}\right),\nu^{1}\left(t_{v_{2}}\right),\nu^{1}\left(t_{v_{3}}\right)$ can be nonzero when $G=P1$ or $P\overline{1}$ (resp. $G$ contains one symmetry direction other than $1$ or $\overline{1}$, $G$ contains more than one symmetry direction). 
Below, we construct models for
generators of $H_{\phi}^{1}\left(G;\mathbb{Z}\right)$ and then the
model corresponding to a generic $[\nu^{1}]$ will be obtained by
stacking.

For a non-orientation-preserving space group $G$, let $\nu_{r}^{1}$
be a 1-cocycle satisfying $\nu_{r}^{1}\left(g\right)=0,\forall g\in G_{0}$
and $\nu_{r}^{1}\left(r\right)\pmod2=1$; the corresponding $[\nu_{r}^{1}]\in H_{\phi}^{1}\left(G;\mathbb{Z}\right)$
clearly has order 2. Let us first construct a $G$-symmetric SRE state
with $[\nu_{r}^{1}]\in H_{\phi}^{1}\left(G;\mathbb{Z}\right)$, which
is done by recasting the illustrative construction in Fig.~\ref{fig:sg47_E8}
in a general setting. We attach an bubble of $\mathfrak{e_{f}}\mathfrak{m_{f}}$
topological state to $\partial\mathcal{F}$ from inside and duplicate
it at $\partial\left(g\mathcal{F}\right)$ by all symmetries $g\in G$.
Then there are two copies of $\mathfrak{e_{f}}\mathfrak{m_{f}}$ topological
states on each 2-simplex, which is always an interface between two
fundamental domains $g_{1}\mathcal{F}$ and $g_{2}\mathcal{F}$ for
some $g_{1},g_{2}\in G$. Let $\left(\mathfrak{e_{f}},\mathfrak{e_{f}}\right)$
(resp. $\left(\mathfrak{m_{f}},\mathfrak{m_{f}}\right)$) denote the
anyon formed by pairing $\mathfrak{e_{f}}$ (resp. $\mathfrak{m_{f}}$)
from each copy. Condensing $\left(\mathfrak{e_{f}},\mathfrak{e_{f}}\right)$
and $\left(\mathfrak{m_{f}},\mathfrak{m_{f}}\right)$ on all 2-simplices
results in a $G$-symmetric SRE state, denoted $\left|\Psi_{r}\right\rangle $.
Its $E_{8}$ configuration can be described by $\partial q_{3}$ with
$q_{3}=\sum_{g\in G_{0}}g\mathcal{F}$; there is a (resp. no) $E_{8}$
state on the interface between $g_{1}\mathcal{F}$ and $g_{2}\mathcal{F}$
with $\phi\left(g_{1}\right)\neq\phi\left(g_{2}\right)$ (resp. $\phi\left(g_{1}\right)=\phi\left(g_{2}\right)$).
Clearly, $rq_{3}=q_{3}-X$ and hence $\nu^{1}\left(r\right)\pmod2=1$.
In addition, with only $G_{0}$ respected, the $E_{8}$ configuration
can be canceled by growing a bubble of an $\overline{E_{8}}$ state
inside each $g\mathcal{F}$ for $g\in G_{0}$. Thus, as a $G_{0}$-symmetric
cSPT phase, $\left|\Psi_{r}\right\rangle $ corresponds to $0\in H_{\phi}^{1}\left(G_{0};\mathbb{Z}\right)$;
in particular, $\nu^{1}\left(t_{v}\right)=0$ for any translation
$t_{v}$. Therefore, $\left|\Psi_{r}\right\rangle $ is a $G$-symmetric
SRE state mapped to $[\nu_{r}^{1}]\in H_{\phi}^{1}\left(G;\mathbb{Z}\right)$
by $\mathfrak{D}$. 
Moreover, the order of $\left|\Psi_{r}\right\rangle $
is also 2 in $\text{SPT}^{3}(G)$ by the same argument as the one in Sec.~\ref{subsec:cSPT_Pmmm}
	showing that stacking two copies of the models in Fig.~\ref{fig:sg47_E8}
	gives a trivial cSPT phase.

Below, the construction of $\mathfrak{C}$ will be completed case
by case based on the number of symmetry directions in the international
(Hermann-Mauguin) symbol of $G$. 

\subsubsection{$G$ with more than one symmetry direction}

For $G$ with more than one symmetry direction, if $G$ is orientation-preserving,
then $H_{\phi}^{1}\left(G;\mathbb{Z}\right)=0$ and hence the definition
of $\mathfrak{C}$ is obvious: $\mathfrak{C}$ maps $0\in H_{\phi}^{1}\left(G;\mathbb{Z}\right)$
to the trivial phase in $G\text{-SPT}$. Clearly, $\mathfrak{D}\circ\mathfrak{C}$
is the identity.

On the other hand, if $G$ does not preserve the orientation of $\mathbb{E}^{3}$,
then $H_{\phi}^{1}\left(G;\mathbb{Z}\right)=\mathbb{Z}_{2}$ with
$[\nu_{r}^{1}]$ the only nontrivial element. Since the order of $\left|\Psi_{r}\right\rangle $
is 2 in $\text{SPT}^{d}(G)$, the group homomorphism $\mathfrak{C}$ can
be specified by mapping $[\nu_{r}^{1}]$ to the cSPT phase presented
by $\left|\Psi_{r}\right\rangle $. By construction, $\mathfrak{D}\circ\mathfrak{C}$
equals the identity on $H_{\phi}^{1}\left(G;\mathbb{Z}\right)$.

\subsubsection{$G=P1$ and $P\overline{1}$}

For $G=P1$ and $P\overline{1}$, we pick a coordinate system such
that $t_{v_{1}}$, $t_{v_{2}}$, $t_{v_{3}}$ work as $t_{x}$, $t_{y}$,
$t_{z}$ in Eqs.~(\ref{eq:tx}-\ref{eq:tz}) (and such that the origin
is an inversion center\footnote{Every orientation-reversing symmetry of $P\overline{1}$ is an inversion.} of $r$ for $P\overline{1}$).
Let $[\nu_{i}^{1}]$ be an element of $H_{\phi}^{1}\left(G;\mathbb{Z}\right)$
presented by a 1-cocycle satisfying $\nu_{i}^{1}\left(t_{v_{j}}\right)=\delta_{ij}$
and $\nu_{i}^{1}\left(r\right)\pmod2=1$ for $i=1,2,3$. As an example,
a $G$-symmetric SRE state $\left|\Psi_{2}\right\rangle $ for $[\nu_{2}^{1}]$
can be constructed by putting an $E_{8}$ state $\left|E_{8}^{y}\right\rangle $
on the plane $y=0$ with the symmetries $t_{v_{1}},t_{v_{3}}$ (as
well as the inversion $r:\left(x,y,z\right)\mapsto-\left(x,y,z\right)$
for $P\overline{1}$) respected and its translation image $t_{v_{2}}^{n}\left|E_{8}^{y}\right\rangle $
on planes $y=n$ for $n\in\mathbb{Z}$. In particular, since there
is a single $E_{8}$ layer passing the inversion center of $r$ in
the case $G=P\overline{1}$, the $E_{8}$ configuration of $\left|\Psi_{2}\right\rangle $
determines $\nu_{2}^{1}\left(r\right)\pmod2=1$. A $G$-symmetric
SRE state $\left|\overline{\Psi}_{2}\right\rangle $, inverse to $\left|\Psi_{2}\right\rangle $
in $\text{SPT}^{3}(G)$, can be obtained by replacing each $E_{8}$ layer
by its inverse. Clearly, $\left|\overline{\Psi}_{2}\right\rangle $
is mapped to $-[\nu_{y}^{1}]$ by $\mathfrak{D}$. Analogously, we
construct a $G$-symmetric SRE state $\left|\Psi_{i}\right\rangle $
for $[\nu_{i}^{1}]$ and its inverse $\left|\overline{\Psi}_{i}\right\rangle $
for $i=1,3$ as well. 

Noticing that each $[\nu^{1}]\in H_{\phi}^{1}\left(G;\mathbb{Z}\right)$
can uniquely be expressed as $\sum_{i=1}^{3}n_{i}[\nu_{i}^{1}]$ for
$G=P1$ and $n_{r}[\nu_{r}^{1}]+\sum_{i=1}^{3}n_{i}[\nu_{i}^{1}]$
for $G=P\overline{1}$ with $n_{r}=0,1$ and $n_{1},n_{2},n_{3}\in\mathbb{Z}$,
we can define $\mathfrak{C}\left([\nu^{1}]\right)$ as the phase presented
by the $G$-symmetric SRE state $\left|\Psi_{r}\right\rangle ^{\otimes n_{r}}\otimes\left|\Psi_{1}\right\rangle ^{\otimes n_{1}}\otimes\left|\Psi_{2}\right\rangle ^{\otimes n_{2}}\otimes\left|\Psi_{3}\right\rangle ^{\otimes n_{3}}$,
where $\left|\Psi_{r}\right\rangle ^{\otimes n_{r}}$ is needed only
for $G=P\overline{1}$ and $\left|\Psi\right\rangle ^{\otimes n}$
denotes the stacking of $\left|n\right|$ copies of $\left|\Psi\right\rangle $
(resp. its inverse $\left|\overline{\Psi}\right\rangle $) for $n\geq0$
(resp. $n<0$). In particular, $\left|\Psi\right\rangle ^{\otimes0}$
denotes any trivial state. By construction, $\mathfrak{C}$ is a group
homomorphism and $\mathfrak{D}\circ\mathfrak{C}$ is the identity
on $H_{\phi}^{1}\left(G;\mathbb{Z}\right)$.

\subsubsection{G with exactly one symmetry diretion other than $1$ or $\overline{1}$
	\label{subsec:Construction1}}

For $G$ with exactly one symmetry direction other than $1$ or $\overline{1}$
(e.g. $I\overline{4}$, $P2$, $P2_{1}/c$, $Pm$), we pick a Cartesian
coordinate system whose $z$ axis lies along the symmetry direction.
With each point represented by a column vector of its coordinates
$u=\left(x,y,z\right)^{\mathsf{T}}\in\mathbb{R}^{3}$, each $g\in G$
is represented as $u\mapsto\phi\left(g\right)R_{z}\left(\varphi\right)u+w$
with $w=\left(w^{x},w^{y},w^{z}\right)^{\mathsf{T}}\in\mathbb{R}^{3}$
and an orthogonal matrix
\begin{equation}
R_{z}\left(\varphi\right)\coloneqq\left(\begin{array}{ccc}
\cos\varphi & -\sin\varphi & 0\\
\sin\varphi & \cos\varphi & 0\\
0 & 0 & 1
\end{array}\right)
\end{equation}
describing a rotation about the $z$ axis. Clearly, $w^{z}$ is independent of the origin position
for $g\in G_{0}$, where $G_{0}$ is the orientation-preserving
subgroup of $G$. Let $\mathbb{A}$ be the collection of $w^{z}$
of all orientation-preserving symmetries. Then $\mathbb{A}=\kappa\mathbb{Z}$
for some positive real number $\kappa$. Pick $h\in G_{0}$ with $w^{z}=\kappa$. 

If $G_{0}\neq G$, then $r\in G-G_{0}$ may be an inversion, a rotoinversion,
a reflection, or a glide reflection. Clearly, there is a plane $\Pi$
perpendicular to the $z$ direction satisfying $r\Pi=\Pi$. If $G_{0}=G$,
we just pick $\Pi$ to be any plane perpendicular to the $z$ direction.
For convenience, we may choose the coordinate origin on $\Pi$ such
that $w^{z}=0$ for $r$. In such a coordinate system,  $w^{z}\in\kappa\mathbb{Z}$
for all $g\in G$ even if $\phi\left(g\right)=-1$. Moreover, $G_{\Pi}\coloneqq\left\{ g\in G|g\Pi=\Pi\right\} $
contains symmetries with $w^{z}=0$. It is a wallpaper group for $\Pi$,
which preserves the orientation of $\Pi$.

Putting a $G_{\Pi}$-symmetric $E_{8}$ layer $\left|E_{8}^{z}\right\rangle $
on $\Pi$ and its duplicate $h^{n}\left|E_{8}^{z}\right\rangle $
on $h^{n}\Pi$, we get an $h$-symmetric SRE state $\left|\Psi_{z}\right\rangle \coloneqq\cdots\otimes h^{-1}\left|E_{8}^{z}\right\rangle \otimes\left|E_{8}^{z}\right\rangle \otimes h\left|E_{8}^{z}\right\rangle \otimes\cdots$.
For all $f\in G_{\Pi}$, $fh^{n}\left|E_{8}^{z}\right\rangle =h^{\phi\left(f\right)n}f_{n}\left|E_{8}^{z}\right\rangle =h^{\phi\left(f\right)n}\left|E_{8}^{z}\right\rangle $,
where $f_{n}=h^{-\phi\left(f\right)n}gh^{n}$ is an element of $G_{\Pi}$
and hence leaves $\left|E_{8}^{z}\right\rangle $ invariant. Thus,
$\left|\Psi_{z}\right\rangle $ is $G_{\Pi}$-symmetric and hence
$G$-symmetric, as each $g=G$ can be expressed as $h^{m}f$ with
$m\in\mathbb{Z}$ and $f\in G_{\Pi}$. A $G$-symmetric SRE state
$\left|\overline{\Psi}_{z}\right\rangle $, inverse to $\left|\Psi_{z}\right\rangle $,
can be obtained by replacing each $E_{8}$ by its inverse. 

Let $[\nu_{z}^{1}]\in H_{\phi}^{1}\left(G;\mathbb{Z}\right)$ (presented
by a 1-cocyle $\nu_{z}^{1}$) be the image of $\left|\Psi_{z}\right\rangle $
under $\mathfrak{D}$. By construction, $\nu_{z}^{1}\left(g\right)=w^{z}/\kappa$
on each orientation-preserving symmetry $g:u\mapsto R_{z}\left(\varphi\right)u+\left(w^{x},w^{y},w^{z}\right)$.
When $G_{0}\neq G$, a single $E_{8}$ layer on $\Pi$ implies the
nonexistence of $r$-symmetric $q_{3}$ with $\partial q_{3}$ describing
the $E_{8}$ configuration of $\left|E_{8}^{z}\right\rangle $ and
hence $\nu_{z}^{1}\left(r\right)\pmod2=1$. Clearly, $\left|\overline{\Psi}_{z}\right\rangle $
is mapped to $-[\nu_{z}^{1}]$ by $\mathfrak{D}$.

Further, we notice that every $[\nu^{1}]\in H_{\phi}^{1}\left(G;\mathbb{Z}\right)$
can be uniquely expressed as $n_{h}[\nu_{z}^{1}]$ (resp. $n_{r}[\nu_{r}^{1}]+n_{h}[\nu_{z}^{1}]$)
when $G_{0}=G$ (resp. $G_{0}\neq G$); comparing values on $r$ and
$h$ gives $n_{h}=\nu^{1}\left(h\right)\in\mathbb{Z}$ and $n_{r}=\nu^{1}\left(r\right)+\nu^{1}\left(h\right)\pmod2\in\left\{ 0,1\right\} $.
Thus, a map $\mathfrak{C}$ can be defined by $n_{h}[\nu_{z}^{1}]\mapsto\left|\Psi_{z}\right\rangle ^{\otimes n_{h}}$
(resp. $n_{r}[\nu_{r}^{1}]+n_{h}[\nu_{z}^{1}]\mapsto\left|\Psi_{r}\right\rangle ^{\otimes n_{r}}\otimes\left|\Psi_{z}\right\rangle ^{\otimes n_{h}}$),
where $\left|\Psi\right\rangle ^{\otimes n}$ denotes the stacking
of $\left|n\right|$ copies of $\left|\Psi\right\rangle $ (resp.
its inverse $\left|\overline{\Psi}\right\rangle $) for $n\geq0$
(resp. $n<0$). In particular, $\left|\Psi\right\rangle ^{\otimes0}$
denotes any trivial state. Since $\left|\Psi_{r}\right\rangle $ has
order 2 in $\text{SPT}^{3}(G)$, $\mathfrak{C}$ is a group homomorphism.
By construction, $\mathfrak{D}\circ\mathfrak{C}$ is clearly the identity
on $H_{\phi}^{1}\left(G;\mathbb{Z}\right)$.

\section{Discussion \label{sec:discussions}}

Dimensional reduction procedure suggests the cSPT phases can be presented by lower dimensional SRE states in general. Earlier works systematically showed that the restricted bosonic cases ignoring possible presence of 2D nontrivial invertible topological orders (\emph{i.e.}, $E_8$ states or its multiples) in $d\leq3$ spatial dimensions classified by the group cohomology  $H_{\text{Borel},\phi}^{d+1}(G;U(1)) \cong H_{\phi}^{d+2}(G;\mathbb{Z})$. 

%This result supports the Crystalline Equivalence Principle: no matter symmetries keep the locations of degrees of freedom or not, the corresponding SPT phases are classified by 

%the same group cohomology also partially classifies the topological phases protected by internal symmetries (\emph{i.e.}, symmetries that keep locations of degrees of freedom, such as Ising symmetry and time reversal symmetry). 

In this paper, we studied the ignored cSPT states and found that the complete classification of cSPT phases is given by the twisted generalized cohomology $h_{\phi}^{d}(G;F_{\bullet})\coloneqq[EG,\Omega F_{d+1}]_{G}$ with $\Omega$-spectrum $F_{\bullet}=\{F_{d}\}$ given by Eq.~(\ref{F1}-\ref{F3}) in low dimensions for bosonic systems. Since $h_{\phi}^{d}(G)$ 
was previously conjectured to completely classify SPT phases with \emph{internal} symmetry (\emph{i.e.}, symmetry keeping the location of degrees of freedom) $G$; thus, our result supports the Crystalline Equivalent Principle at the level  beyond group cohomology phases.

%We focused on the cSPT phases for space groups in 
In particular, 
the abelian group $\text{SPT}^{d}(G)$ of cSPT phases (equipped with the stacking operation) for space group $G$ in $d=3$ spatial dimensions is isomorphic to  $h_{\phi}^{3}(G;F_{\bullet})\cong H_{\phi}^{5}(G;\mathbb{Z})\oplus H_{\phi}^{1}(G;\mathbb{Z})$, where $\phi$ indicates that $g\in G$ acting as multiplication by $\phi(g)=\pm 1$ depending on whether space orientation is preserved by $g$ or not.
The summand $H_{\phi}^{1}(G;\mathbb{Z})$, missed in the  group cohomology proposal, is related to $E_{8}$
state configurations on the 2-skeleton of the space  $\mathbb{E}^3$ which is triangulated into a $G$-simplicial complex. According to Theorem~\ref{thm:H1_formula}, it is isomorphic to $\mathbb{Z}^{k}$ if $G$ preserves orientation and $\mathbb{Z}^{k} \times \mathbb{Z}_{2}$ otherwise, where $k=0,1,3$ can be easily ready off from the international (Hermann-Mauguin) symbol of $G$. The results of $H_{\phi}^{5}(G; \mathbb{Z})$, $H_{\phi}^{1}(G;\mathbb{Z})$, and  $\text{SPT}^{3}(G)$ is tabulated in Table~\ref{table:H1}.

To show $H_{\phi}^{5}(G;\mathbb{Z})\oplus H_{\phi}^{1}(G;\mathbb{Z}) \cong \text{SPT}^{3}(G)$ and to explain its physical meaning, 
we defined a group homorphism 
$\mathfrak{C}:H_{\phi}^{1}(G;\mathbb{Z})\rightarrow \text{SPT}^{3}(G)$ by constructing a model for each $[\nu^{1}] \in H_{\phi}^{1}(G;\mathbb{Z})$.
Explicitly, if $G$ does not preserve space orientation, then a model corresponding to the nontrivial element of the $\mathbb{Z}_{2}$ factor in $H_{\phi}^{1}(G;\mathbb{Z})$ 
can be obtained through anyon condensation from $\mathfrak{e_f m_f}$ state bubbles on boundary of fundamental domains. Such a model gives an example how lower dimensional invertible topological orders naturally appear and is responsible for the existence of bosonic cSPT phases beyond group cohomology. Thus, a twisted generalized cohomology $h_{\phi}^{3}(G;F_{\bullet})\cong H_{\phi}^{5}(G;\mathbb{Z})\oplus H_{\phi}^{1}(G;\mathbb{Z})$ is needed for classifying cSPT phases completely; it naturally includes crystalline invertible topological phases generated by layered $E_{8}$ states, which corresponds to the $\mathbb{Z}^{k}$ factor in $H_{\phi}^{1}(G;\mathbb{Z})$.

We now discuss some natural generalizations and the outlook for further developments
motivated by the results presented here. First, the isomorphism $h_{\phi}^{3}(G;F_{\bullet})\cong H_{\phi}^{5}(G;\mathbb{Z}) \oplus H_{\phi}^{1}(G;\mathbb{Z})$ holds and our computation strategy of $H^{1}(G;\mathbb{Z})$ works in general, predicting a universal $\mathbb{Z}_{2}$ factor (from $H_{\phi}^{1}(G;\mathbb{Z})$) lying beyond the group cohomology classification (\emph{i.e.}, $H_{\phi}^{5}(G;\mathbb{Z})$) in the presence of any spacetime-orientation-reversing symmetry. There may be an extra $\mathbb{Z}^k$ factor labeling the phases of infinite order (e.g. layered $E_8$ states with the same chirality), which are invertible but not completely trivial even without symmetry. On the other hand, phases of finite order would become trivial when symmetry protection is removed.

In particular, the result of $H_{\phi}^1(G;\mathbb{Z})$ for any magnetic space group $G$ is given in Appendix~\ref{subsec:H1_mg}. It only depends on the magnetic point group $\wp G$ associated with $G$ and hence can be read off directly from the Opechowski-Guccione or Belov-Neronova-Smirnova symbol of $G$. Type I magnetic space group is just pure space group and the corresponding $H_{\phi}^{\phi}(G;\mathbb{Z})$ is already tabulated in Table~\ref{table:H1}. 
If $G$ is of type II or type IV, then $H_{\phi}^{\phi}(G;\mathbb{Z})$ is always $\mathbb{Z}_{2}$. If $G$ is of type III, then 
$H_{\phi}^{1}(G;\mathbb{Z})=\mathbb{Z}^{k}\times\mathbb{Z}_{2}^{\ell}$ with $k$ and $\ell$ listed in Table~\ref{tab:mg} for all the 58 possibilities of  the associated magnetic point group.

Secondly, we have carefully checked that $h_{\phi}^{d}(BG;F_{\bullet})$ gives a complete classification of cSPT phases, \emph{i.e.}, $h_{\phi}^{d}(BG;F_{\bullet}) \cong \text{SPT}^{d}(G)$, for any space group $G$ (at least for low dimensions $d\leq3$), provided that every cSPT phase can be represented by a state with arbitrarily short correlation length.
It is still an open question to see weather a generic crystalline gapped quantum phase hosting no nontrivial excitation always has such a representation. 
If it is true, then the dimensional reduction procedure applies to cSPT phases in general and we will be more confident that our classification is complete without further technical assumption.

Another remaining important problem is to complete a generalized cohomology theory for fermionic SPT phases. Along the same line of thinking, fermionic SPT (including cSPT) phases should also be classified by some generalized cohomology. 
Obviously, the corresponding $\Omega$-spectrum $F_{\bullet}=\{F_{n}\}$ must be different the one describing Bosonic SPT phases. However, the fermionic situation is much richer: a symmetry group $G$ may fractionalize on fermions in many different ways. 
For instance, a mirror reflection may square to $(-1)^{F}$ and translations may only commute
we may have $t_{x}t_{y}=(-1)^{F} t_{y} t_{x}$ and $m^{2}=(-1)^{F}$, where $t_{x}$, $t_{y}$ are two translations, $m$ denotes a mirror reflection, and $(-1)^{F}$ is the fermion parity which assigns distinct signs to states with even and odd number of fermions. It is clear that different symmetry fractionalizations on fermions are labeled by $H^{2}(G;\mathbb{Z}_{2})$. If possible,
we would desire a generalized cohomology that classifies fermionic SPT phases in a universal way, which produces a classification with a quadruple input $(d, G, \phi, \omega)$ and apply it for various physical systems, where $d$ is the dimension of space, $G$ is the symmetry group, $\phi: G\rightarrow \{\pm 1\}$ determines whether spacetime orientation is preserved or not, and 
$\omega \in H^{2}(G;\mathbb{Z}_{2})$ describing symmetry fractionalization on fermions.

\begin{acknowledgments}
We  would like to thank  Dominic  Else and Meng Cheng for  useful  correspondence. We are  grateful to Liang Fu, Michael Hermele, and Yi-Ping Huang for collaboration  on  related
prior work. H.S. acknowledges financial support from the Spanish MINECO grant FIS2015-67411, the CAM research consortium QUITEMAD+ S2013/ICE-2801, and Grant FEI-EU-17-14 PICC.
The research of S.-J.H. is supported by the U.S. Department of Energy, Office of Science, Basic Energy Sciences (BES) under Award number DE-SC0014415. 
\end{acknowledgments}

\appendix

\section{Computation of $H_{\phi}^{1}(G;\ZZZ)$\label{sec:comp_H1}}

\subsection{230 space groups (3D) \label{subsec:H1_sg}}

Here we compute $H_{\phi}^{1}(G;\ZZZ)$, which classifies the $E_{8}$ state configuration for bosonic SPT phases, for all 230 space groups in three spatial dimensions.
We will first prove a formula for $H_{\phi}^{1}(G;\ZZZ)$,
in Theorem \ref{thm:H1_formula}, which will then be used to obtain the classification, in Table \ref{table:H1}.

\begin{thm} Let $G$ be a 3D space group and $\phi:G\fromto\braces{\pm1}$
	be the homomorphism that sends orientation-preserving elements to
	1 and the rest to $-1$. Then 
	\begin{eqnarray}
	H_{\phi}^{1}(G;\ZZZ)=\begin{cases}
	\ZZZ^{k}, & \mbox{\ensuremath{G} preserves orientation},\\
	\ZZZ^{k}\times\ZZZ_{2}, & \mbox{otherwise},
	\end{cases}
	\end{eqnarray}
	where $k=0,1,3$ respectively if, in the international (Hermann-Mauguin) symbol, there are more than one symmetry direction
	listed, exactly one symmetry direction listed and it is not $1$ or
	$\bar{1}$, and exactly one symmetry direction listed and it is $1$
	or $\bar{1}$. \label{thm:H1_formula} \end{thm}

We recall that $H_{\phi}^{1}(G;\ZZZ)$ is the quotient of group 1-cocycles
by group 1-coboundaries \citep{AdemMilgram}. That is, 
\begin{equation}
H_{\phi}^{1}(G;\ZZZ)\coloneq Z_{\phi}^{1}(G;\ZZZ)/B_{\phi}^{1}(G;\ZZZ),\label{Z1/B1}
\end{equation}
where 
\begin{eqnarray}
Z_{\phi}^{1}(G;\ZZZ) & \coloneq & \braces{\nu^{1}:G\fromto\ZZZ~|~d\nu^{1}=0},\label{Z1}\\
B_{\phi}^{1}(G;\ZZZ) & \coloneq & \braces{d\nu^{0}~|~\nu^{0}\in C_{\phi}^{0}(G;\ZZZ)}.\label{B1}
\end{eqnarray}
The condition $d\nu^{1}=0$ (cocycle condition) reads 
\begin{equation}
\nu^{1}(g_{1}g_{2})=\nu^{1}(g_{1})+\phi(g_{1})\nu^{1}(g_{2}),~~\forall g_{1},g_{2}\in G,\label{cocycle_condition-1}
\end{equation}
whereas the elements $d\nu^{0}$ (group 1-coboundaries) are of the
form 
\begin{eqnarray}
d\nu^{0}:G & \fromto & \ZZZ\nonumber \\
g & \mapsto & \begin{cases}
2m, & \phi(g)=-1,\\
0, & \phi(g)=1,
\end{cases}\label{coboundaries}
\end{eqnarray}
for some $m\in\ZZZ$. To proceed, we fix three linearly independent
translations $t_{v_{1}},t_{v_{2}},t_{v_{3}}\in G$ as well as an orientation-reversing
element $r\in G$, if any, where $t_{v}$ denotes a translation by
vector $v\in\RRR^{3}$. We have the following lemma.

\begin{lem}\label{lemma1}An element $\nu^{1}\in Z_{\phi}^{1}(G;\ZZZ)$
	is completely determined by $\nu^{1}(t_{v_{1}})$, $\nu^{1}(t_{v_{2}})$,
	and $\nu^{1}(t_{v_{3}})$, together with $\nu^{1}(r)$ if $G$ is
	non-orientation-preserving. \end{lem}

\begin{proof} Setting $g_{1}=g_{2}$ to be the identity element $1\in G$
	in Eq.\,(\ref{cocycle_condition-1}), we see that $\nu^{1}(1)=2\nu^{1}(1)$,
	so $\nu^{1}(1)=0$. If an element $g\in G$ is orientation-preserving
	and has finite order $n$, then Eq.\,(\ref{cocycle_condition-1})
	implies $n\nu^{1}(g)=\nu^{1}(g^{n})=\nu^{1}(1)=0$, so $\nu^{1}(g)=0$.
	If an element $g\in G$ is orientation-preserving and has infinite
	order, then it is either a translation or a screw rotation; either
	way, there exists an integer $n$ such that $g^{n}=t_{v_{1}}^{n_{1}}t_{v_{2}}^{n_{2}}t_{v_{3}}^{n_{3}}$
	for some integers $n_{1},n_{2},n_{3}$, and Eq.\,(\ref{cocycle_condition-1})
	implies 
	\begin{equation}
	\nu^{1}(g)=\frac{1}{n}\brackets{n_{1}\nu^{1}(t_{v_{1}})+n_{2}\nu^{1}(t_{v_{2}})+n_{3}\nu^{1}(t_{v_{3}})}. \label{eq:nu_t}
	\end{equation}
	Thus the value of $\nu^{1}$ on the orientation-preserving subgroup
	of $G$ is completely determined. The orientation-reversing elements
	of $G$, if any, can all be written $gr$ for some orientation-preserving
	$g$, and Eq.\,(\ref{cocycle_condition-1}) implies 
	\begin{equation}
	\nu^{1}(gr)=\nu^{1}(g)+\nu^{1}(r).\label{extending_by_r}
	\end{equation}
	Thus the value of $\nu^{1}$ on the orientation-reversing subset of
	$G$ is completely determined as well. \end{proof}

Let us show that there are constraints on the integers $\nu^{1}(t_{v_{1}})$,
$\nu^{1}(t_{v_{2}})$, and $\nu^{1}(t_{v_{3}})$. It is convenient
to represent the triple $\paren{\nu^{1}(t_{v_{1}}),\nu^{1}(t_{v_{2}}),\nu^{1}(t_{v_{3}})}$
by the unique vector $\mu\in\RRR^{3}$ such that $\nu^{1}(t_{v})=\mu\cdot v$.
In Cartesian coordinates, each point of the Euclidean space $\mathbb{E}^3$ can be identified with a column vector $u=(x,y,z)^{\mathsf{T}}\in \mathbb{R}^3$ and each $g\in G$ can be presented as $u\mapsto W u + w$, where the superscript $\mathsf{T}$ denotes matrix transpose, $W \in O(3)$ (\emph{i.e.}, $W$ is an orthogonal matrix) and $w\in \mathbb{R}^3$ is a column vector. Thus, any $g\in G$ can be denoted by a pair $(W,w)\in O(3)\times \mathbb{R}^3$. Clearly, $\phi$ is also well-defined for $W\in O(3)$ based on whether $W$ preserves orientation
and we have $\phi(g)=\phi(W)=\det (W)$. 
Noticing 
\begin{eqnarray}
gt_{v}g^{-1}=t_{Wv},~~\forall t_{v}\in G,
\end{eqnarray}
and applying
$\nu^{1}$ to both sides and utilizing the cocycle condition (\ref{cocycle_condition-1}),
we obtain 
\begin{equation}
\phi(g)\nu^{1}(t_{v})=\nu^{1}(t_{Wv}),\label{W_condition_1}
\end{equation}
which can be recast as 
\begin{equation}
\phi(W)\mu\cdot v=\mu\cdot Wv.
\end{equation} 
Since this holds for all $v\in\RRR^{3}$ for which $t_{v}\in G$,
which span all three dimensions, we must have 
\begin{equation}
\left[\phi(W)I-W^{\mathsf{T}}\right]\mu=0,
\end{equation}
or equivalently $[\phi(W) W - I]\mu=0$ with $I$ denoting the identity matrix.
Thus we see that
\begin{equation}
\mu\in\bigcap_{W\in{\wp G}}\kernel\brackets{\phi(W)W-I},
\end{equation}
where $\wp$ denotes the group homomorphism $G\rightarrow O(3)$ that maps $(W,w)$ to $W$ and hence $\wp G$ is the point group of $G$.

The dimensionality of $\bigcap_{W\in \wp G}\kernel\brackets{\phi(W)W-I} \subseteq \mathbb{R}^3$ can
be determined as follows.

\begin{lem}  The vector space $\bigcap_{W \in \wp G}\kernel\brackets{\phi(W)W-I}$
	is 0-, 1-, or 3-dimensional if, respectively, in the international (Hermann-Mauguin)
	symbol of space group $G$, there are more than one symmetry directions listed,
	exactly one symmetry direction listed and it is not $1$ or $\bar{1}$,
	or exactly one symmetry direction listed and it is $1$ or $\bar{1}$.
\end{lem}

\begin{proof} 
	We can focus on international symbol of the point group,
	which has the same number of symmetry directions as the space group.
	A symbol $n$ (resp.\,$\bar{n}$) represents an $n$-fold rotation
	(resp.\,rotoinversion) about a given axis. It is easy to see, if
	$W$ represents such a rotation or rotoinversion, that $\kernel\brackets{\phi(W)W-I}$
	is nothing but the symmetry axis, unless $n=1$, in which case $\kernel\brackets{\phi(W)W-I}=\RRR^{3}$.
	Now we simply take the intersection for all symmetry generators listed
	in the international (Hermann-Mauguin) symbol. 
\end{proof}

%If $g=(W,w)\in G$ preserves orientation, then Eq.~\eqref{eq:nu_t} implies that $\nu^{1}(g)=\mu\cdot w$. 

For short, let $V$ denote $\bigcap_{W\in \wp G} \kernel\brackets{\phi(W)W-I}$. Due to the constraint that $\nu^{1}$ is integer-valued, $\mu$ in fact
lives in a discrete subset of $V$. Below, let us express this subset explicitly and show that it is isomorphic to $\mathbb{Z}^{k}$ with $k\coloneqq \dim V$. 
To label the translations in $G$ along $V$, 
let $\mathbb{L}(G)  \coloneqq\left\{ v\in \mathbb{R}^3|t_{v}\in G\right\}$ and $\mathbb{L}(G, V)\coloneqq \mathbb{L}(G)\cap V$. We have $\mathbb{L}(G,V)\cong \mathbb{Z}^{k}$ with generators denoted $a_{j}$ for $j=1,2,\cdots,k$. 
There is a group homomorphism $\beta: G_0\rightarrow V$ mapping $g=(W,w' + w'')$ to $w' \in V$, where $w'' \perp V$ and $G_0$ denotes the orientation-preserving subgroup of $G$.
%Then $\mathbb{L}\left(G,V\right) \subseteq \beta\left(G_{0}\right)\subseteq\frac{1}{\left|\wp G_0\right|}\mathbb{L}\left(G,V\right)$ and hence $\beta(G_0)\cong \mathbb{Z}^k$.
Then we have
$\beta\left(G_{0}\right)\subseteq\frac{1}{\left|\wp G_0\right|}\mathbb{L}\left(G,V\right)$ based on a case-by-case discussion: (1) This is obvious if $\wp G_0$ is trivial; (2) If $\wp G_0$ is nontrivial and $W=1$, then $w'+w''\in \mathbb{L}(G)$ and $\left|\wp G_0\right| w' =
\sum_{i=1}^{\left|\wp G_0\right|}W_{*}^{i}(w'+w'')\in \mathbb{L}(G,V)$ for any nontrivial $W_{*}\in \wp G_0$; (3) If $W\neq 1$, then $g^{\left|\wp G_0\right|}$ is the translation by $\left|\wp G_0\right| w'\in \mathbb{L}(G,V)$.
Thus, $\mathbb{L}\left(G,V\right) \subseteq
\beta\left(G_{0}\right)\subseteq\frac{1}{\left|\wp G_0\right|}\mathbb{L}\left(G,V\right)$ and hence $\beta(G_0)\cong \mathbb{Z}^k$. Pick $\{a_j\}_{j=1,2,\cdots,k}$ that generates $\beta(G_0)$ and let $\{\alpha_i\}_{i=1,2,\cdots,k}$ be the dual basis satisfying $\alpha_{i}\cdot a_{j}=\delta_{ij}$. Pick $g_{j}$ such that $\beta(g_j)=a_j$ for $j=1,2,\cdots,k$. It is clear that $\mu$ corresponding to $\nu^{1}\in Z_\phi^{1}(G;\mathbb{Z})$ has the form $\sum_{i=1}^{k} \nu^{1}(g_i)\alpha_{i}$.

Conversely, each $\mu \in \langle \left\{ \alpha_{i}\right\} _{i=1,2,\cdots,k} \rangle $
determines a cocycle $\nu^{1}\in Z_{\phi}^{1}(G_{0};\mathbb{Z})$ mapping $(W,w)\in G_{0}$ to $\mu \cdot w\in\mathbb{Z}$, 
where $\langle \left\{ \alpha_{i}\right\} _{i=1,2,\cdots,k} \rangle \cong \mathbb{Z}^{k} $ is the lattice generated by $\left\{ \alpha_{i}\right\} _{i=1,2,\cdots,k}$.
If $G_0=G$, this completes a one-to-one correspondence between $Z_{\phi}^{1}(G;\mathbb{Z})$ and $\langle \left\{ \alpha_{i}\right\} _{i=1,2,\cdots,k} \rangle$. 
If $G_0 \neq G$, then it is straightforward to check that $\nu^{1}\left(ghg^{-1}\right)=\phi\left(g\right)\nu^{1}\left(h\right), \forall g\in G, \forall h\in G_{0}$ from construction and hence that  Eq.~(\ref{extending_by_r})
always extends to a well-defined cocycle $\nu^{1}\in Z_{\phi}^{1}(G;\mathbb{Z})$ regardless of the value of $\nu^{1}(r)\in \mathbb{Z}$ with $r\in G-G_{0}$. 
%Thus, $Z_{\phi}^{1}(G;\mathbb{Z}) \cong \langle \left\{ \alpha_{i}\right\} _{i=1,2,\cdots,k} \rangle \times \mathbb{Z}$.
Thus, the cocycles in $Z_{\phi}^{1}(G;\mathbb{Z})$ can be labeled by $\mu \in \langle \left\{ \alpha_{i}\right\} _{i=1,2,\cdots,k} \rangle $ and $\nu^{1}(r)\in \mathbb{Z}$.

From definitions (\ref{Z1/B1}-\ref{B1}) and Eq.~(\ref{coboundaries}),
we further see that $\nu^{1}(r)$ is defined modulo 2 when the cohomology
class $[\nu^{1}]\in H_{\phi}^{1}(G;\ZZZ)$ is concerned. 
Thus, the proof of Theorem \ref{thm:H1_formula} is completed.

Now we can use Theorem~\ref{thm:H1_formula} to obtained the missed factor $H_{\phi}^{1}(G;\ZZZ)$ (due to nontrivial $E_{8}$ state configurations)
in the classification of 3D
bosonic cSPT phases for all 230 space groups. The results
have been tabulated in Table~\ref{table:H1}. We have juxtaposed these
results with the classification of non-$E_{8}$-based cSPT phases
previously obtained in Ref.~\citep{Huang_dimensional_reduction}.
The product of the two gives the complete classification of 3D bosonic
cSPT phases.

\subsection{Point groups and magnetic groups (3D)\label{subsec:H1_mg}}

The above computation will becomes much simpler when we only considers the 32 point groups (3D). 
%For any point group $G$, let $\lambda=0$ if $G$ preserves orientation and $\lambda=1$ otherwise. Then  
Due to the absence of translation symmetry, the factor $\mathbb{Z}^{k}$ is not needed. 
The final result is 
$ H_{\phi}^{1}(G;\mathbb{Z})$ is trivial if the point group $G$ preserves orientation and equals $\mathbb{Z}_{2}$ otherwise; it is also gives the extra factor describing the $E_{8}$ state configuration missed in the classification of 3D bosonic SPT phases protected point group symmetry in Ref.~\cite{cSPT2017}. 
Actually, this result of $ H_{\phi}^{1}(G;\mathbb{Z})$ holds even for a magnetic point group $G$.

Furthermore, $H_{\phi}^{1}(G;\mathbb{Z})$ can also be computed as in Appendix~\ref{subsec:H1_sg} for the 1651 magnetic space groups; only  $k=\dim\bigcap_{W\in \wp G}\left(\phi\left(W\right)W-I\right)$ is computed in a slight different way,
where $\wp G$ is the magnetic point group associated with $G$. To proceed, let us first review the basic structure of magnetic space groups.

Let $\mathfrak{\mathbb{E}}^{d}$ be the $d$-dimensional Euclidean
space and $\mathcal{I}\left(\mathbb{\mathbb{E}}^{d}\right)$ be the
group of its isometries. A subgroup $G$ of $\mathcal{I}\left(\mathbb{E}^{d}\right)\cup\mathcal{I}\left(\mathbb{E}^{d}\right)\mathcal{T}$
is called a magnetic space group if its translation subgroup is generated
by translations along $d$ linearly independent vectors $v_{j}$,
$j=1,2,\cdots,d$, where $\mathcal{T}$ is the time reversion and
$\mathcal{I}\left(\mathbb{E}^{d}\right)\mathcal{T}$ collects
the operations of an space isometry combined with $\mathcal{T}$.
Elements of $\mathcal{I}\left(\mathbb{E}^{d}\right)$ (resp. $\mathcal{I}\left(\mathbb{E}^{d}\right)\mathcal{T}$)
are called \emph{unitary }(resp.\emph{ anti-unitary}); they can be
presented by $\left(W,w\right)$ with $W=Q$ (resp. $W=Q'$) in a
cartesian coordinate system of $\mathbb{E}^{d}$, where $Q\in O\left(d\right)$
(\emph{i.e.}, is a $d\times d$ orthogonal matrix), $w\in\mathbb{R}^{d}$,
and $Q'\coloneqq Q\mathcal{T}$ denotes the combined operation of
rotation $Q$ and time reversion $\mathcal{T}$. Clearly, $\left(Q',w\right)=\left(Q,w\right)\mathcal{T}$
and a point labeled by $u\in\mathbb{R}^{3}$ is moved to $Wu+w$ by
$\left(W,w\right)$ with $Wu$ denoting $Qu$ for both $W=Q,Q'$.
Let $G^{+}$ (resp. $G^{-}$) collect all the unitary (resp. anti-unitary)
operations in $G$; it is clear that $G^{+}$ is a normal subgroup
of $G$. There is a well-defined group homomorphism $\wp:G\rightarrow O\left(d\right)\cup O\left(d\right)\mathcal{T}$
sending $\left(W,w\right)$ to $W$. The image of $\wp$ (\emph{i.e.},
$\wp G=\wp G^{+}\cup\wp G^{-}$) is the magnetic point group associated
with $G$. In addition, $\beta^{-1}\left(I\right)\subseteq G^{+}$
is the translation subgroup of $G$ and $\beta^{-1}\left(I'\right)\subseteq G^{-}$
contains operations of the form $\left(I',w\right)$, where $I$ is
the identity matrix. 

A generic magnetic space group $G$ must be one of four basic types. 
Type I:
all elements of $G$ are unitary, \emph{i.e.}, $G^{-}=\emptyset$;
in other words, $G$ is just a space group. Type II: $G^{-}=G^{+}\mathcal{T}$;
equivalently, $G$ contains the time reversion $\mathcal{T}$. Type
III: $G^{-}\neq\emptyset$ but $I'\notin\wp G$. Type IV: $I'\in\wp G$
but $\mathcal{T}\notin G$. In three ($d=3$) spatial dimensions,
there are 230 type I, 230 type II, 674 type III, and 517 type IV magnetic
space groups respectively; thus, we have $230+230+674+517=1651$ in total \cite{magneticG}.

The calculation process for $H_{\phi}^{1}\left(G;\mathbb{Z}\right)$ in Appendix~\ref{subsec:H1_sg} can be repeated for magnetic space groups and we get the following result.
\begin{thm}
	\label{thm:H1_magnetic}Let $G$ be a 3D magnetic space group and
	$\phi:G\rightarrow\left\{ \pm1\right\} $ be the homomorphism sending
	elements that preserve the orientation of spacetime to $1$ and the
	rest to $-1$. Then 
	\begin{equation}
	H_{\phi}^{1}\left(G;\mathbb{Z}\right)=\mathbb{Z}^{k}\times\mathbb{Z}_{2}^{\ell},
	\end{equation}
	where $k=\dim\bigcap_{W\in\wp G}\ker\left(\phi\left(W\right)W-I\right)$
	with $\wp G$ denoting the magnetic point group associated with $G$
	and $\ell=0$ (resp. $1$) if there exists no (resp. some) element
	of $G$ that reverses spacetime orientation.
	
	For $G$ of type I, $k=0,1,3$ respectively if, in the international
	(Hermann-Mauguin) symbol, there are more than one symmetry direction
	listed, exactly one symmetry direction listed and it is not $1$ or
	$\bar{1}$, and exactly one symmetry direction listed and it is $1$
	or $\bar{1}$.
	
	For $G$ of type II or type IV, we always have $k=0$ and $\ell=1$,
	i.e., $H_{\phi}^{1}\left(G;\mathbb{Z}\right)=\mathbb{Z}_{2}$.
	
	For $G$ of type III, $k$ and $\ell$ only depend on the magnetic
	point group $\wp G$ as given in Table~\ref{tab:mg}.
\end{thm}

Theorem~\ref{thm:H1_formula} already gives $H_{\phi}^{1}\left(G;\mathbb{Z}\right)$ for type I magnetic space group (\emph{i.e.}, space group) $G$ and the results are tabulated in Table~\ref{table:H1}. 
For type II and type IV magnetic space group $G$, we have $I'\in \wp G$, $\ker(\phi(I')I'-I)=0$ and hence $H_{\phi}^{1}\left(G;\mathbb{Z}\right)=\mathbb{Z}_{2}$, where both $I'$ and $I$ act as the identity map on $\mathbb{R}^3$ but $\phi(I')=-1$ because of the time reversion.

For any magnetic group $G$ of type III, $\wp G$ must be one of the
58 magnetic point groups listed in Table~\ref{tab:mg}, which are
called type III as well. For $W=Q\in O\left(d\right)$, $\ker\left(\phi\left(W\right)W-I\right)=\ker\left(\det\left(Q\right)Q-I\right)$
is the symmetry direction if $Q$ is $n$ (\emph{i.e.}, an $n$-fold
rotation) or $\overline{n}$ (\emph{i.e.}, an $n$-fold rotoinvertion)
with $n>1$. For $W=Q'\in O\left(d\right)\mathcal{T}$, $\ker\left(\phi\left(W\right)W-I\right)=\ker\left(\det\left(Q\right)Q+I\right)$
is the plane perpendicular to the symmetry direction if $Q$ is $2$
or $\overline{2}=m$ (\emph{i.e.}, a mirror reflection); otherwise,
$\ker\left(\phi\left(W\right)W-I\right)=0$. Thus, it is easy to obtain
$k$ as well as $\ell$ tabulated in Table~\ref{tab:mg} from the
international symbols, where primed symmetries are those combined with time reversion.

\section{Stacking of crystalline states}
\label{sec:stacking}

%\charles{This section: might need adjustment depending on notation used in general construction.}

In this appendix, we discuss the notation of stacking operation and establish the Abelian group structure of $\text{SPT}^{d}(G)$ (\emph{i.e.}, the collection of SPT orders protected by symmetry $G$ in $d$ spatial dimensions) with an emphasis on space group symmetries.

\begin{comment}

In Sec.\,\ref{sec:reduction_and_construction} \charles{or Appendix\,\ref{sec:general_construction}?}
we showed how one may construct a crystalline state with symmetry
$G$ given a group 1-cocycle $\nu^{1}:G\fromto\ZZZ$. There is a natural
Abelian group structure on the set of group 1-cocycles, given by 
\begin{equation}
(\nu^{1}+\nu'^{1})(g)\coloneq\nu^{1}(g)+\nu'^{1}(g).
\end{equation}
Here we show that this Abelian group structure has a clear physical
meaning.

Let us define a binary operation on the set of crystalling states
with symmetry $G$, as follows. 
\end{comment}

Given a symmetry group (probably a space group) $G$, 
let $a$ be a system made of spins (or more generic bosonic degrees of freedom) $\sigma_{u}$ located at discrete positions $u\in A \subseteq \mathbb{E}^d$ in $d$ spatial dimensions. Let 
$\mathcal{H}_{a}$ denote its Hilbert space with $G$ represented by $\rho_{a}: G \rightarrow \Aut(\mathcal{H}_{a})$, where $\Aut(\mathcal{H}_{a})$ is the group of unitary and antiunitary\footnote{Antiunitary operators are needed when time reversion is involved.} operators on $\mathcal{H}_{a}$. Clearly, we need $\rho_{a}(g)\sigma_{u}\rho_{a}^{-1}(g)=\sigma_{g.u}$, \emph{i.e.}, $g$ mapping the spin at $u$ to the one at $g.u$. With spins symmetrically arranged in space,
such a system has a $G$-symmetric tensor product state $\left|T_{a}\right\rangle $.

In addition, let $\left|\Psi_{a}\right\rangle $ be a short range correlated state; due to its short range correlation nature, $\left|\Psi_{a}\right\rangle $
can be described as the ground state of some gapped Hamiltonian $\hat{H}_{a}$
with finite range interactions. If $\hat{H}_{a}$ allows no nontrivial lower dimensional excitation (e.g. anyons in a 2D system and loops in a 3D system) in bulk, it is called \emph{short range entangled} (SRE). Clearly, $\left|T_{a}\right\rangle $ is SRE. %Moreover, as we have seen in plenty examples in this paper, for any $G$-symmetric SRE state $\left|\Psi_{a}\right\rangle$, there is a $G$-symmetric SRE state $\left|\overline{\Psi}_{a}\right\rangle$ 
Theory of SPT phases presented in this paper classifies $G$-symmetric SRE states. 

\begin{figure}
	\centering \includegraphics[width=3.3in]{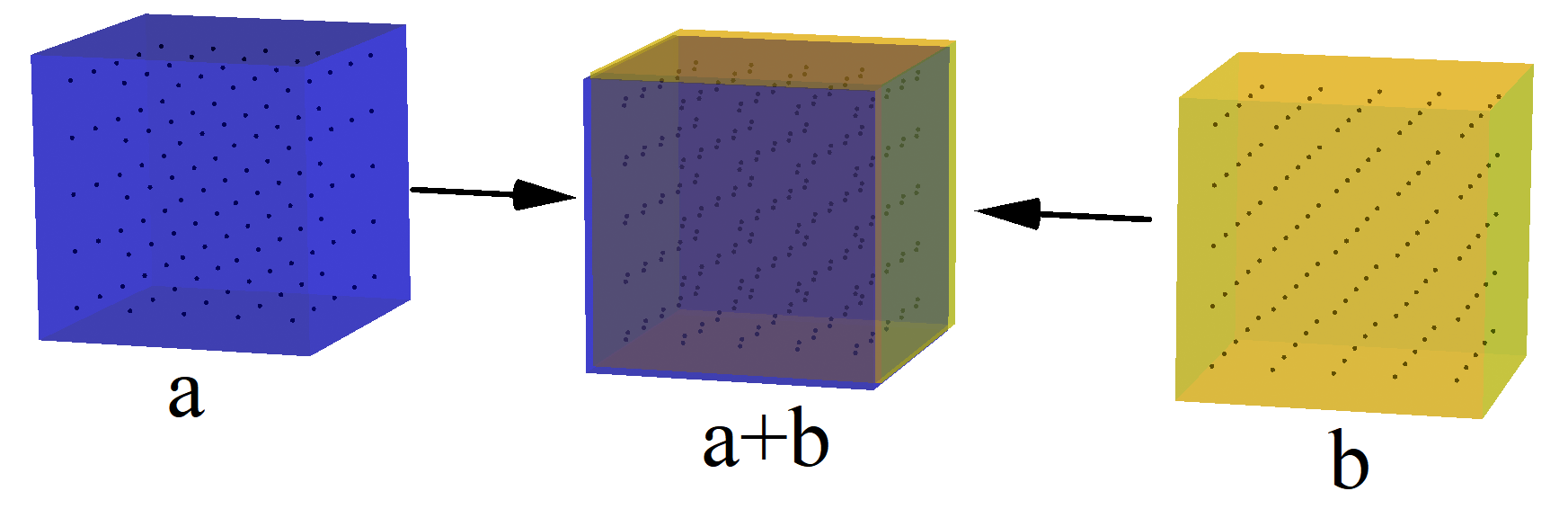} \caption{The stacking operation combines two crystalline states with symmetry
		$G$ in dimension $d$ into a new crystalline state with symmetry
		$G$ in dimension $d$.}
	\label{fig:stacking} 
\end{figure}

Given another system $b$ of spins $\tau_{u}$ at $u\in B \subseteq \mathbb{E}^{d}$ with symmetry $G$ represented by $\rho_{b}$ on its Hilbert space $\mathcal{H}_{b}$, we can stack $a$ and $b$ together as illustrated in Fig.~\ref{fig:stacking} to obtain a new system, whose local degrees of freedom are given by $\sigma_{u}\otimes\tau_{u}$ for $u\in A\cup B$. If $u\notin A$ (resp. $u\notin B$), then $\sigma_{u}$ (resp. $\tau_{u}$) is trivial. The resulting Hilbert space is $\mathcal{H}_{a}\otimes\mathcal{H}_{b}$ with $G$ represented as $\rho_{a}\otimes \rho_{b}$. Clearly, if both $\left|\Psi_{a}\right\rangle$ and $\left|\Psi_{b}\right\rangle$ are SRE (resp. $G$-symmetric), then
$\left|\Psi_{a}\right\rangle \otimes\left|\Psi_{b}\right\rangle $ is also SRE (resp. $G$-symmetric). 

To see that the stacking operation defines an Abelian group structure of $\text{SPT}^{d}(G)$, let us clarify the notion of SPT (including cSPT) orders first. For a specific system like $a$,  two $G$-symmetric SRE states $\left|\Psi_{a}\right\rangle, \left|\Psi_{a}^{\prime}\right\rangle \in \mathcal{H}_{a}$ are considered to lie in the same $G$-SPT phase and we write $\left|\Psi_{a}\right\rangle \simeq_{G} \left|\Psi_{a}^{\prime}\right\rangle$ if they are connected by a continuous path of $G$-symmetric SRE states. Further, to establish a universal notion of SPT orders regardless of specific systems, we would like to compare $G$-symmetric SRE states 
$\left|\Psi_{a}\right\rangle \in \mathcal{H}_{a}$ and $\left|\Psi_{b}\right\rangle \in \mathcal{H}_{b}$
of different systems. For this, we need to choose\footnote{The choice may not be unique; as we have seen in Fig.~\ref{fig:Pm_SPT}, distinct cSPT phases may be constructed with only 0D symmetry charges and hence can be presented by a tensor product state.} a $G$-symmetric tensor product state $\left|T_{a}\right\rangle $ (resp.  $\left|T_{b}\right\rangle $) as the reference for system $a$ (resp. $b$) and call it the trivial state. We say that $\left|\Psi_{a}\right\rangle$ and $\left|\Psi_{b}\right\rangle$ have the same $G$-SPT order (compared to $\left|T_{a}\right\rangle $ and $\left|T_{a}\right\rangle $ respectively) and write $\left(\Psi_{a},T_{a}\right)\simeq_{G}\left(\Psi_{b},T_{b}\right)$ if $\left|\Psi_{a}\right\rangle \otimes\left|T_{b}\right\rangle ,\left|T_{a}\right\rangle \otimes\left|\Psi_{b}\right\rangle \in\mathcal{H}_{a}\otimes\mathcal{H}_{b}$
lie in the same $G$-SPT phase. 
This is an equivalence relation: it is clearly reflective ($\left(\Psi,T\right)\simeq_{G}\left(\Psi,T\right)$) and symmetric ($\left(\Psi_{a},T_{a}\right)\simeq_{G}\left(\Psi_{b},T_{b}\right)\Longleftrightarrow\left(\Psi_{b},T_{b}\right)\simeq_{G}\left(\Psi_{a},T_{a}\right)$) and its transitive property ($\left(\Psi_{a},T_{a}\right)\simeq_{G}\left(\Psi_{b},T_{b}\right)$ together with $\left(\Psi_{b},T_{b}\right)\simeq_{G}\left(\Psi_{c},T_{c}\right)$
implies $\left(\Psi_{a},T_{a}\right)\simeq_{G}\left(\Psi_{c},T_{c}\right)$) is checked by noticing that $G$-SPT phases are stable against extra trivial degrees of freedom (\emph{i.e.}, for any $G$-symmetric state $\left|T\right\rangle$,
$\left|\Psi\right\rangle \simeq_{G}\left|\Psi'\right\rangle \Longleftrightarrow\left|\Psi\right\rangle \otimes\left|T\right\rangle \simeq_{G}\left|\Psi'\right\rangle \otimes\left|T\right\rangle $).
As an equivalence class,
an $G$-SPT order represented by $G$-symmetric SRE states $\left(\Psi,T\right)$ of some system is denoted $\left[\Psi,T\right]_{G}$ (or simply $\left[\Psi\right]_{G}$ if there is a unique $G$-symmetric tensor product state). We denote
the collection of all possible $G$-SPT orders in $d$ spatial dimensions by  $\text{SPT}^{d}(G)$. For $d=0$, $\text{SPT}^{d}(G)$ is actually the set of symmetries charges: $\left[\Psi_{a},T_{a}\right]_{G}$ corresponds to the symmetry charge of $\left|\Psi_{a}\right\rangle $ relative to $\left|T_{a}\right\rangle $.

Any two  $\left[\Psi_{a},T_{a}\right]_{G},\left[\Psi_{b},T_{b}\right]_{G}\in\text{SPT}^{d}\left(G\right)$ can be added by stacking: their sum $\left[\Psi_{a},T_{a}\right]_{G}+\left[\Psi_{b},T_{b}\right]_{G}$ is defined
as $\left[\Psi_{a}\otimes\Psi_{b},T_{a}\otimes T_{b}\right]_{G}$, where $a$ and $b$ may or may not be the same system. 
It is straightforward to check that this binary operation is associative and commutative. %\emph{i.e.}, $\left[\Psi\otimes\Psi',T\otimes T'\right]_{G} =\left[\Psi'\otimes\Psi,T'\otimes T\right]_{G}$.  
Clearly, $\left[T_{a},T_{a}\right]_{G}=\left[T_{b},T_{b}\right]_{G}$ is the identity and denoted 0 with respect to this addition.
Moreover, as in the plenty of examples we have seen in this paper, there is an inverse $\left|\Psi_{b}\right\rangle $ for each $G$-symmetric SRE state $\left|\Psi_{a}\right\rangle $ satisfying $\left|\Psi_{a}\right\rangle \otimes \left|\Psi_{b}\right\rangle \simeq_{G}\left|T_{a}\right\rangle \otimes\left|T_{b}\right\rangle$, which implies $\left[\Psi_{a},T_{a}\right]_{G}+\left[\Psi_{b},T_{b}\right]_{G}=0$. Thus, $G$-SPT is an Abelian group.

In practice, we could always pick a tensor product state with all spins transforming trivially under site symmetry as the preferred trivial state. Explicitly, a $G$-SPT order is simply represented by a $G$-symmetry model $\mathcal{M}_{a}$, which can be encoded by the distribution of spins (\emph{i.e.}, bosonic local degrees of freedom) $\sigma_{a}(u)$ in space, a representation $\rho_{a}$ of $G$ on its Hilbert space $\mathcal{H}_{a}$, and a Hamiltonian $\hat{H}_{a}$ fixing a $G$-symmetric SRE ground state $\left|\Psi_{a}\right\rangle $. The trivial state $\left|T_{a}\right\rangle $ is chosen by default to be one with all spins transform trivially the action $\rho_{a}$ of their site symmetry. 
For any two $G$-SPT models $\mathcal{M}_a = (\sigma_{a}(u), \mathcal{H}_{a}, \rho_{a}, \hat{H}_{a}, \left|\Psi_{a}\right\rangle)$ and $\mathcal{M}_b = (\sigma_{b}(u), \mathcal{H}_{b}, \rho_{b}, \hat{H}_{b}, \left|\Psi_{b}\right\rangle)$, we simply write $\mathcal{M}_{a}+\mathcal{M}_{b}$ for the $G$-SPT order represented by the model $(\sigma_{a}(u)\otimes\sigma_{b}(u), \mathcal{H}_{a}\otimes\mathcal{H}_{b}, \rho_{a}\otimes\rho_{b}, \hat{H}_{a}\otimes\hat{\mathbb{I}}_{b}+\hat{\mathbb{I}}_{a}\otimes\hat{H}_{b}, \left|\Psi_{a}\right\rangle)\otimes\left|\Psi_{b}\right\rangle)$ (with respect to $\left|T_{a}\right\rangle \otimes \left|T_{b}\right\rangle$).

\section{The generalized cohomology hypothesis\label{sec:GCH}}

A generalized cohomology theory \citep{Hatcher,DavisKirk,Adams1,Adams2} $h$ can be represented by an $\Omega$-spectrum
$F_{\bullet}=\{F_{n}\}$, which by definition is a sequence 
\begin{equation}
\ldots,F_{-2},F_{-1},F_{0},F_{1},F_{2},\ldots
\end{equation}
of pointed topological spaces together with pointed homotopy equivalences
\begin{equation}
F_{n}\homotopic\Omega F_{n+1},\label{FOmegaF}
\end{equation}
where $\Omega$ is the loop space functor. In the non-twisted case,
the generalized cohomology theory $h$ outputs an abelian group $h^{n}(X)$
for each given topological space $X$ and integer $n$, according
to 
\begin{equation}
h^{n}(X)\coloneq\brackets{X,\Omega F_{n+1}}.
\end{equation}
Here $[X,Y]$ denotes the set of homotopy classes of maps from $X$
to $Y$. One can define abelian group structure on $h^n(X)$. Namely, given representatives $a, b: X \fromto \Omega F_{n+1}$ of two homotopy classes $[a]$ and $[b]$, one defines $[a]+[b]$ to be the homotopy class represented by the map $a+b: X \fromto \Omega F_{n+1}$, where $(a+b)(x)$ is the concatenation of the loops $a(x)$ and $b(x)$ for any $x$.

In the twisted case, one is given an integer $n$, a pointed topological
space $X$, and an action $\phi_{X}$ of the fundamental group $\pi_{1}(X)$
on the $\Omega$-spectrum $F$. The generalized cohomology theory
$h$ then outputs an abelian group according to 
\begin{equation}
h^{n}\paren{X,\phi_{X}}\coloneq\brackets{\widetilde{X},\Omega F_{n+1}}_{\pi_{1}(X)},\label{widetilde_X_Omega_F_n+1}
\end{equation}
where $\widetilde{X}$ denotes the universal cover of $X$, and $[X,Y]_{G}$
denotes the set of homotopy classes of $G$-equivariant maps from
$X$ to $Y$. As in the non-twisted case, $h^{n}\paren{X,\phi_{X}}$ can be given an abelian group structure by concatenating loops. Here, recall that a
$G$-equivariant map $f:X\fromto Y$ is a map that commutes with the
action of $G$, i.e., $g.(f(x))=f(g.x)~\forall g\in G$ and $x\in X$.
It is a simple exercise to show that if $\pi_{1}(X)$ acts trivially
on the $\Omega$-spectrum, then $h^{n}(X,\phi_{X})=h^{n}(X)$.

In Appendix\,\ref{sec:stacking}, we noted that for given dimension $d$, symmetry $G$, and homomorphism $\phi: G \fromto \braces{\pm 1}$ indicating whether a group element reverses the spacetime orientation, the set of bosonic SPT phases
\begin{equation}
\text{SPT}^d(G)
\end{equation}
has a natural abelian group structure defined by stacking. Based on Kitaev's argument \cite{Kitaev_Stony_Brook_2011_SRE_1, Kitaev_Stony_Brook_2013_SRE, Kitaev_IPAM} that the classification of SPT phases should carry the structure of a generalized cohomology theory, Refs.\,\cite{Xiong, Xiong_Alexandradinata} formulated the following ``generalized cohomology hypothesis.''

\begin{framednameddef}[Generalized cohomology hypothesis]
There exists a generalized cohomology theory $h$
such that for any $d\in\NNN$, group $G$, and continuous homomorphism $\phi: G \fromto \braces{\pm 1}$, we have an abelian group isomorphism
\begin{equation}
\text{SPT}^d\paren{G}\isomorphic h^d\paren{BG,\phi}\label{GCH_iso}
\end{equation}
between the classification of bosonic SPT phases and the generalized cohomology group.
\end{framednameddef}
 
\noindent (The formulation of the generalized cohomology hypothesis in Refs.\,\cite{Xiong, Xiong_Alexandradinata} was actually stronger than as stated here and incorporated an additional structure called functoriality. There was also a fermionic counterpart to the hypothesis \cite{Xiong_Alexandradinata}.)

In the right-hand side of isomorphism (\ref{GCH_iso}), the action of $\pi_1\paren{BG} \isomorphic \pi_0\paren{G}$ on the $\Omega$-spectrum $F_{\bullet}$ that represents $h$ is as follows. First, note that by continuity $\phi: G \fromto \braces{\pm 1}$ can be viewed as a homomorphism from $\pi_0\paren{G}$ to $\braces{\pm 1}$ instead. Then, an element $x$ of $\pi_1\paren{BG} \isomorphic \pi_0\paren{G}$ acts on the $\Omega F_{n+1}$ in Eq.\,(\ref{widetilde_X_Omega_F_n+1}) by reversing the loop if $\phi(x) = -1$ and by the identity if $\phi(x) = +1$. In other words, we have the following

\begin{framednameddef}[Addendum to hypothesis]
A symmetry element behaves antiunitarily in the generalized cohomology hypothesis if it reverses the orientation of spacetime, and unitarily otherwise.
\end{framednameddef}

In the main text, $h^d\paren{BG,\phi}$ is  explicitly defined by Eq.~\eqref{hypothesis_expression} and denoted as $h_{\phi}^d\paren{G;F_{\bullet}}$ in a way analogous to $H_{\phi}^{n}(G;\mathbb{Z})$.

\renewcommand*\arraystretch{1}

\begin{table}
	\caption{The classification of 3D bosonic crystalline SPT (cSPT) phases for all 230 space groups. The first and second columns list the numbers and international (Hermann-Mauguin) symbols of space groups. The third through fifth columns list the classification of non-$E_8$-based phases (given by $H^5_\phi(G;\ZZZ)$ and cited from Ref.~\cite{Huang_dimensional_reduction}), $E_8$ state configurations (given by $H^1_\phi(G;\ZZZ)$ and computed using Theorem \ref{thm:H1_formula}), and all 3D bosonic cSPT phases $\text{SPT}^{3}(G) \cong H^5_\phi(G;\ZZZ)\oplus H^1_\phi(G;\ZZZ)
		$ respectively. Entries left blank have trivial classifications.}
	\label{table:H1}
	\begin{tabular}{ccccc}
		\hline
		\hline
		\multirow{2}{*}{No.} & \multirow{2}{*}{~Symbol~} & \multicolumn{3}{c}{Classification of 3D bosonic cSPT phases} \\
		\cline{3-5}
		&                         & ~$H_{\phi}^{5}\left(G;\mathbb{Z}\right)$~ & ~$H_{\phi}^{1}\left(G;\mathbb{Z}\right)$~ & ~$\text{SPT}^{3}(G)$~ \\
		\hline
		1     & $P1$  &       & {$\mathbb{Z}^{3}$} & $\mathbb{Z}^{3}$ \\
		2     & $P\overline{1}$ & $\mathbb{Z}_2^{8}$ & {$\mathbb{Z}_2\times\mathbb{Z}^{3}$} & $\mathbb{Z}_2^{9}\times\mathbb{Z}^{3}$ \\
		3     & $P2$  & $\mathbb{Z}_2^{4}$ & {$\mathbb{Z}$} & $\mathbb{Z}_2^{4}\times\mathbb{Z}$ \\
		4     & $P2_{1}$ &       & {$\mathbb{Z}$} & $\mathbb{Z}$ \\
		5     & $C2$  & $\mathbb{Z}_2^{2}$ & {$\mathbb{Z}$} & $\mathbb{Z}_2^{2}\times\mathbb{Z}$ \\
		6     & $Pm$  & $\mathbb{Z}_2^{4}$ & {$\mathbb{Z}_2\times\mathbb{Z}$} & $\mathbb{Z}_2^{5}\times\mathbb{Z}$ \\
		7     & $Pc$  &       & {$\mathbb{Z}_2\times\mathbb{Z}$} & $\mathbb{Z}_2\times\mathbb{Z}$ \\
		8     & $Cm$  & $\mathbb{Z}_2^2$ & {$\mathbb{Z}_2\times\mathbb{Z}$} & $\mathbb{Z}_2^3\times\mathbb{Z}$ \\
		9     & $Cc$  &       & {$\mathbb{Z}_2\times\mathbb{Z}$} & $\mathbb{Z}_2\times\mathbb{Z}$ \\
		10    & $P2/m$ & $\mathbb{Z}_2^{18}$ & {$\mathbb{Z}_2\times\mathbb{Z}$} & $\mathbb{Z}_2^{19}\times\mathbb{Z}$ \\
		11    & $P2_{1}/m$ & $\mathbb{Z}_2^6$ & {$\mathbb{Z}_2\times\mathbb{Z}$} & $\mathbb{Z}_2^{7}\times\mathbb{Z}$ \\
		12    & $C2/m$ & $\mathbb{Z}_2^{11}$ & {$\mathbb{Z}_2\times\mathbb{Z}$} & $\mathbb{Z}_2^{12}\times\mathbb{Z}$ \\
		13    & $P2/c$ & $\mathbb{Z}_2^{6}$ & {$\mathbb{Z}_2\times\mathbb{Z}$} & $\mathbb{Z}_2^{7}\times\mathbb{Z}$ \\
		14    & $P2_{1}/c$ & $\mathbb{Z}_2^{4}$ & {$\mathbb{Z}_2\times\mathbb{Z}$} & $\mathbb{Z}_2^{5}\times\mathbb{Z}$ \\
		15    & $C2/c$ & $\mathbb{Z}_2^{5}$ & {$\mathbb{Z}_2\times\mathbb{Z}$} & $\mathbb{Z}_2^{6}\times\mathbb{Z}$ \\
		16    & $P222$ & $\mathbb{Z}_2^{16}$ &       & $\mathbb{Z}_2^{16}$ \\
		17    & $P222_{1}$ & $\mathbb{Z}_2^{4}$ &       & $\mathbb{Z}_2^{4}$ \\
		18    & $P2_{1}2_{1}2$ & $\mathbb{Z}_2^{2}$ &       & $\mathbb{Z}_2^{2}$ \\
		19    & $P2_{1}2_{1}2_{1}$ &       &       &  \\
		20    & $C222_{1}$ & $\mathbb{Z}_2^{2}$ &       & $\mathbb{Z}_2^{2}$ \\
		21    & $C222$ & $\mathbb{Z}_2^{9}$ &       & $\mathbb{Z}_2^{9}$ \\
		22    & $F222$ & $\mathbb{Z}_2^{8}$ &       & $\mathbb{Z}_2^{8}$ \\
		23    & $I222$ & $\mathbb{Z}_2^{8}$ &       & $\mathbb{Z}_2^{8}$ \\
		24    & $I2_{1}2_{1}2_{1}$ & $\mathbb{Z}_2^{3}$ &       & $\mathbb{Z}_2^{3}$ \\
		25    & $Pmm2$ & $\mathbb{Z}_2^{16}$ & $\mathbb{Z}_2$      & $\mathbb{Z}_2^{17}$ \\
		26    & $Pmc2_{1}$ & $\mathbb{Z}_2^{4}$ & $\mathbb{Z}_2$      & $\mathbb{Z}_2^{5}$ \\
		27    & $Pcc2$ & $\mathbb{Z}_2^{4}$ & $\mathbb{Z}_2$      & $\mathbb{Z}_2^{5}$ \\
		28    & $Pma2$ & $\mathbb{Z}_2^{4}$ & $\mathbb{Z}_2$      & $\mathbb{Z}_2^{5}$ \\
		29    & $Pca2_{1}$ &       & $\mathbb{Z}_2$      & $\mathbb{Z}_2$ \\
		30    & $Pnc2$ & $\mathbb{Z}_2^{2}$ & $\mathbb{Z}_2$      & $\mathbb{Z}_2^{3}$ \\
		31    & $Pmn2_{1}$ & $\mathbb{Z}_2^2$ & $\mathbb{Z}_2$      & $\mathbb{Z}_2^3$ \\
		32    & $Pba2$ & $\mathbb{Z}_2^{2}$ & $\mathbb{Z}_2$      & $\mathbb{Z}_2^{3}$ \\
		33    & $Pna2_{1}$ &       & $\mathbb{Z}_2$      & $\mathbb{Z}_2$ \\
		34    & $Pnn2$ & $\mathbb{Z}_2^{2}$ & $\mathbb{Z}_2$      & $\mathbb{Z}_2^{3}$ \\
		35    & $Cmm2$ & $\mathbb{Z}_2^{9}$ & $\mathbb{Z}_2$      & $\mathbb{Z}_2^{10}$ \\
		36    & $Cmc2_{1}$ & $\mathbb{Z}_2^2$ & $\mathbb{Z}_2$      & $\mathbb{Z}_2^3$ \\
		37    & $Ccc2$ & $\mathbb{Z}_2^{3}$ & $\mathbb{Z}_2$      & $\mathbb{Z}_2^{4}$ \\
		38    & $Amm2$ & $\mathbb{Z}_2^{9}$ & $\mathbb{Z}_2$      & $\mathbb{Z}_2^{10}$ \\
		39    & $Aem2$ & $\mathbb{Z}_2^{4}$ & $\mathbb{Z}_2$      & $\mathbb{Z}_2^{5}$ \\
		40    & $Ama2$ & $\mathbb{Z}_2^{3}$ & $\mathbb{Z}_2$      & $\mathbb{Z}_2^{4}$ \\
		41    & $Aea2$ & $\mathbb{Z}_2$ & $\mathbb{Z}_2$      & $\mathbb{Z}_2^2$ \\
		42    & $Fmm2$ & $\mathbb{Z}_2^{6}$ & $\mathbb{Z}_2$      & $\mathbb{Z}_2^{7}$ \\
		\hline
		\hline
	\end{tabular}
\end{table}
\begin{table}
	\begin{tabular}{ccccc}
		\hline
		\hline
		\multirow{2}{*}{No.} & \multirow{2}{*}{~Symbol~} & \multicolumn{3}{c}{Classification of 3D bosonic cSPT phases} \\
		\cline{3-5}
		&                         & ~$H_{\phi}^{5}\left(G;\mathbb{Z}\right)$~ & ~$H_{\phi}^{1}\left(G;\mathbb{Z}\right)$~ & ~$\text{SPT}^{3}(G)$~ \\
		\hline
		43    & $Fdd2$ & $\mathbb{Z}_2$ & $\mathbb{Z}_2$      & $\mathbb{Z}_2^2$ \\
		44    & $Imm2$ & $\mathbb{Z}_2^{8}$ & $\mathbb{Z}_2$      & $\mathbb{Z}_2^{9}$ \\
		45    & $Iba2$ & $\mathbb{Z}_2^{2}$ &  $\mathbb{Z}_2$     & $\mathbb{Z}_2^{3}$ \\
		46    & $Ima2$ & $\mathbb{Z}_2^{3}$ &  $\mathbb{Z}_2$     & $\mathbb{Z}_2^{4}$ \\
		47    & $Pmmm$ & $\mathbb{Z}_2^{42}$ & $\mathbb{Z}_2$      & $\mathbb{Z}_2^{43}$ \\
		48    & $Pnnn$ & $\mathbb{Z}_2^{10}$ & $\mathbb{Z}_2$      & $\mathbb{Z}_2^{11}$ \\
		49    & $Pccm$ & $\mathbb{Z}_2^{17}$ & $\mathbb{Z}_2$      & $\mathbb{Z}_2^{18}$ \\
		50    & $Pban$ & $\mathbb{Z}_2^{10}$ & $\mathbb{Z}_2$      & $\mathbb{Z}_2^{11}$ \\
		51    & $Pmma$ & $\mathbb{Z}_2^{17}$ & $\mathbb{Z}_2$      & $\mathbb{Z}_2^{18}$ \\
		52    & $Pnna$ & $\mathbb{Z}_2^{4}$ & $\mathbb{Z}_2$      & $\mathbb{Z}_2^{5}$ \\
		53    & $Pmna$ & $\mathbb{Z}_2^{10}$ & $\mathbb{Z}_2$      & $\mathbb{Z}_2^{11}$ \\
		54    & $Pcca$ & $\mathbb{Z}_2^{5}$ & $\mathbb{Z}_2$      & $\mathbb{Z}_2^{6}$ \\
		55    & $Pbam$ & $\mathbb{Z}_2^{10}$ &  $\mathbb{Z}_2$     & $\mathbb{Z}_2^{11}$ \\
		56    & $Pccn$ & $\mathbb{Z}_2^{4}$ &$\mathbb{Z}_2$       & $\mathbb{Z}_2^{5}$ \\
		57    & $Pbcm$ & $\mathbb{Z}_2^{5}$ &  $\mathbb{Z}_2$     & $\mathbb{Z}_2^{6}$ \\
		58    & $Pnnm$ & $\mathbb{Z}_2^{9}$ &  $\mathbb{Z}_2$     & $\mathbb{Z}_2^{10}$ \\
		59    & $Pmmn$ & $\mathbb{Z}_2^{10}$ & $\mathbb{Z}_2$      & $\mathbb{Z}_2^{11}$ \\
		60    & $Pbcn$ & $\mathbb{Z}_2^{3}$ &  $\mathbb{Z}_2$     & $\mathbb{Z}_2^{4}$ \\
		61    & $Pbca$ & $\mathbb{Z}_2^{2}$ &   $\mathbb{Z}_2$    & $\mathbb{Z}_2^{3}$ \\
		62    & $Pnma$ & $\mathbb{Z}_2^{4}$ &   $\mathbb{Z}_2$    & $\mathbb{Z}_2^{5}$ \\
		63    & $Cmcm$ & $\mathbb{Z}_2^{10}$ &  $\mathbb{Z}_2$     & $\mathbb{Z}_2^{11}$ \\
		64    & $Cmca$ & $\mathbb{Z}_2^{7}$ &   $\mathbb{Z}_2$    & $\mathbb{Z}_2^{8}$ \\
		65    & $Cmmm$ & $\mathbb{Z}_2^{26}$ & $\mathbb{Z}_2$      & $\mathbb{Z}_2^{27}$ \\
		66    & $Cccm$ & $\mathbb{Z}_2^{13}$ &  $\mathbb{Z}_2$     & $\mathbb{Z}_2^{14}$ \\
		67    & $Cmme$ & $\mathbb{Z}_2^{17}$ &  $\mathbb{Z}_2$     & $\mathbb{Z}_2^{18}$ \\
		68    & $Ccce$ & $\mathbb{Z}_2^{7}$ &   $\mathbb{Z}_2$    & $\mathbb{Z}_2^{8}$ \\
		69    & $Fmmm$ & $\mathbb{Z}_2^{20}$ &  $\mathbb{Z}_2$     & $\mathbb{Z}_2^{21}$ \\
		70    & $Fddd$ & $\mathbb{Z}_2^{6}$ &   $\mathbb{Z}_2$    & $\mathbb{Z}_2^{7}$ \\
		71    & $Immm$ & $\mathbb{Z}_2^{22}$ & $\mathbb{Z}_2$      & $\mathbb{Z}_2^{23}$ \\
		72    & $Ibam$ & $\mathbb{Z}_2^{10}$ &  $\mathbb{Z}_2$     & $\mathbb{Z}_2^{11}$ \\
		73    & $Ibca$ & $\mathbb{Z}_2^{5}$ & $\mathbb{Z}_2$      & $\mathbb{Z}_2^{6}$ \\
		74    & $Imma$ & $\mathbb{Z}_2^{13}$ & $\mathbb{Z}_2$      & $\mathbb{Z}_2^{14}$ \\
		75    & $P4$  & $\mathbb{Z}_{4}^{2}\times\mathbb{Z}_2$ & {$\mathbb{Z}$} & $\mathbb{Z}_{4}^{2}\times\mathbb{Z}_2\times\mathbb{Z}$ \\
		76    & $P4_{1}$ &       & {$\mathbb{Z}$} & $\mathbb{Z}$ \\
		77    & $P4_{2}$ & $\mathbb{Z}_2^{3}$ & {$\mathbb{Z}$} & $\mathbb{Z}_2^{3}\times{\mathbb{Z}}$ \\
		78    & $P4_{3}$ &       & {$\mathbb{Z}$} & $\mathbb{Z}$ \\
		79    & $I4$  & $\mathbb{Z}_{4}\times\mathbb{Z}_2$ & {$\mathbb{Z}$} & $\mathbb{Z}_{4}\times\mathbb{Z}_2\times\mathbb{Z}$ \\
		80    & $I4_{1}$ & $\mathbb{Z}_2$ & {$\mathbb{Z}$} & $\mathbb{Z}_2\times\mathbb{Z}$ \\
		81    & $P\overline{4}$ & $\mathbb{Z}_{4}^{2}\times\mathbb{Z}_2^{3}$ & {$\mathbb{Z}_2\times\mathbb{Z}$} & $\mathbb{Z}_{4}^{2}\times\mathbb{Z}_2^{4}\times\mathbb{Z}$ \\
		82    & $I\overline{4}$ & $\mathbb{Z}_{4}^{2}\times\mathbb{Z}_2^{2}$ & {$\mathbb{Z}_2\times\mathbb{Z}$} & $\mathbb{Z}_{4}^{2}\times\mathbb{Z}_2^{3}\times\mathbb{Z}$ \\
		83    & $P4/m$ & $\mathbb{\mathbb{Z}}_{4}^{2}\times\mathbb{Z}_2^{12}$ & {$\mathbb{Z}_2\times\mathbb{Z}$} & $\mathbb{\mathbb{Z}}_{4}^{2}\times\mathbb{Z}_2^{13}\times\mathbb{Z}$ \\
		84    & $P4_{2}/m$ & $\mathbb{Z}_2^{11}$ & {$\mathbb{Z}_2\times\mathbb{Z}$} & $\mathbb{Z}_2^{12}\times\mathbb{Z}$ \\
		85    & $P4/n$ & $\mathbb{Z}_{4}^{2}\times\mathbb{Z}_2^{3}$ & {$\mathbb{Z}_2\times\mathbb{Z}$} & $\mathbb{Z}_{4}^{2}\times\mathbb{Z}_2^{4}\times\mathbb{Z}$ \\
		86    & $P4_{2}/n$ & $\mathbb{Z}_{4}\times\mathbb{Z}_2^{4}$ & {$\mathbb{Z}_2\times\mathbb{Z}$} & $\mathbb{Z}_{4}\times\mathbb{Z}_2^{5}\times\mathbb{Z}$ \\
		87    & $I4/m$ & $\mathbb{Z}_{4}\times\mathbb{Z}_2^{8}$ & {$\mathbb{Z}_2\times\mathbb{Z}$} & $\mathbb{Z}_{4}\times\mathbb{Z}_2^{9}\times\mathbb{Z}$ \\
		88    & $I4_{1}/a$ & $\mathbb{Z}_{4}\times\mathbb{Z}_2^{3}$ & {$\mathbb{Z}_2\times\mathbb{Z}$} & $\mathbb{Z}_{4}\times\mathbb{Z}_2^{4}\times\mathbb{Z}$ \\
		89    & $P422$ & $\mathbb{Z}_2^{12}$ &       & $\mathbb{Z}_2^{12}$ \\
		90    & $P42_{1}2$ & $\mathbb{Z}_{4}\times\mathbb{Z}_2^{4}$ &       & $\mathbb{Z}_{4}\times\mathbb{Z}_2^{4}$ \\
		91    & $P4_{1}22$ & $\mathbb{Z}_2^{3}$ &       & $\mathbb{Z}_2^{3}$ \\
		92    & $P4_{1}2_{1}2$ & $\mathbb{Z}_2$ &       & $\mathbb{Z}_2$ \\
		\hline
		\hline
	\end{tabular}
\end{table}
\begin{table}
\begin{tabular}{ccccc}
\hline
\hline
\multirow{2}{*}{No.} & \multirow{2}{*}{~Symbol~} & \multicolumn{3}{c}{Classification of 3D bosonic cSPT phases} \\
\cline{3-5}
&                         & ~$H_{\phi}^{5}\left(G;\mathbb{Z}\right)$~ & ~$H_{\phi}^{1}\left(G;\mathbb{Z}\right)$~ & ~$\text{SPT}^{3}(G)$~ \\
\hline
93    & $P4_{2}22$ & $\mathbb{Z}_2^{12}$ &       & $\mathbb{Z}_2^{12}$ \\
94    & $P4_{2}2_{1}2$ & $\mathbb{Z}_2^{5}$ &       & $\mathbb{Z}_2^{5}$ \\
95    & $P4_{3}22$ & $\mathbb{Z}_2^{3}$ &       & $\mathbb{Z}_2^{3}$ \\
96    & $P4_{3}2_{1}2$ & $\mathbb{Z}_2$ &       & $\mathbb{Z}_2$ \\
97    & $I422$ & $\mathbb{Z}_2^{8}$ &       & $\mathbb{Z}_2^{8}$ \\
98    & $I4_{1}22$ & $\mathbb{Z}_2^{5}$ &       & $\mathbb{Z}_2^{5}$ \\
99    & $P4mm$ & $\mathbb{Z}_2^{12}$ & $\mathbb{Z}_2$      & $\mathbb{Z}_2^{13}$ \\
100   & $P4bm$ & $\mathbb{Z}_{4}\times\mathbb{Z}_2^{4}$ & $\mathbb{Z}_2$      & $\mathbb{Z}_{4}\times\mathbb{Z}_2^{5}$ \\
101   & $P4_{2}cm$ & $\mathbb{Z}_2^{6}$ & $\mathbb{Z}_2$      & $\mathbb{Z}_2^{7}$ \\
102   & $P4_{2}nm$ & $\mathbb{Z}_2^{5}$ &   $\mathbb{Z}_2$    & $\mathbb{Z}_2^{6}$ \\
103   & $P4cc$ & $\mathbb{Z}_2^{3}$ &   $\mathbb{Z}_2$    & $\mathbb{Z}_2^{4}$ \\
104   & $P4nc$ & $\mathbb{Z}_{4}\times\mathbb{Z}_2$ &  $\mathbb{Z}_2$     & $\mathbb{Z}_{4}\times\mathbb{Z}_2^2$ \\
105   & $P4_{2}mc$ & $\mathbb{Z}_2^{9}$ &   $\mathbb{Z}_2$    & $\mathbb{Z}_2^{10}$ \\
106   & $P4_{2}bc$ & $\mathbb{Z}_2^{2}$ &  $\mathbb{Z}_2$     & $\mathbb{Z}_2^{3}$ \\
107   & $I4mm$ & $\mathbb{Z}_2^{7}$ &    $\mathbb{Z}_2$   & $\mathbb{Z}_2^{8}$ \\
108   & $I4cm$ & $\mathbb{Z}_2^{4}$ &   $\mathbb{Z}_2$    & $\mathbb{Z}_2^{5}$ \\
109   & $I4_{1}md$ & $\mathbb{Z}_2^{4}$ &  $\mathbb{Z}_2$     & $\mathbb{Z}_2^{5}$ \\
110   & $I4_{1}cd$ & $\mathbb{Z}_2$ &    $\mathbb{Z}_2$   & $\mathbb{Z}_2^2$ \\
111   & $P\overline{4}2m$ & $\mathbb{Z}_2^{13}$ &  $\mathbb{Z}_2$     & $\mathbb{Z}_2^{14}$ \\
112   & $P\overline{4}2c$ & $\mathbb{Z}_2^{10}$ &  $\mathbb{Z}_2$     & $\mathbb{Z}_2^{11}$ \\
113   & $P\overline{4}2_{1}m$ & $\mathbb{Z}_{4}\times\mathbb{Z}_2^{5}$ &  $\mathbb{Z}_2$     & $\mathbb{Z}_{4}\times\mathbb{Z}_2^{6}$ \\
114   & $P\overline{4}2_{1}c$ & $\mathbb{Z}_{4}\times\mathbb{Z}_2^{2}$ & $\mathbb{Z}_2$      & $\mathbb{Z}_{4}\times\mathbb{Z}_2^{3}$ \\
115   & $P\overline{4}m2$ & $\mathbb{Z}_2^{13}$ & $\mathbb{Z}_2$      & $\mathbb{Z}_2^{14}$ \\
116   & $P\overline{4}c2$ & $\mathbb{Z}_2^{7}$ &    $\mathbb{Z}_2$   & $\mathbb{Z}_2^{8}$ \\
117   & $P\overline{4}b2$ & $\mathbb{Z}_{4}\times\mathbb{Z}_2^{5}$ & $\mathbb{Z}_2$      & $\mathbb{Z}_{4}\times\mathbb{Z}_2^{6}$ \\
118   & $P\overline{4}n2$ & $\mathbb{Z}_{4}\times\mathbb{Z}_2^{5}$ & $\mathbb{Z}_2$      & $\mathbb{Z}_{4}\times\mathbb{Z}_2^{6}$ \\
119   & $I\overline{4}m2$ & $\mathbb{Z}_2^{9}$ & $\mathbb{Z}_2$      & $\mathbb{Z}_2^{10}$ \\
120   & $I\overline{4}c2$ & $\mathbb{Z}_2^{6}$ & $\mathbb{Z}_2$      & $\mathbb{Z}_2^{7}$ \\
121   & $I\overline{4}2m$ & $\mathbb{Z}_2^{8}$ & $\mathbb{Z}_2$      & $\mathbb{Z}_2^{9}$ \\
122   & $I\overline{4}2d$ & $\mathbb{Z}_{4}\times\mathbb{Z}_2^{2}$ & $\mathbb{Z}_2$      & $\mathbb{Z}_{4}\times\mathbb{Z}_2^{3}$ \\
123   & $P4/mmm$ & $\mathbb{Z}_2^{32}$ &  $\mathbb{Z}_2$     & $\mathbb{Z}_2^{33}$ \\
124   & $P4/mcc$ & $\mathbb{Z}_2^{13}$ & $\mathbb{Z}_2$      & $\mathbb{Z}_2^{14}$ \\
125   & $P4/nbm$ & $\mathbb{Z}_2^{13}$ &  $\mathbb{Z}_2$     & $\mathbb{Z}_2^{14}$ \\
126   & $P4/nnc$ & $\mathbb{Z}_2^{8}$ &  $\mathbb{Z}_2$     & $\mathbb{Z}_2^{9}$ \\
127   & $P4/mbm$ & $\mathbb{Z}_{4}\times\mathbb{Z}_2^{15}$ &   $\mathbb{Z}_2$    & $\mathbb{Z}_{4}\times\mathbb{Z}_2^{16}$ \\
128   & $P4/mnc$ & $\mathbb{Z}_{4}\times\mathbb{Z}_2^{8}$ &  $\mathbb{Z}_2$     & $\mathbb{Z}_{4}\times\mathbb{Z}_2^{9}$ \\
129   & $P4/nmm$ & $\mathbb{Z}_2^{13}$ &  $\mathbb{Z}_2$     & $\mathbb{Z}_2^{14}$ \\
130   & $P4/ncc$ & $\mathbb{Z}_2^{5}$ &   $\mathbb{Z}_2$    & $\mathbb{Z}_2^{6}$ \\
131   & $P4_{2}/mmc$ & $\mathbb{Z}_2^{24}$ &   $\mathbb{Z}_2$    & $\mathbb{Z}_2^{25}$ \\
132   & $P4_{2}/mcm$ & $\mathbb{Z}_2^{18}$ &   $\mathbb{Z}_2$    & $\mathbb{Z}_2^{19}$ \\
133   & $P4_{2}/nbc$ & $\mathbb{Z}_2^{8}$ &    $\mathbb{Z}_2$   & $\mathbb{Z}_2^{9}$ \\
134   & $P4_{2}/nnm$ & $\mathbb{Z}_2^{13}$ &  $\mathbb{Z}_2$     & $\mathbb{Z}_2^{14}$ \\
135   & $P4_{2}/mbc$ & $\mathbb{Z}_2^{8}$ &  $\mathbb{Z}_2$     & $\mathbb{Z}_2^{9}$ \\
136   & $P4_{2}/mnm$ & $\mathbb{Z}_2^{14}$ &  $\mathbb{Z}_2$     & $\mathbb{Z}_2^{15}$ \\
137   & $P4_{2}/nmc$ & $\mathbb{Z}_2^{8}$ &   $\mathbb{Z}_2$    & $\mathbb{Z}_2^{9}$ \\
138   & $P4_{2}/ncm$ & $\mathbb{Z}_2^{10}$ &  $\mathbb{Z}_2$     & $\mathbb{Z}_2^{11}$ \\
139   & $I4/mmm$ & $\mathbb{Z}_2^{20}$ &   $\mathbb{Z}_2$    & $\mathbb{Z}_2^{21}$ \\
140   & $I4/mcm$ & $\mathbb{Z}_2^{14}$ &  $\mathbb{Z}_2$     & $\mathbb{Z}_2^{15}$ \\
141   & $I4_{1}/amd$ & $\mathbb{Z}_2^{9}$ &  $\mathbb{Z}_2$     & $\mathbb{Z}_2^{10}$ \\
142   & $I4_{1}/acd$ & $\mathbb{Z}_2^{5}$ & $\mathbb{Z}_2$      & $\mathbb{Z}_2^{6}$ \\
\hline
\hline
\end{tabular}
\end{table}
\begin{table}
\begin{tabular}{ccccc}
\hline
\hline
\multirow{2}{*}{No.} & \multirow{2}{*}{~Symbol~} & \multicolumn{3}{c}{Classification of 3D bosonic cSPT phases} \\
\cline{3-5}
&                         & ~$H_{\phi}^{5}\left(G;\mathbb{Z}\right)$~ & ~$H_{\phi}^{1}\left(G;\mathbb{Z}\right)$~ & ~$\text{SPT}^{3}(G)$~ \\
\hline
143   & $P3$  & $\mathbb{Z}_{3}^{3}$ & {$\mathbb{Z}$} & $\mathbb{Z}_{3}^{3}\times\mathbb{Z}$ \\
144   & $P3_{1}$ &       & {$\mathbb{Z}$} & $\mathbb{Z}$ \\
145   & $P3_{2}$ &       & {$\mathbb{Z}$} & $\mathbb{Z}$ \\
146   & $R3$  & $\mathbb{Z}_{3}$ & {$\mathbb{Z}$} & $\mathbb{Z}_{3}\times\mathbb{Z}$ \\
147   & $P\overline{3}$ & $\mathbb{Z}_{3}^{2}\times\mathbb{Z}_2^{4}$ & {$\mathbb{Z}_2\times\mathbb{Z}$} & $\mathbb{Z}_{3}^{2}\times\mathbb{Z}_2^{5}\times\mathbb{Z}$ \\
148   & $R\overline{3}$ & $\mathbb{Z}_{3}\times\mathbb{Z}_2^{4}$ & {$\mathbb{Z}_2\times\mathbb{Z}$} & $\mathbb{Z}_{3}\times\mathbb{Z}_2^{5}\times\mathbb{Z}$ \\
149   & $P312$ & $\mathbb{Z}_2^{2}$ &       & $\mathbb{Z}_2^{2}$ \\
150   & $P321$ & $\mathbb{Z}_{3}\times\mathbb{Z}_2^{2}$ &       & $\mathbb{Z}_{3}\times\mathbb{Z}_2^{2}$ \\
151   & $P3_{1}12$ & $\mathbb{Z}_2^{2}$ &       & $\mathbb{Z}_2^{2}$ \\
152   & $P3_{1}21$ & $\mathbb{Z}_2^{2}$ &       & $\mathbb{Z}_2^{2}$ \\
153   & $P3_{2}12$ & $\mathbb{Z}_2^{2}$ &       & $\mathbb{Z}_2^{2}$ \\
154   & $P3_{2}21$ & $\mathbb{Z}_2^{2}$ &       & $\mathbb{Z}_2^{2}$ \\
155   & $R32$ & $\mathbb{Z}_2^{2}$ &       & $\mathbb{Z}_2^{2}$ \\
156   & $P3m1$ & $\mathbb{Z}_2^2$ & $\mathbb{Z}_2$      & $\mathbb{Z}_2^3$ \\
157   & $P31m$ & $\mathbb{Z}_{3}\times\mathbb{Z}_2^2$ & $\mathbb{Z}_2$      & $\mathbb{Z}_{3}\times\mathbb{Z}_2^3$ \\
158   & $P3c1$ &       & $\mathbb{Z}_2$      &  $\mathbb{Z}_2$\\
159   & $P31c$ & $\mathbb{Z}_{3}$ & $\mathbb{Z}_2$      & $\mathbb{Z}_{3}\times\mathbb{Z}_2$ \\
160   & $R3m$ & $\mathbb{Z}_2^2$ &  $\mathbb{Z}_2$     & $\mathbb{Z}_2^3$ \\
161   & $R3c$ &       & $\mathbb{Z}_2$      & $\mathbb{Z}_2$ \\
162   & $P\overline{3}1m$ & $\mathbb{Z}_2^{9}$ &  $\mathbb{Z}_2$     & $\mathbb{Z}_2^{10}$ \\
163   & $P\overline{3}1c$ & $\mathbb{Z}_2^{3}$ & $\mathbb{Z}_2$      & $\mathbb{Z}_2^{4}$ \\
164   & $P\overline{3}m1$ & $\mathbb{Z}_2^{9}$ & $\mathbb{Z}_2$      & $\mathbb{Z}_2^{10}$ \\
165   & $P\overline{3}c1$ & $\mathbb{Z}_2^{3}$ & $\mathbb{Z}_2$      & $\mathbb{Z}_2^{4}$ \\
166   & $R\overline{3}m$ & $\mathbb{Z}_2^{9}$ &  $\mathbb{Z}_2$     & $\mathbb{Z}_2^{10}$ \\
167   & $R\overline{3}c$ & $\mathbb{Z}_2^{3}$ &$\mathbb{Z}_2$       & $\mathbb{Z}_2^{4}$ \\
168   & $P6$  & $\mathbb{Z}_{3}^{2}\times\mathbb{Z}_2^{2}$ & {$\mathbb{Z}$} & $\mathbb{Z}_{3}^{2}\times\mathbb{Z}_2^{2}\times\mathbb{Z}$ \\
169   & $P6_{1}$ &       & {$\mathbb{Z}$} & $\mathbb{Z}$ \\
170   & $P6_{5}$ &       & {$\mathbb{Z}$} & $\mathbb{Z}$ \\
171   & $P6_{2}$ & $\mathbb{Z}_2^{2}$ & {$\mathbb{Z}$} & $\mathbb{Z}_2^{2}\times\mathbb{Z}$ \\
172   & $P6_{4}$ & $\mathbb{Z}_2^{2}$ & {$\mathbb{Z}$} & $\mathbb{Z}_2^{2}\times\mathbb{Z}$ \\
173   & $P6_{3}$ & $\mathbb{Z}_{3}^{2}$ & {$\mathbb{Z}$} & $\mathbb{Z}_{3}^{2}\times\mathbb{Z}$ \\
174   & $P\overline{6}$ & $\mathbb{Z}_{3}^{3}\times\mathbb{Z}_2^{4}$ & {$\mathbb{Z}_2\times\mathbb{Z}$} & $\mathbb{Z}_{3}^{3}\times\mathbb{Z}_2^{5}\times\mathbb{Z}$ \\
175   & $P6/m$ & $\mathbb{Z}_{3}^{2}\times\mathbb{Z}_2^{10}$ & {$\mathbb{Z}_2\times\mathbb{Z}$} & $\mathbb{Z}_{3}^{2}\times\mathbb{Z}_2^{11}\times\mathbb{Z}$ \\
176   & $P6_{3}/m$ & $\mathbb{Z}_{3}^{2}\times\mathbb{Z}_2^{4}$ & {$\mathbb{Z}_2\times\mathbb{Z}$} & $\mathbb{Z}_{3}^{2}\times\mathbb{Z}_2^{5}\times\mathbb{Z}$ \\
177   & $P622$ & $\mathbb{Z}_2^{8}$ &       & $\mathbb{Z}_2^{8}$ \\
178   & $P6_{1}22$ & $\mathbb{Z}_2^{2}$ &       & $\mathbb{Z}_2^{2}$ \\
179   & $P6_{5}22$ & $\mathbb{Z}_2^{2}$ &       & $\mathbb{Z}_2^{2}$ \\
180   & $P6_{2}22$ & $\mathbb{Z}_2^{8}$ &       & $\mathbb{Z}_2^{8}$ \\
181   & $P6_{4}22$ & $\mathbb{Z}_2^{8}$ &       & $\mathbb{Z}_2^{8}$ \\
182   & $P6_{3}22$ & $\mathbb{Z}_2^{2}$ &       & $\mathbb{Z}_2^{2}$ \\
183   & $P6mm$ & $\mathbb{Z}_2^{8}$ & $\mathbb{Z}_2$      & $\mathbb{Z}_2^{9}$ \\
184   & $P6cc$ & $\mathbb{Z}_2^{2}$ & $\mathbb{Z}_2$      & $\mathbb{Z}_2^{3}$ \\
185   & $P6_{3}cm$ & $\mathbb{Z}_2^2$ & $\mathbb{Z}_2$      & $\mathbb{Z}_2^3$ \\
186   & $P6_{3}mc$ & $\mathbb{Z}_2^2$ & $\mathbb{Z}_2$      & $\mathbb{Z}_2^3$ \\
187   & $P\overline{6}m2$ & $\mathbb{Z}_2^{9}$ &  $\mathbb{Z}_2$     & $\mathbb{Z}_2^{10}$ \\
188   & $P\overline{6}c2$ & $\mathbb{Z}_2^{3}$ & $\mathbb{Z}_2$      & $\mathbb{Z}_2^{4}$ \\
189   & $P\overline{6}2m$ & $\mathbb{Z}_{3}\times\mathbb{Z}_2^{9}$ & $\mathbb{Z}_2$      & $\mathbb{Z}_{3}\times\mathbb{Z}_2^{10}$ \\
190   & $P\overline{6}2c$ & $\mathbb{Z}_{3}\times\mathbb{Z}_2^{3}$ & $\mathbb{Z}_2$      & $\mathbb{Z}_{3}\times\mathbb{Z}_2^{4}$ \\
191   & $P6/mmm$ & $\mathbb{Z}_2^{22}$ & $\mathbb{Z}_2$      & $\mathbb{Z}_2^{23}$ \\
192   & $P6/mcc$ & $\mathbb{Z}_2^{9}$ & $\mathbb{Z}_2$      & $\mathbb{Z}_2^{10}$ \\
\hline
\hline
\end{tabular}
\end{table}
\begin{table}

\begin{tabular}{ccccc}
\hline
\hline
\multirow{2}{*}{No.} & \multirow{2}{*}{~Symbol~} & \multicolumn{3}{c}{Classification of 3D bosonic cSPT phases} \\
\cline{3-5}
&                         & ~$H_{\phi}^{5}\left(G;\mathbb{Z}\right)$~ & ~$H_{\phi}^{1}\left(G;\mathbb{Z}\right)$~ & ~$\text{SPT}^{3}(G)$~ \\
\hline
193   & $P6_{3}/mcm$ & $\mathbb{Z}_2^{9}$ &  $\mathbb{Z}_2$     & $\mathbb{Z}_2^{10}$ \\
194   & $P6_{3}/mmc$ & $\mathbb{Z}_2^{9}$ &  $\mathbb{Z}_2$     & $\mathbb{Z}_2^{10}$ \\
195   & $P23$ & $\mathbb{Z}_{3}\times\mathbb{Z}_2^{4}$ &       & $\mathbb{Z}_{3}\times\mathbb{Z}_2^{4}$ \\
196   & $F23$ & $\mathbb{Z}_{3}$ &       & $\mathbb{Z}_{3}$ \\
197   & $I23$ & $\mathbb{Z}_{3}\times\mathbb{Z}_2^{2}$ &       & $\mathbb{Z}_{3}\times\mathbb{Z}_2^{2}$ \\
198   & $P2_{1}3$ & $\mathbb{Z}_{3}$ &       & $\mathbb{Z}_{3}$ \\
199   & $I2_{1}3$ & $\mathbb{Z}_{3}\times\mathbb{Z}_2$ &       & $\mathbb{Z}_{3}\times\mathbb{Z}_2$ \\
200   & $Pm\overline{3}$ & $\mathbb{Z}_{3}\times\mathbb{Z}_2^{14}$ & $\mathbb{Z}_2$      & $\mathbb{Z}_{3}\times\mathbb{Z}_2^{15}$ \\
201   & $Pn\overline{3}$ & $\mathbb{Z}_{3}\times\mathbb{Z}_2^{4}$ &  $\mathbb{Z}_2$     & $\mathbb{Z}_{3}\times\mathbb{Z}_2^{5}$ \\
202   & $Fm\overline{3}$ & $\mathbb{Z}_{3}\times\mathbb{Z}_2^{6}$ & $\mathbb{Z}_2$      & $\mathbb{Z}_{3}\times\mathbb{Z}_2^{7}$ \\
203   & $Fd\overline{3}$ & $\mathbb{Z}_{3}\times\mathbb{Z}_2^{2}$ & $\mathbb{Z}_2$      & $\mathbb{Z}_{3}\times\mathbb{Z}_2^{3}$ \\
204   & $Im\overline{3}$ & $\mathbb{Z}_{3}\times\mathbb{Z}_2^8$ &  $\mathbb{Z}_2$     & $\mathbb{Z}_{3}\times\mathbb{Z}_2^9$ \\
205   & $Pa\overline{3}$ & $\mathbb{Z}_{3}\times\mathbb{Z}_2^{2}$ &  $\mathbb{Z}_2$     & $\mathbb{Z}_{3}\times\mathbb{Z}_2^{3}$ \\
206   & $Ia\overline{3}$ & $\mathbb{Z}_{3}\times\mathbb{Z}_2^{3}$ &  $\mathbb{Z}_2$     & $\mathbb{Z}_{3}\times\mathbb{Z}_2^{4}$ \\
207   & $P432$ & $\mathbb{Z}_2^{6}$ &       & $\mathbb{Z}_2^{6}$ \\
208   & $P4_{2}32$ & $\mathbb{Z}_2^{6}$ &       & $\mathbb{Z}_2^{6}$ \\
209   & $F432$ & $\mathbb{Z}_2^{4}$ &       & $\mathbb{Z}_2^{4}$ \\
210   & $F4_{1}32$ & $\mathbb{Z}_2$ &       & $\mathbb{Z}_2$ \\
211   & $I432$ & $\mathbb{Z}_2^{5}$ &       & $\mathbb{Z}_2^{5}$ \\
212   & $P4_{3}32$ & $\mathbb{Z}_2$ &       & $\mathbb{Z}_2$ \\
213   & $P4_{1}32$ & $\mathbb{Z}_2$ &       & $\mathbb{Z}_2$ \\
214   & $I4_{1}32$ & $\mathbb{Z}_2^{4}$ &       & $\mathbb{Z}_2^{4}$ \\
215   & $P\overline{4}3m$ & $\mathbb{Z}_2^{7}$ &  $\mathbb{Z}_2$     & $\mathbb{Z}_2^8$ \\
216   & $F\overline{4}3m$ & $\mathbb{Z}_2^5$ & $\mathbb{Z}_2$      & $\mathbb{Z}_2^6$ \\
217   & $I\overline{4}3m$ & $\mathbb{Z}_2^5$ & $\mathbb{Z}_2$      & $\mathbb{Z}_2^6$ \\
218   & $P\overline{4}3n$ & $\mathbb{Z}_2^{4}$ & $\mathbb{Z}_2$      & $\mathbb{Z}_2^{5}$ \\
219   & $F\overline{4}3c$ & $\mathbb{Z}_2^{2}$ &   $\mathbb{Z}_2$    & $\mathbb{Z}_2^{3}$ \\
220   & $I\overline{4}3d$ & $\mathbb{Z}_{4}\times\mathbb{Z}_2$ & $\mathbb{Z}_2$      & $\mathbb{Z}_{4}\times\mathbb{Z}_2^2$ \\
221   & $Pm\overline{3}m$ & $\mathbb{Z}_2^{18}$ & $\mathbb{Z}_2$      & $\mathbb{Z}_2^{19}$ \\
222   & $Pn\overline{3}n$ & $\mathbb{Z}_2^{5}$ &   $\mathbb{Z}_2$    & $\mathbb{Z}_2^{6}$ \\
223   & $Pm\overline{3}n$ & $\mathbb{Z}_2^{10}$ &   $\mathbb{Z}_2$    & $\mathbb{Z}_2^{11}$ \\
224   & $Pn\overline{3}m$ & $\mathbb{Z}_2^{10}$ & $\mathbb{Z}_2$      & $\mathbb{Z}_2^{11}$ \\
225   & $Fm\overline{3}m$ & $\mathbb{Z}_2^{13}$ &  $\mathbb{Z}_2$     & $\mathbb{Z}_2^{14}$ \\
226   & $Fm\overline{3}c$ & $\mathbb{Z}_2^{7}$ & $\mathbb{Z}_2$      & $\mathbb{Z}_2^8$ \\
227   & $Fd\overline{3}m$ & $\mathbb{Z}_2^7$ &   $\mathbb{Z}_2$    & $\mathbb{Z}_2^8$ \\
228   & $Fd\overline{3}c$ & $\mathbb{Z}_2^{3}$ &   $\mathbb{Z}_2$    & $\mathbb{Z}_2^{4}$ \\
229   & $Im\overline{3}m$ & $\mathbb{Z}_2^{13}$ &  $\mathbb{Z}_2$     & $\mathbb{Z}_2^{14}$ \\
230   & $Ia\overline{3}d$ & $\mathbb{Z}_2^{4}$ &    $\mathbb{Z}_2$   & $\mathbb{Z}_2^{5}$ \\
\hline
\hline
\end{tabular}
\end{table}

\renewcommand*\arraystretch{1.5}
\begin{table}
\caption{List of 58 magnetic point groups of type III with its symbol in the
second column. For each magnetic group $G$, the third column gives
$k=\dim\bigcap_{W\in G}\left(\phi\left(W\right)W-I\right)$. The forth
column ($\ell$) shows 0 if it preserves the orientation of spacetime
and 1 otherwise.\label{tab:mg}}
\begin{tabular}{|c|c|c|c|c|c|c|c|c|}
\cline{1-4} \cline{6-9} 
No. & Symbol & \enskip{}$k$\enskip{} & \enskip{}$\ell$\enskip{} & \quad{} & No. & Symbol & \enskip{}$k$\enspace{} & \enskip{}$\ell$\enspace{}\tabularnewline
\cline{1-4} \cline{6-9} 
1 & $\overline{1}'$ & 0 & 0 &  & 30 & $\overline{3}'$ & 0 & 0\tabularnewline
\cline{1-4} \cline{6-9} 
2 & $2'$ & 2 & 1 &  & 31 & $32'$ & 1 & 1\tabularnewline
\cline{1-4} \cline{6-9} 
3 & $m'$ & 2 & 0 &  & 32 & $3m'$ & 1 & 0\tabularnewline
\cline{1-4} \cline{6-9} 
4 & $\frac{2'}{m}$ & 0 & 1 &  & 33 & $\overline{3}'m$ & 0 & 1\tabularnewline
\cline{1-4} \cline{6-9} 
5 & $\frac{2}{m'}$ & 0 & 0 &  & 34 & $\overline{3}'m'$ & 0 & 0\tabularnewline
\cline{1-4} \cline{6-9} 
6 & $\frac{2'}{m'}$ & 2 & 1 &  & 35 & $\overline{3}m'$ & 1 & 1\tabularnewline
\cline{1-4} \cline{6-9} 
7 & $2'2'2$ & 1 & 1 &  & 36 & $6'$ & 0 & 1\tabularnewline
\cline{1-4} \cline{6-9} 
8 & $m'm2'$ & 1 & 1 &  & 37 & $\overline{6}'$ & 0 & 0\tabularnewline
\cline{1-4} \cline{6-9} 
9 & $m'm'2$ & 1 & 0 &  & 38 & $\frac{6'}{m}$ & 0 & 1\tabularnewline
\cline{1-4} \cline{6-9} 
10 & $m'mm$ & 0 & 1 &  & 39 & $\frac{6}{m'}$ & 0 & 0\tabularnewline
\cline{1-4} \cline{6-9} 
11 & $m'm'm$ & 1 & 1 &  & 40 & $\frac{6'}{m'}$ & 0 & 1\tabularnewline
\cline{1-4} \cline{6-9} 
12 & $m'm'm'$ & 0 & 0 &  & 41 & $6'22'$ & 0 & 1\tabularnewline
\cline{1-4} \cline{6-9} 
13 & $4'$ & 0 & 1 &  & 42 & $62'2'$ & 1 & 1\tabularnewline
\cline{1-4} \cline{6-9} 
14 & $\overline{4}'$ & 0 & 0 &  & 43 & $6'mm'$ & 0 & 1\tabularnewline
\cline{1-4} \cline{6-9} 
15 & $\frac{4'}{m}$ & 0 & 1 &  & 44 & $6m'm'$ & 1 & 0\tabularnewline
\cline{1-4} \cline{6-9} 
16 & $\frac{4}{m'}$ & 0 & 0 &  & 45 & $\overline{6}'m'2$ & 0 & 0\tabularnewline
\cline{1-4} \cline{6-9} 
17 & $\frac{4'}{m'}$ & 0 & 1 &  & 46 & $\overline{6}'m2'$ & 0 & 1\tabularnewline
\cline{1-4} \cline{6-9} 
18 & $4'22'$ & 0 & 1 &  & 47 & $\overline{6}m'2'$ & 1 & 1\tabularnewline
\cline{1-4} \cline{6-9} 
19 & $42'2'$ & 1 & 1 &  & 48 & $\frac{6}{m'}mm$ & 0 & 1\tabularnewline
\cline{1-4} \cline{6-9} 
20 & $4'm'm$ & 0 & 1 &  & 49 & $\frac{6'}{m}mm'$ & 0 & 1\tabularnewline
\cline{1-4} \cline{6-9} 
21 & $4m'm'$ & 1 & 0 &  & 50 & $\frac{6'}{m'}mm'$ & 0 & 1\tabularnewline
\cline{1-4} \cline{6-9} 
22 & $\overline{4}'2'm$ & 0 & 1 &  & 51 & $\frac{6}{m}m'm'$ & 1 & 1\tabularnewline
\cline{1-4} \cline{6-9} 
23 & $\overline{4}'2m'$ & 0 & 0 &  & 52 & $\frac{6}{m'}m'm'$ & 0 & 0\tabularnewline
\cline{1-4} \cline{6-9} 
24 & $\overline{4}2'm'$ & 1 & 1 &  & 53 & $m'\overline{3}'$ & 0 & 0\tabularnewline
\cline{1-4} \cline{6-9} 
25 & $\frac{4}{m'}mm$ & 0 & 1 &  & 54 & $\overline{4}'32'$ & 0 & 1\tabularnewline
\cline{1-4} \cline{6-9} 
26 & $\frac{4'}{m}m'm$ & 0 & 1 &  & 55 & $\overline{4}'3m'$ & 0 & 0\tabularnewline
\cline{1-4} \cline{6-9} 
27 & $\frac{4'}{m'}m'm$ & 0 & 1 &  & 56 & $m'\overline{3}'m$ & 0 & 1\tabularnewline
\cline{1-4} \cline{6-9} 
28 & $\frac{4}{m}m'm'$ & 1 & 1 &  & 57 & $m\overline{3}m'$ & 0 & 1\tabularnewline
\cline{1-4} \cline{6-9} 
29 & $\frac{4}{m'}m'm'$ & 0 & 0 &  & 58 & $m'\overline{3}'m'$ & 0 & 0\tabularnewline
\cline{1-4} \cline{6-9} 
\end{tabular}
\end{table}

\bibliographystyle{apsrev4-1}
\bibliography{glide_spt_bib_01,bib_June2017,space_group_bib_01,tensor_network_bib_01,cSPT_E8}

\end{document}